\documentclass[a4paper,12pt]{article}
\usepackage{fullpage}
\usepackage{amsmath,amssymb,amsthm,amsfonts}
\usepackage{enumerate}
\usepackage{graphicx}
\usepackage[hang]{subfigure}
\usepackage{hyperref}
\usepackage{stmaryrd}
\usepackage{upgreek,pifont}
\usepackage{array}
\graphicspath{{./}{./Fig/}}
\usepackage{tikz}
\usetikzlibrary{decorations.pathreplacing}

\usepackage{pifont}
\usepackage{xspace}
\newcommand{\SM}[1][]{\ensuremath{\mathfrak{M}_{#1}}\xspace} 
\newcommand{\SigSet}{\ensuremath{\mathcal{M}}\xspace} 
\newcommand{\ColSet}{\ensuremath{\mathcal{R}}\xspace} 
\newcommand{\ColIn}[1][C]{\ensuremath{#1^{-}}\xspace} 
\newcommand{\ColOut}[1][C]{\ensuremath{#1^{+}}\xspace} 
\newcommand{\SpeedFun}{\ensuremath{\mathcal{S}}\xspace} 
\newcommand{\Speed}[1]{\ensuremath{\mathcal{S}(#1)}\xspace} 
\newcommand{\Accu}{\text{\ding{90}}\xspace} 
\newcommand{\Void}{\ensuremath{\oslash}\xspace} 
\newcommand{\Sig}[1][]{\ensuremath{\mu_{#1}}\xspace} 
\newcommand{\SigName}[1]{\ensuremath{\textsf{\sf{#1}}\xspace}}			
\newcommand{\SigNameL}[1]{\ensuremath{\overleftarrow{\SigName{#1}}\xspace}}	
\newcommand{\SigNameR}[1]{\ensuremath{\overrightarrow{\SigName{#1}}\xspace}}	
\newcommand{\Col}{\ensuremath{\ColIn\rightarrow \ColOut}\xspace} 
\newcommand{\ColName}[1][]{\ensuremath{\rho_{#1}}\xspace} 
\newcommand{\ColRule}[2]{\ensuremath{\{#1\}\rightarrow\{#2\}}\xspace} 
\newcommand{\SigAt}[2]{\ensuremath{#1}@{#2}\xspace} 
\newcommand{\Space}{\ensuremath{\mathbb{R}}\xspace} 
\newcommand{\Time}{\ensuremath{\mathbb{R^{+}}}\xspace} 
\newcommand{\SpaceTime}{\ensuremath{\Space\times\Time}\xspace} 
\newcommand{\STD}{\ensuremath{\mathbb{D}}\xspace} 
\newcommand{\STDVal}{\ensuremath{\mathcal{V}}\xspace} 
\newcommand{\STDinclu}{\ensuremath{\subseteq}\xspace} 
\newcommand{\Config}[1][]{\ensuremath{c_{#1}}\xspace} 
\newcommand{\TLimit}{\ensuremath{\tilde{t}}\xspace} 
\newcommand{\MaxSpeed}{\ensuremath{\nu_{max}}\xspace} 
\newcommand{\MinSpeed}{\ensuremath{\nu_{min}}\xspace} 
\newcommand{\CausalPast}{\ensuremath{\Gamma^{-}}\xspace} 

\newcommand{\SigWallZero}{\SigName{w$_0$}\xspace}
\newcommand{\SigWallx}{\SigName{w$_x$}\xspace}
\newcommand{\SMsubtraction}{\ensuremath{\SM[sub]}\xspace}	
\newcommand{\SigSubInit}{\SigNameR{init}\xspace}		
\newcommand{\SigSubZigSmall}{\SigNameR{zig}\xspace}		
\newcommand{\SigSubZagSmall}{\SigNameL{zag}\xspace}		
\newcommand{\SigSubZigBig}{\SigNameR{ZIG}\xspace}		
\newcommand{\SigSubZagBig}{\SigNameL{ZAG}\xspace}		
\newcommand{\SigSubWallZero}{\SigName{w$_0$}\xspace}		
\newcommand{\SigSubWalla}{\SigName{w$_a$}\xspace}		
\newcommand{\SigSubWallb}{\SigName{w$_b$}\xspace}		
\newcommand{\SigSubWallr}{\SigName{w$_r$}\xspace}		
\newcommand{\SMmodulo}{\ensuremath{\SM[mod]}\xspace}		
\newcommand{\SigModInit}{\SigNameR{init}\xspace}		
\newcommand{\SigModZigSmall}{\SigNameR{zig}\xspace}		
\newcommand{\SigModZagSmall}{\SigNameL{zag}\xspace}		
\newcommand{\SigModZigBig}{\SigNameR{ZIG}\xspace}		
\newcommand{\SigModZagBig}{\SigNameL{ZAG}\xspace}		
\newcommand{\SigModWallZero}{\SigName{w$_0$}\xspace}		
\newcommand{\SigModWalla}{\SigName{w$_a$}\xspace}		
\newcommand{\SigModWallb}{\SigName{w$_b$}\xspace}		
\newcommand{\SigModWallr}{\SigName{w$_r$}\xspace}		

\newcommand{\SigExDeuxGreen}{\SigName{s$_1$}\xspace}
\newcommand{\SigExDeuxBlack}{\SigName{s$_2$}\xspace}
\newcommand{\SigExDeuxBlue}{\SigName{s$_3$}\xspace}
\newcommand{\SigExDeuxRed}{\SigName{s$_4$}\xspace}

\newcommand{\SigLeft}{\SigNameR{left}\xspace}
\newcommand{\SigRight}{\SigNameL{right}\xspace}
\newcommand{\SigZig}{\SigNameR{zig}\xspace}
\newcommand{\SigZag}{\SigNameL{zag}\xspace}

\newcommand{\Supp}[1]{\ensuremath{\widehat{#1}}\xspace}
\newcommand{\SigZero}{\SigName{S}\xspace}
\newcommand{\SigOne}{\SigNameR{R}\xspace}

\newcommand{\SMeuclidPhi}{\ensuremath{\SM[gcd]^{\varphi}}\xspace}
\newcommand{\SMeuclid}{\ensuremath{\SM[gcd]}\xspace}	
\newcommand{\SigEuclidZigSmall}{\SigNameR{zig}\xspace}	
\newcommand{\SigEuclidZagSmall}{\SigNameL{zag}\xspace}	
\newcommand{\SigEuclidZigBig}{\SigNameR{ZIG}\xspace}	
\newcommand{\SigEuclidZagBig}{\SigNameL{ZAG}\xspace}	
\newcommand{\SigEuclidWallZero}{\SigName{w$_0$}\xspace}	
\newcommand{\SigEuclidWalla}{\SigName{w$_a$}\xspace}	
\newcommand{\SigEuclidWallb}{\SigName{w$_b$}\xspace}	
\newcommand{\SigEuclidStart}{\SigNameR{start}\xspace}	

\newcommand{\SigS}{\SigName{S}\xspace}			
\newcommand{\SigR}{\SigNameR{R}\xspace}			
\newcommand{\SigL}{\SigNameL{L}\xspace}			
\newcommand{\Posi}{\ensuremath{x_0}\xspace} 		
\newcommand{\SMpq}{\ensuremath{\SM[3]^{p/q}}\xspace}	
\newcommand{\SMnu}{\ensuremath{\SM[3]^{\nu}}\xspace}	
\newcommand{\StripName}[1][]{strip#1\xspace}		
\newcommand{\StripMultiName}[1][]{mesh#1\xspace}	
\newcommand{\Strip}[4]{\ensuremath{(#1,#2,#3,#4)}--\ensuremath{\mathcal{S}trip}\xspace}
\newcommand{\StripMulti}[5]{\ensuremath{(#1,#2,#3,#4,#5)}--\ensuremath{\mathcal{M}esh}\xspace}

\newlength\Width
\newlength\ArrowSpacing
\newlength\Height
\newlength\Sep
\tikzstyle{arrow}=[thin] 

\newcommand{\Latin}[1]{\textsl{#1}\xspace}
\newcommand{\ie}{\Latin{i.e.}}
\newcommand{\eg}{\Latin{e.g.}}

\newcommand{\N}{\ensuremath{\mathbb{N}}\xspace} 
\newcommand{\Z}{\ensuremath{\mathbb{Z}}\xspace} 
\newcommand{\Q}{\ensuremath{\mathbb{Q}}\xspace} 
\newcommand{\R}{\ensuremath{\mathbb{R}}\xspace} 
\newcommand{\PowerSet}[1]{\ensuremath{\mathcal{P}(#1)}\xspace} 
\newcommand{\Card}[1]{\ensuremath{|#1|}} 
\newcommand{\Injection}{injective\xspace} 
\newcommand{\Image}[1]{\ensuremath{\mathcal{I}m(#1)}\xspace} 
\newcommand{\Support}[1]{\ensuremath{support(#1)}\xspace} 
\def\EqClass#1{\expandafter\ensuremath{[#1]_{\sim}}\xspace}
\newcommand{\SetQuotient}[2]{\ensuremath{#1/\!\raisebox{-.65ex}{\ensuremath{#2}}}}

\newcommand{\Eq}[1]{Eq.~(\ref{#1})}

\newcommand{\Sec}[1]{Sec.~\ref{#1}}
\newcommand{\Section}[1]{Section~\ref{#1}}

\newcommand{\Subsec}[1]{Subsect.~\ref{#1}}
\newcommand{\Def}[1]{Def.~\ref{#1}}

\newcommand{\Cor}[1]{Coro.~\ref{#1}}

\newcommand{\Lem}[1]{Lem.~\ref{#1}}
\newcommand{\Lemma}[1]{Lemma~\ref{#1}}
\newcommand{\Fig}[1]{Fig.~\ref{#1}}
\newcommand{\Figure}[1]{Figure~\ref{#1}}

\newenvironment{remark}
  {\smallskip\noindent{\bf Remark.~}}
  {\smallskip}

\newtheorem{theorem}{Theorem}
\newtheorem{definition}{Definition}
\newtheorem{corollary}{Corollary}

\newtheorem{proposition}{Proposition}
\newtheorem{lemma}{Lemma}

\usepackage{titlesec}
\titleformat{\paragraph}[runin]{}{}{}{\bfseries}[.]%
\begin{document}
%

\title{Abstract Geometrical Computation 8:\\
  Small Machines, Accumulations \& Rationality\thanks{This work was partially funded by the ANR project AGAPE, ANR-09-BLAN-0159-03.}
}

\author{Florent Becker$^{1}$
  \and Mathieu Chapelle$^{1}$
  \and J\'er\^ome Durand-Lose$^{1}$
  \thanks{corresponding author: \texttt{jerome.durand-lose@univ-orleans.fr}}%
  \and Vincent Levorato$^{1,2}$
  \and Maxime Senot$^{1}$
}

\date{}

\maketitle

$^{1}$ Univ. Orl\'eans, ENSI de Bourges, LIFO \'EA 4022, F-45067 Orl\'eans, France 


$^{2}$ CESI-IRISE, 959 rue de la Bergeresse, F-45160 Olivet, France

\bigskip

\begin{abstract}
In the context of abstract geometrical computation, computing with colored line segments, we study the possibility of having an accumulation with \emph{small} signal machines, \ie, signal machines having only a very limited number of distinct speeds.
The cases of $2$ and $4$ speeds are trivial: we provide a proof that no machine can produce an accumulation in the case of $2$ speeds and exhibit an accumulation with $4$ speeds.

The main result is the twofold case of $3$ speeds.
On the one hand, we prove that accumulations cannot happen when all ratios between speeds and all ratios between initial distances are rational.
On the other hand, we provide examples of an accumulation in the case of an irrational ratio between two speeds and in the case of an irrational ratio between two distances in the initial configuration.

This dichotomy is explained by the presence of a phenomenon computing Euclid's algorithm ($\gcd$): it stops if and only if its input is commensurate (\ie, of rational ratio).
\end{abstract}
\renewcommand{\abstractname}{Keywords}
\begin{abstract}
  Abstract Geometrical Computation;
  Signal Machine;
  Accumulation;
  Unconventional Computing;
  Euclid's Algorithm.
\end{abstract}

%
%
%
\section{Introduction}
\label{sec:introduction}
Imagine yourself with some color pencils and a sheet of paper together with ruler and compass.
There is something drawn in a corner of the paper and you are given rules so as to extend the drawing.
According to the rules and the initial drawing/configuration you might stop soon, have to extend the paper indefinitely or draw forever in a bounded part of the paper as on top of \Fig{fig:most-basic-accumulation}.
Already on this simple setting emerges the usual questions related to dynamical systems: liveness, unbounded orbit or convergence/accumulation.

In this article we concentrate on accumulation in the case where the dynamical system is a signal machine in the context of abstract geometrical computation.
In this setting, one drawing direction is distinguished and used as time axis.
Line segments are enlarged synchronously until they intersect another one.
This goes on until no more collision can happen.

\begin{figure}[hbt]
  \centering
  \subfigure[Most basic accumulation.\label{fig:most-basic-accumulation}]{%
    \includegraphics[width=0.3\textwidth]{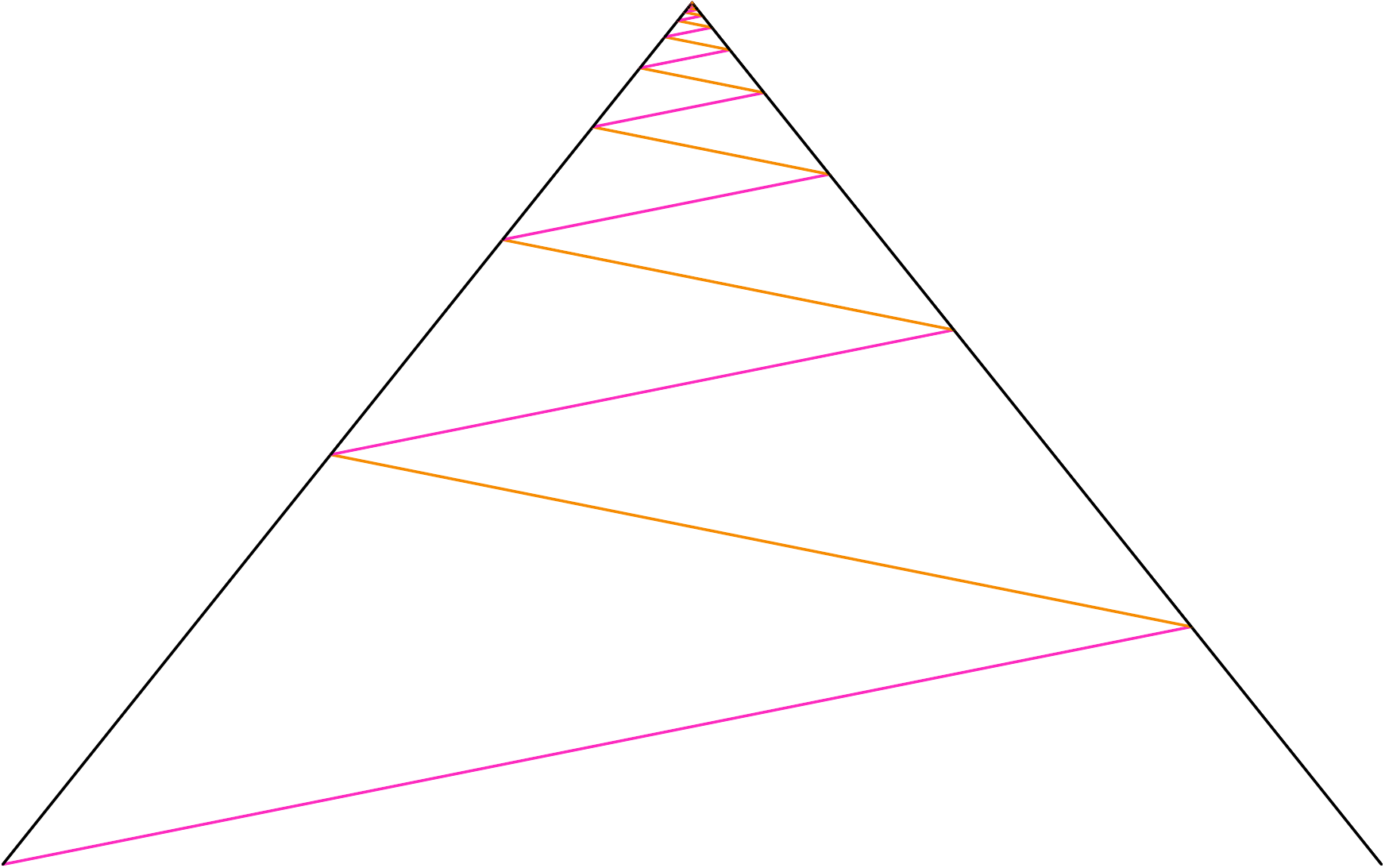}}
  \hfill
  \subfigure[$1$ speed: no collision.\label{fig:1-speed}]{%
    \quad
    \begin{tikzpicture}[x=1.2em,y=.9em]
      \foreach \x/\c in {1/blue,1.3/dotted,2/black,3/blue,3.8/dashed,4.5/red} {
        \draw[\c] (\x,0) +(-3,0)-- (\x,7) ;
      }
    \end{tikzpicture}
    \quad
  }
  \hfill
  \subfigure[$2$ speeds: no accumulation.\label{fig:2-speed}]{%
    \begin{tikzpicture}[x=1.2em,y=.9em]
      \foreach \x/\c in {1/blue,1.3/dotted,2/black,3/blue,3.8/dashed,4.5/red} {
        \draw[\c] (\x,0) +(-3.5,0)-- (\x,7) ;
        \draw[\c] (-\x,0) +(6,0)-- (-\x,7) ;
      }
    \end{tikzpicture}
  }
  \caption{Basic cases.}
\end{figure}
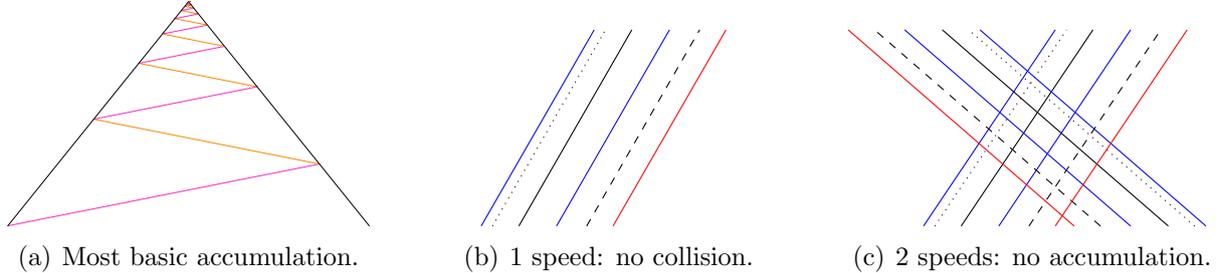

The line segments are the traces of \emph{signals} and their intersections are \emph{collisions} of signals.
Each signal corresponds to some \emph{meta-signal}.
In-coming signals are removed and new ones are emitted according to the meta-signals associated with the incoming signals.
This is called a \emph{collision rule}.
Signals that correspond to the same meta-signal must travel at the same speed, the resulting traces are parallel.
There are finitely many meta-signals so there are finitely many collision rules.

The signals move on a one dimensional Euclidean space orthogonal to the temporal axis (in the figures, space is horizontal and time elapses upwards).
Considering the traces leads to two dimensional drawings called  \emph{space-time diagrams} (as illustrated throughout the article).
Please let us point out that space and time are continuous spaces (the real line) and that signals as well as collisions are dimensionless points.
Computations are exact, there is no noise nor approximation.

Signal machines are very powerful and colorful complex systems. 
Accumulation is easy to achieve and is in fact the cornerstone to hypercomputation in the model \cite{durand-lose09nc}.
In the present article, we investigate the minimum size of a machine so that an accumulation is possible or not.
Already with four meta-signals of different speeds (or directions on the drawing) an accumulation can happen as depicted on \Fig{fig:most-basic-accumulation}.

In fact, we will see in the sequel of this article that the number of meta-signals is not relevant; the relevant measure here is the number of different speeds and in the case of three speeds, their values and the initial positions as explained below.
One speed does not even allow any collision (see \Fig{fig:1-speed}).
With two speeds, the number of collisions is finite and signals have to follow a regular grid which has no accumulation (see \Fig{fig:2-speed}).

When three signals of different speeds are present, in order to make accumulations more likely to occur, we can imagine that each collision generates all signals.
Then, in the generated space-time diagrams, we can exhibit the emulation of Euclid's algorithm to compute some greatest common divisor ($\gcd$).
If the ratios involved are irrational, then the algorithm does not stop and brings forth an accumulation.
On the contrary, if every ratio that could be involved in a $\gcd$ computation is rational, then a global $\gcd$ is generated and some global regular \emph{mesh} emerges.
Whatever the number of meta-signals and whatever the collision rules are, there is no way for the signals to escape this mesh.
The signals have to be on the mesh and the mesh does not have any accumulation point. Hence, the diagram cannot have an accumulation.

\paragraph{State of the art}

Signal machines are one of the (unconventional) models of computation dealing with Euclidean geometry.
To name a few, one can think of: Euclidean abstract machines \cite{dacosta06,huckenbeck89tcs}, piecewise constant derivatives systems \cite{bournez97icalp}, colored universes \cite{jacopini+sontacchi90}\dots

Signal machines were originally introduced as a continuous counterpart of cellular automaton to provide a context for the underlying Euclidean reasoning often found in the discrete cellular automata literature as well as propose an abstract formalization of the concept of \emph{signal} \cite{durand-lose08jac,mazoyer96,mazoyer+terrier99}.

Signal machines are able to compute in the classical Turing understanding. 
This paper is somehow a companion to \cite{durand-lose11tcs} where a Turing-universal signal machine is presented with only $13$ meta-signals and $4$ speeds.
This research takes place in the quest for minimality in order to compute, and thus get unpredictable behavior.
Let us cite \cite{rogozhin96,neary+woods07mcu} for Turing machines, \cite{cook04,ollinger+richard11tcs} for cellular automata, and \cite{margenstern00tcs} for a more general picture.

Being in a continuous setting, signal machines can carry out analog computations in the sense of both the BSS's model \cite{blum+shub+smale89,durand-lose08cie} and Computable analysis \cite{weihrauch00,durand-lose09uc}.

Massive parallelism and the capability to approximate a fractal (with potentially infinitely many accumulations) allows to provide efficient solutions to hard problems: SAT for the class NP \cite{duchier+durand-lose+senot10isaac} and Q-SAT for PSPACE \cite{duchier+durand-lose+senot12tamc}.
From previous work on signal machine, accumulations are known to be easy to generate. They are a powerful tool to accelerate a computation and provide hypercomputation \cite{durand-lose09nc}.
Recently, it has been proved that, starting from a rational machine and configuration, the locations of isolated accumulations have to be \emph{d-c.e} (difference of computably enumerable) real numbers and that any such number can be reached \cite{durand-lose12nc-uc}.

\paragraph{Outline}
Formal definitions of signal machines, their dynamics, space-time diagrams, properties and normalizations are given in \Sec{sec:definitions}, as well as an introductory example: the computation of the remainder of Euclidean division (later embedded inside the $\gcd$ computation).
In \Sec{sec:4-2speeds}, the case of two and four speeds are settled.
\Section{sec:3speeds} studies the case of three speed systems, where the rationality of ratios enters into play.
With an irrational ratio in distances (resp. in speeds), an accumulation is susceptible to happen, this is proven by following the steps of Euclid's algorithm into the space-time diagram and proving that the accumulation occurs in finite time.
Otherwise, if all ratios are rational, then we prove that all generated signals must be on a regular mesh, this mesh has no accumulation, so that the initial space-time diagram cannot have any.
\Section{sec:conclusion} concludes this paper.


%
%
%
\section{Signal machines}
\label{sec:definitions}
This section introduces basic definitions about signal machines and their space-time diagrams,
with several formulations: in terms of {\em machines}, {\em topology} and {\em dynamics}.
\Subsec{subsec:examples} gives two toy examples ---geometrical computation of a substraction and a {\em modulo}--- 
that will be used to implement Euclid's algorithm with signals.
Some notions and geometrical properties such as affine transformations of machines and inclusions of diagrams are given in \Subsec{subsec:properties}.

\subsection{Definitions}
\label{subsec:definitions}

We formalize here the definition of a signal machine, which corresponds to a set of meta-signals, 
a speed function assigning a real speed to each meta-signal and a function describing the result of a meta-signals collision:

\begin{definition}[Signal machine]\label{def:machine}
  A {\em signal machine \SM} is a triplet $\SM=(\SigSet, \SpeedFun, \ColSet)$ where:
  \begin{enumerate}[(i)]
    \item $\SigSet$ is a finite set of meta-signals;
    \item $\SpeedFun:\SigSet\to\R$ is the speed function which assigns a real speed value $\SpeedFun(\mu)$ to each meta-signal $\mu$;
    \item $\ColSet:\PowerSet{\SigSet}\to\PowerSet{\SigSet}$ is the collision function:
	  each set of meta-signals $\ColIn\in\PowerSet{\SigSet}$ such that $\Card{\ColIn}\geq2$ and $\SpeedFun_{\restriction\ColIn}$ is $\Injection$, is mapped to a set of meta-signals \ColOut so that $\SpeedFun_{\restriction\ColOut}$ is $\Injection$.\label{def:item:collisions}
  \end{enumerate}
\end{definition}

\Figure{fig:sm-example} provides an example of a very simple signal machine, and an evolution of this machine, that we call a {\em space-time diagram}.
The meta-signals are listed in \Fig{fig:example_sm-signals}, and collision rules are given by \Fig{fig:example_sm-rules}.
\Figure{fig:example_sm-std} provides an example of a space-time diagram for this machine.

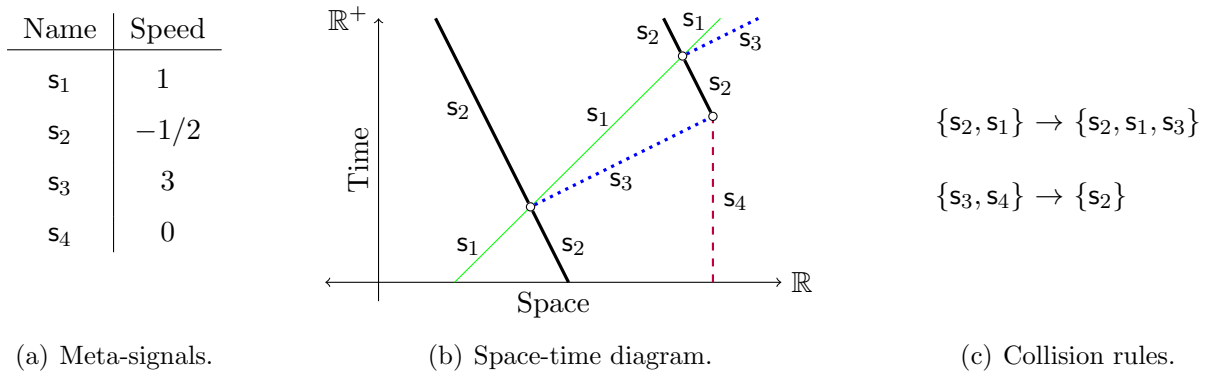
\begin{figure}[hbt]
  \centering\small
  \subfigure[Meta-signals.\label{fig:example_sm-signals}]{%
    \begin{tabular}[b]{c|c}
      Name & Speed\\\hline
      \SigExDeuxGreen & $1$ \rule{0pt}{5mm}\\[2mm]
      \SigExDeuxBlack & $-1/2$\\[2mm]
      \SigExDeuxBlue & $3$\\[2mm]
      \SigExDeuxRed & $0$\\\noalign{\vskip10mm}
    \end{tabular}}\hfill
  \subfigure[Space-time diagram.\label{fig:example_sm-std}]{%
    \begin{tabular}[b]{c}%
    \begin{tikzpicture}
      \draw[<->] (-1.7,0)--node[below] {Space} (4.3,0) node[right,pos=1]{\Space};
      \draw[->] (-1,-.25)--node[above,sloped] {Time} (-1,3.5) node[left,pos=1]{\Time} ;
      \begin{scope}[nodes={inner sep=.1em,minimum size=.1em,draw,circle}]
        \draw (1,1) node (A) {};
        \draw (3.4,2.2) node (B) {};
        \draw (3,3) node (C) {};        
      \end{scope}
      \draw[green] (0,0)--node[black,left]{\SigExDeuxGreen} (A)--node[black,left,pos=.6]{\SigExDeuxGreen} (C)--node[black,left,pos=.9]{\SigExDeuxGreen} (3.5,3.5);
      \draw[thick,dashed,purple] (3.4,0)--node[black,right]{\SigExDeuxRed} (B);
      \begin{scope}[very thick]
        \draw (1.5,0)--node[black,right]{\SigExDeuxBlack} (A)--node[black,left]{\SigExDeuxBlack}(-.25,3.5);
        \draw (B)-- node[black,right]{\SigExDeuxBlack}(C)--node[black,left]{\SigExDeuxBlack}(2.75,3.5);
      \end{scope}
      \begin{scope}[very thick,dotted,blue]
        \draw (A)--node[black,below]{\SigExDeuxBlue}(B);
        \draw (C)--node[black,below,pos=.9]{\SigExDeuxBlue} (4,3.5);
      \end{scope}
    \end{tikzpicture}
   \end{tabular}}\hfill
  \subfigure[Collision rules.\label{fig:example_sm-rules}]{%
    \begin{tabular}[b]{r@{ $\rightarrow$ }l}
      $\{\SigExDeuxBlack, \SigExDeuxGreen\}$ & $\{\SigExDeuxBlack, \SigExDeuxGreen, \SigExDeuxBlue\}$\\[5mm]
      $\{\SigExDeuxBlue, \SigExDeuxRed\}$    & $\{\SigExDeuxBlack\}$\\\noalign{\vskip15mm}
      \end{tabular}}
  \caption{Example of a signal machine and one of its possible evolution.}\label{fig:sm-example}
\end{figure}

Condition ({\it \ref{def:item:collisions}}) of \ColSet definition means that signals can collide only if they 
have distincts speeds, and a collision involves at least two signals. The signals resulting from a collision must also have distinct speeds. 
Since we interpret $\ColSet(\ColIn)=\ColOut$ as a rule, we rather note $\ColIn\rightarrow\ColOut$ instead of $\ColSet(\ColIn)=\ColOut\enspace$.

We call {\em $n$-speed machine} any signal machine having exactly $n$ distinct values for its meta-signals speeds, \ie, \SM is a $n$-speed machine if $\Card{\Image{\SpeedFun}} = n\enspace$.

\paragraph{Configurations}

A {\em configuration} is a function from the real line (space) into the set of meta-signals and collision rules plus two extra values: \Void (for nothing there) and \Accu (for accumulation).
We note \STDVal the set of values, \ie, $\STDVal = \SigSet\cup\ColSet\cup\{\Void\}\cup\{\Accu\}$.
A configuration can be seen as the ``global state'' of the signal machine at a given time, and describes the presence and the disposition of signals and collisions in the space \R.

Any signal or collision must be spatially isolated: there is nothing else but \Void arbitrarily closed.
The accumulation points of non-\Void locations must be \Accu. These are spatial {\em static} accumulations.

\begin{definition}[Configuration]\label{def:configuration}
  A {\em configuration}, \Config, is a function from the real line into meta-signals, 
  rules, \Void and  \Accu (let $\STDVal=\SigSet\cup\ColSet\cup\{\Void,\Accu\}$ 
  so that $\Config:\R\rightarrow\STDVal$) such that:
  \begin{enumerate}[(i)]
    \item all signals and collisions are isolated: \\
	  $\forall x\in\Space,\ \Config(x)\in\SigSet\cup\ColSet\ \Rightarrow \exists \varepsilon > 0,\  
	  \forall y,\ 0<|x-y|<\varepsilon \Rightarrow \Config(y)=\Void\enspace$;
    \item spatial accumulation are marked accordingly: any $x$ that is an accumulation point
	  of $\Config^{-1}(\STDVal\setminus\{\Void\})$ verifies $\Config(x)=\Accu$ ($x$ in \R is
	  an accumulation point of a subset E of \R iff 
	  $\forall\varepsilon$, $1<|E\cap(x-\varepsilon,x+\varepsilon)|$).
  \end{enumerate}
\end{definition}

If there is a signal of speed $s$ at $x$, then, unless there is a collision before, after a duration $\Delta t$, its position is $x+s{\cdot} \Delta t$.
At a collision, all incoming signals are immediately replaced by outgoing signals in the following configurations according to collision rules.

For the next definition, the {\em support} of a configuration \Config will denote the set of non-\Void-valued positions, \ie, 
$\Support{\Config} = \{ x\in\Space~|~\Config[0](x) \neq \Void \}$.

\begin{definition}[Initial configuration]\label{def:initial-configuration}
  An {\em initial configuration} is a configuration $\Config[0]:\Space\rightarrow\STDVal$ so that: 
  \begin{enumerate}[(i)]
    \item the support of \Config[0] is finite, \ie, $\{ x\in\Space~|~\Config[0](x) \neq \Void \}$ is finite;
    \item for all $x\in\Space$, $\Config[0](x) \neq \Accu\enspace$.
  \end{enumerate}
\end{definition}

Accumulations (the \Accu values) are forbidden in the initial configuration.
The case of an initial collision is interpreted by the possibility of having several signals with distinct speeds at the same initial position: 
for an initial collision $\{\Sig[1],\ldots,\Sig[p]\}\rightarrow\{\Sig[1]',\ldots,\Sig[q]'\}$ occuring at position $x$, 
we rather note $[\Sig[1]', \ldots,\Sig[q]']@x$ or $\Sig[1]'@x, \ldots,\Sig[q]'@x$.
In this way, an initial configuration can be expressed only in terms of signals at some positions (the other positions taking the value \Void).
So we can give an initial configuration \Config[0] by a finite set of the form $\Config[0] = \{ \Sig[1]@x_1, \ldots, \Sig[k]@x_k\}\enspace$
where $\Sig[i]@x_i$ means $\Sig[i]$ is initially located at the spatial position $x_i$ \ie $\Config[0](x_i) = \Sig[i]$.

We insist that $[\Sig[1]', \ldots,\Sig[q]']@x = \Sig[1]'@x, \ldots,\Sig[q]'@x$ is just a notation to simplify the writing of collisions when we express a configuration as a set of non-\Void values.
But by definition, a configuration \Config is a {\em function} from \R in \STDVal: for each $x\in\R$, the value $\Config(x)$ is {\em unique} (and is either a meta-signal, a collision, \Accu or \Void).
We will also use this notation of collision in any configuration (and not only for an initial one), and $\Config[t] = \{\Sig[1]'@x, \ldots,\Sig[q]'@x \}$ has to be interpreted by ``any collision producing $\{\Sig[1], \ldots,\Sig[q]'\}$ and occuring at $x$'', \ie, $\Config[t](x) = C$ with $\ColOut = \{\Sig[1]', \ldots,\Sig[q]' \}$.
We allow initial configurations (and only initial configurations) to contain some signals of distinct speeds at the same initial position, 
even if this set of signals doesn't correspond to a collision outcoming set of signals.

\begin{definition}[Rational signal machine]\label{def:rational-machine}
  A signal machine is {\em rational} if all speeds are rational numbers and non-\Void positions in the initial configuration are also rational numbers.
\end{definition}

Since the position of collisions are solutions of rational linear equations systems, they are rational.
In any space-time diagram of a rational signal machine, as long as there is no accumulation, 
the coordinates of all the collisions are rational.

\begin{definition}[Rational-like]\label{def:rational-like}
  A signal machine is {\em rational-like} if all its speeds are two-by-two commensurate, \ie, all ratios between speeds are rationnal.
  A configuration is {\em rational-like} if all distances are two-by-two commensurate.
\end{definition}

This means that a signal machine is rational-like if all its speeds are rational up to a multiplicative coefficient.
In particular, any rational machine is rational-like.

\paragraph{Space-time diagrams}

A space-time diagram can be formulated in a topological way, based on the classical topology of $\R^2$, whose open sets are generated by the Euclidean distance. 
The topological formulation of diagram implies that space and time are considered like a whole object ---the space-time structure--- and there is \emph{a priori} no dynamics. 

We define first a notion that will be used to define accumulation: the notion of {\em causal past} of a point.

\begin{definition}[Causal past and isolated accumulation]\label{def:isolated-accumulation}
  Let \MaxSpeed (maximal right speed) and \MinSpeed (maximal left speed) be the maximum and minimum values taken by the speed function \SpeedFun.
  The value at position $(x,t)$ in the space-time diagram only depends on the values at the position on the \emph{causal past} or \emph{backward light cone}:
  \[\CausalPast(x,t)
    = \left\{\rule{0cm}{1.1em}
    \, (x',t')\, 
    \middle|\, 
    t'<t
    \wedge \MaxSpeed{\cdot}(t'{-}t)< x'-x< \MinSpeed {\cdot}(t'{-}t)
    \,\right\}\enspace.
  \]
\end{definition}

This notion is illustrated in \Fig{fig:cone}.

We can now give the formalization of {\em space-time diagrams}:

\renewcommand{\labelitemi}{$\bullet$}
\begin{definition}[Space-time diagram]\label{def:diagram}
  A {\em space-time diagram} is a map \STD from a time interval $[0, T] \subset \Time$ ($T$ can be infinite)
  to the set of configurations (\ie \STD can be identified to a map $\SpaceTime\to\STDVal$) such that:
  \begin{enumerate}[(i)]
    \item $\forall t\in[0, T] \ \{ x\in\Space\ |\ \Config[t](x) \neq \Void \}$ is finite;
    \item if $\Config[t](x) = \Sig \in \SigSet$ then $\exists t_i, t_f \in [0,T]$ with $t_i<t<t_f$ or 
	  $0=t_i=t<t_f$ or $t_i<t=t_f=T$ such that:
	  \begin{itemize}
	    \item $\forall t'\in~]t_i, t_f[ \ \Config[t'](x + \Speed{\Sig}(t-t')) = \Sig$,
	    \item $t_i=0$ or $\Config[t_i](x_i) \in \ColSet$ and $\Sig\in(\Config[t_i](x_i))^{+}$ where $x_i = x + \Speed{\Sig}(t_i - t)$,
	    \item $t_f=T$ or $\Config[t_f](x_f) \in \ColSet$ and $\Sig\in(\Config[t_f](x_f))^{-}$ where $x_f = x + \Speed{\Sig}(t_f - t)$;
	  \end{itemize}
    \item if $\Config[t](x) = \Col \in \ColSet$ then $\exists\ \varepsilon > 0\ \forall t'\in[t-\varepsilon, t+\varepsilon]\ 
	  \forall x'\in[x-\varepsilon, x+\varepsilon]$
	  \begin{itemize}
	    \item $\Config[t'](x')\in\ColIn\cup\ColOut\cup\{\Void\}$,
	    \item $\forall \Sig \in \ColIn\ :(\ \Config[t'](x') = \Sig\ ) \Leftrightarrow (\ t<t' \text{ and }x' = x + \Speed{\Sig}(t'-t)\ )$,
	    \item $\forall \Sig \in \ColOut\ :(\ \Config[t'](x') = \Sig\ ) \Leftrightarrow (\ t'<t \text{ and }x' = x + \Speed{\Sig}(t'-t)\ )$;
	  \end{itemize}
    \item if $\Config[t](x) = \Accu$ then
	  \begin{itemize}
	    \item $\exists\ \varepsilon > 0\ \forall (x',t') \notin \CausalPast(x,t)$ s.t. 
		  $|x-x'| < \varepsilon,\ |t-t'|<\varepsilon$ we have $\Config[t'](x') = \Void$,
	    \item $\forall\ \varepsilon > 0,\ |\{\ (x',t')\in\CausalPast(x,t)\ |\ t-\varepsilon<t'<t \wedge\Config[t'](x')\in\ColSet\ \}| = \infty\enspace$.
	  \end{itemize}
  \end{enumerate}
\end{definition}

As done for configurations, we can define the {\em support of a diagram \STD} by 
$\Support{\STD} = \{ (x,t)\in\SpaceTime ~|~ \STD(x,t)\neq\Void \}$.

\paragraph{Equivalent diagrams}

To formalize the idea that two geometrical computations are ``the same'', we define the notion of {\em equivalent diagrams}.
Intuitively, two diagrams are equivalent if they have the same structure, \ie, the same causality links between their respective collisions and signals, independently of their positions or the meta-signal names.
We formalize such an invariancy of structure with the notion of homeomorphism: a diagram obtained by a continuous transformation will keep the same structure.

\begin{definition}[Equivalent diagrams]\label{def:diagram-equivalency}
  Let \STD and $\STD'$ two space-time diagrams with respective value sets
  $\STDVal = \SigSet\cup\ColSet\cup\{\Void,\Accu\}$ and $\STDVal' = \SigSet'\cup\ColSet'\cup\{\Void,\Accu\}$.
  We say that \STD and $\STD'$ are {\em equivalent} if:
  \begin{enumerate}[(i)]
   \item there is an homeomorphism $h:\SpaceTime\to\SpaceTime$, and
   \item there is a bijection $\Phi:\STDVal\to\STDVal'$ so that $\Phi[\SigSet]=\SigSet'$, $\Phi[\ColSet]=\ColSet'$, $\Phi(\Accu)=\Accu$ and $\Phi(\Void)=\Void$ and so that 
         $\forall\Sig\in\SigSet, \forall\ColName\in\ColSet,\ \, \Sig\in\ColIn\ (resp. \ColOut) \Rightarrow \Phi(\Sig)\in\Phi(\ColName)^-\ (resp. \Phi(\ColName)^+)$, \label{def:diagram-equivalency_bijection}
  \end{enumerate}
  satisfying $\forall (x,t)\in\SpaceTime,\ \STD'(x,y)=\Phi(\STD(h(x,y)))$.
\end{definition}

The second condition means that $\Phi$ induces a bijection from the sets of meta-signals of \STD into the set of meta-signals of $\STD'$
and from the set of rules of \STD to the set of rules of $\STD'$ so that the place of meta-signals into the rules is kept.

\paragraph{Dynamics}

We give a presentation of a machine evolution in terms of dynamics.
This definition has been introduced by \cite{durand-lose12nc-uc}, in which it has been used to characterize the exact coordinates of isolated accumulations.

\begin{definition}[Dynamics]\label{def:dynamics}
  Considering a configuration \Config, the \emph{time to the next collision}, 
  $\Delta(\Config)$, is equal to the minimum of the positive real numbers $d$ such that:
  \[
  \exists x_1,x_2 \in\R, 
  \exists \mu_1, \mu_2\in\SigSet
  \left\{
    \begin{array}{l@{}}
      x_1+d{\cdot}\SpeedFun(\mu_1)=  x_2+d{\cdot}\SpeedFun(\mu_2)\\
      \Config(x_1) = \mu_1 \vee ( \Config(x_1) = \ColIn \rightarrow \ColOut \wedge \mu_1\in \ColOut )\\
      \Config(x_2) = \mu_2 \vee ( \Config(x_2) = \ColIn \rightarrow \ColOut \wedge \mu_2\in \ColOut )
    \end{array}
  \right.
  \enspace .
  \]
  It is $+\infty$ if there is no such $d$.
  
  Let $\Config[t]$ be the configuration at time $t$; for $t'$ between $t$ and $t+\Delta(\Config[t])$, the configuration at $t'$ is defined as follows.
  First, signals are set according to $\Config[t'](x')=\mu$ iff $\Config[t](x) = \mu \vee ( \Config[t](x) = \ColIn\rightarrow\ColOut \wedge \mu\in\ColOut )$ 
  where  $x= x'+(t{-}t'){\cdot}\SpeedFun(\mu)$. There is no collision to set ($t'$ is before the next collision).
  Then (static) accumulations are set:  $\Config[t'](x')=\Accu$ iff $x'$ is an accumulation point 
  of $\Config[t']^{-1}(\SigSet)$. It is \Void everywhere else.

  For the configuration at $t'=t+\Delta(\Config[t])$, collisions are set first:  
  $\Config[t'](x')=\ColIn\rightarrow\ColOut$ iff for all 
  $\mu\in\ColIn$, $\Config[t](x_\mu) = \mu \vee ( \Config[t](x_\mu) = \ColIn\rightarrow\ColOut \wedge \mu\in\ColOut )$ 
  where  $x_\mu= x'+(t{-}t'){\cdot}\SpeedFun(\mu)$.
  Then meta-signals are set (where there is not already a collision), and finally (static) accumulations.

  The {\em sequence of collision time} is defined by: $t_0=0$, $t_{n+1}=t_n+\Delta(\Config[t_n])$.
  This sequence is finite if there is an $n$ such that $\Delta(\Config[t_n])=+\infty$.
  Otherwise, since it is non-decreasing, it admits a limit.
  If the sequence is finite or its limit is infinite, then the whole space-time diagram is defined.
  These cases are of no interest here since there are no non-static accumulations.

  Only the last case is considered from now on: there is a finite limit, say $\TLimit$.
  The configuration at $\TLimit$ is defined as follows.
  First (dynamic) accumulations are set:  $\Config[\TLimit](x)=\Accu$ iff $\forall\varepsilon>0\ $
  then there exists $x'$ and $t'$ such that $|x-x'|<\varepsilon$, $\TLimit-\varepsilon<t'<\TLimit$ and $\Config[t'](x')\in\ColSet$.
  Then collisions are set: $\Config[\TLimit](x)=\ColIn\rightarrow\ColOut$ iff for all 
  $\mu\in\ColIn$, $\exists \varepsilon$, $\forall\varepsilon'$, $0{<}\varepsilon'{<}\varepsilon$, holds 
  $\Config[\TLimit-\varepsilon'](x'-\varepsilon'{\cdot}\SpeedFun(\mu)) = \mu$.
  Then meta-signals are set:  $\Config[\TLimit](x)=\mu$ iff 
  $\exists \varepsilon$, $\forall\varepsilon'$, $0{<}\varepsilon'{<}\varepsilon$, then 
  $\Config[\TLimit-\varepsilon'](x'-\varepsilon'{\cdot}\SpeedFun(\mu)) = \mu$.
  Finally, static accumulations are set.
\end{definition}

At each time, a position is set only if it is not already set. At the end every unset position is assigned by \Void. The dynamics is uniform in both space and time.
Please note that this definition does not always define an extension to the computation 
(when there are infinitely many signals, $\Delta(\Config)$ is an infimum that could be equal to zero), nevertheless it does, in the cases considered here.
\Figure{fig:space-time-diagram} illustrates the sequence of collision times for the previously given in example.

\begin{figure}[hbt]
  \subfigure[Causal past or backward light cone.\label{fig:cone}]{%
    \Width 0.9em
    \begin{tikzpicture}[x=\Width,y=\Width]
      \draw[white,fill=gray!30] (0,10) -- (6,4) -- (0,0) -- (0,10);
      \draw[white,fill=gray!30] (15,10) -- (6,4) -- (10,0) -- (15,0) -- (15,10);
      \draw[<->] (0,-1) -- node[below]{\small Space} (15,-1);
      \draw[->] (-1,0) -- node[sloped,above]{\small Time} (-1,10);
      \draw (0,10) -- node[above=-.2em,sloped]{\small max left speed} (5,5) -- (10,0);
      \draw (15,10) -- node[above=-.2em,sloped]{\small max right speed} (6,4) -- (0,0);
      \draw (6,4) node[right]{$(x,t)$};
      \draw (6,4) node[inner sep=.1em,minimum size=.1em,draw,circle,fill=black]{};
      \draw (5.4,1.45) node{\normalsize causal past};
      \draw (5.2,.2) node{$\CausalPast(x,t)$};
    \end{tikzpicture}}
  \subfigure[First collision times of the diagram of {\Fig{fig:example_sm-std}}.\label{fig:space-time-diagram}]{\qquad%
    \small\newcommand{\XT}{-1}
    \begin{tikzpicture}
      \draw[<->] (\XT.75,0) -- node[below] {Space}  (4.75,0);
      \draw[->] (\XT,-.25) --   (\XT,3.5) node[above] {Time};
      \begin{scope}[nodes={inner sep=.1em,minimum size=.1em,draw,circle}]
        \draw (1,1) node (A) {};
        \draw (3.4,2.2) node (B) {};
        \draw (3,3) node (C) {};        
      \end{scope}
      \draw[green] (0,0) -- (A) -- (C) -- (3.5,3.5);
      \draw[thick,dashed,purple] (3.4,0) -- (B);
      \begin{scope}[very thick]
        \draw (1.5,0) -- (A) -- (-.25,3.5);
        \draw (B) -- (C) -- (2.75,3.5);
      \end{scope}
      \begin{scope}[very thick,dotted,blue]
        \draw (A) -- (B);
        \draw (C) -- (4,3.5);
      \end{scope}
      \begin{scope}[dotted,left]
        \draw (\XT,0) node [above left] {$t_0$};
        \foreach \t/\l in {1/1,2.2/2,3/3} 
        \draw (4.5,\t) -- (\XT.25,\t) node {$t_{\l}$};
      \end{scope}
    \end{tikzpicture}\qquad}\quad
  \caption{Example of a space-time diagram and causal past.}\label{fig:def+cone}
\end{figure}
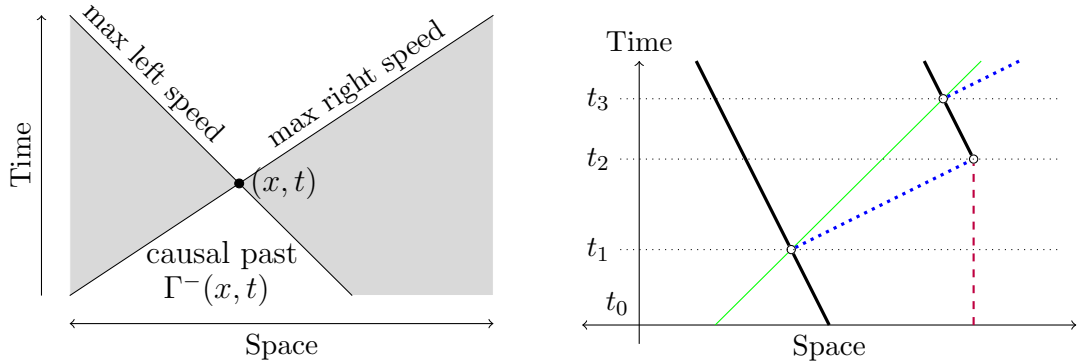

\paragraph{Notations}

In the whole paper, \SM will always designate a signal machine, and the function (resp. the parameters) of the machine can be indicated as index (resp. exponent). 
For instance, $\SM[2]^{a,b}$ will designate a $2$-speed machine, having its speed values equal to $a$ and $b$.
A signal name will be noted in \SigName{sf} font, and the sign of its speed will be indicated by a right or left overlined arrow.
Configurations at time $t$ will always be noted by \Config[t], and some parameters can be expressed as exponent, \eg, $\Config[0]^{x}$ stands for an initial configuration in which the value $x$ is used for some positions.
The notations \STD and \STDVal respectively represent diagrams and diagrams values, \ie, \STDVal is the set $\SigSet\cup\ColSet\cup\{\Void\}\cup\{\Accu\}$.

%
%
%
\subsection{A 3-speed toy example: computing the modulo}
\label{subsec:examples}
\newcounter{rulesnumberssub}
\newcounter{rulesnumbersmod}
\setcounter{rulesnumberssub}{0}
\setcounter{rulesnumbersmod}{0}

When considering \emph{small} signal machines, that is with few distinct speeds, one natural question is whether such machines can still process meaningful computations.
As a matter of fact, we may see in \Sec{subsec:2speeds} that allowing only two distinct speeds is too restrictive.

In this section, we provide two detailed (yet informal) examples of 3-speed machines illustrating our computational model on simple practical algorithms.
Those two machines geometrically compute the subtraction and \emph{modulo} operators, respectively, between two positive reals.

The notions of euclidean division, modulo and greatest common divisor (which will be used in \Sec{sec:3speeds_accu}), usually defined above integers, are generalized here to real numbers in a natural way.
More precisely, for $a, b\in\R$, the \emph{euclidean division (or integer division) of $a$ by $b$} is the unique $n \in \Z$ (given by $\lfloor\frac{a}{b}\rfloor$) so that $a = b \cdot n + r$, where $r \in \R$ and $0 \leq r < |b|$.
The remainder $r$ defines the value of \emph{$a$ modulo $b$}, noted by $a \bmod b$.
We say that \emph{the real number $b$ divides the real number $a$} if $a \bmod b = 0$.
The \emph{greatest common divisor of two reals $a$ and $b$}, denoted by $\gcd(a,b)$, is the greatest real which divides both $a$ and $b$.
By now, we will only consider positive real numbers, since divisions, modulo and $\gcd$ of any reals can easily be deduced from those of their absolute values.

\paragraph{Geometrical encoding of a value}

There may be several ways to geometrically encode values in a signal machine.
In this article, we choose to encode any value $x \neq 0$ (integer or real) by the distance between two stationary signals \SigWallZero and \SigWallx (where \SigName{w} stands for ``\SigName{wall}''), the stationary signal \SigWallZero being unique and common to every encoded value.

\paragraph{A 3-speed machine for the subtraction}

We first illustrate our computational model by constructing a 3-speed machine \SMsubtraction computing a single subtraction $a - b$ between two positive values $a$ and $b$.
For the sake of comprehension, we suppose $a > b$.

We define meta-signals and collision rules in order to implement the following idea.
The two positive values $a$ and $b$ are encoded by the distance between a common stationary meta-signal \SigSubWallZero and respective stationary meta-signals \SigSubWalla and \SigSubWallb, as explained above.
Using several temporary meta-signals of speed $-1$ and $1$ (\SigSubZigSmall, \SigSubZagSmall, \SigSubZigBig, \SigSubZagBig), we geometrically copy the distance between \SigSubWallZero and \SigSubWallb, which represents the value $b$, and shift the signal \SigSubWalla to the left by this exact distance (which is then renamed \SigSubWallr as being the result of the subtraction operation).
The initial configuration is set to $\Config[0]^{a,b} = \{ \SigAt{\SigSubInit}{-1}, \SigSubWallZero@0, \SigSubWallb@b, \SigSubWalla@a \}$.

Using basic geometry notions and observing that the signals shifting operation defines a parallelogram on the machine's diagram, one can easily prove that the position of the signal \SigSubWallr is such that the distance between \SigSubWallZero and \SigSubWallr corresponds exactly to $r = a - b$.
The definition of this machine is given in \Fig{fig:SMsubtraction_rules}, and we give a run example of this machine for values $a = 11$ and $b = 3$ in \Fig{fig:SMsubtraction_run-example}.

\begin{figure}[hbt]%
  \centering
  \subfigcapskip 1em
  \subfigure[Meta-signals and collisions rules of \SMsubtraction.\label{fig:SMsubtraction_rules}]{%
    \footnotesize%
    \begin{tabular}[b]{c}
      \begin{tabular}[b]{r|c}
		\bf Meta-signals & \bf Speeds \\
		\hline
		\SigSubInit, \SigSubZigSmall, \SigSubZigBig  &  $1$ \rule{0pt}{5mm} \\[.2em]
		\SigSubWallZero, \SigSubWalla, \SigSubWallb, \SigSubWallr  &  $0$ \\[.2em]
		\SigSubZagSmall, \SigSubZagBig  &  $-1$ \\[.2em]
      \end{tabular}
      \\[1em]%
      \begin{tabular}[b]{@{\scriptsize}>{(\refstepcounter{rulesnumberssub}\therulesnumberssub)}l r@{ $\rightarrow$ }l}
		\multicolumn{3}{c}{\bf Collision rules} \\
		\hline
		\label{rules:SMsubtraction_1} & \{ \SigSubInit, \SigSubWallb \} & \{ \SigSubZagSmall, \SigSubWallb, \SigSubZigBig \}\rule{0pt}{5mm} \\[.2em] 
		\label{rules:SMsubtraction_2} & \{ \SigSubWallZero, \SigSubZagSmall \} & \{ \SigSubWallZero, \SigSubZigSmall \} \\[.2em] 
		\label{rules:SMsubtraction_3} & \{ \SigSubZigSmall, \SigSubWallb \} & \{ \SigSubZigBig \} \\[.2em] 
		\label{rules:SMsubtraction_4} & \{ \SigSubZigBig, \SigSubWalla \} & \{ \SigSubZagBig \} \\[.2em] 
		\label{rules:SMsubtraction_5} & \{ \SigSubWallb, \SigSubZagBig \} & \{ \SigSubZagBig, \SigSubWallb \} \\[.2em] 
		\label{rules:SMsubtraction_6} & \{ \SigSubZigBig, \SigSubZagBig \} & \{ \SigSubWallr \} \\[.2em] 
		\label{rules:SMsubtraction_7} & \{ \SigSubZigSmall, \SigSubZagBig \} & \{ \SigSubWallr \} \\[.2em] 
		\label{rules:SMsubtraction_8} & \{ \SigSubZigSmall, \SigSubWallb, \SigSubZagBig \} & \{ \SigSubWallr \} \\[.2em] 
      \end{tabular}%
    \end{tabular}}%
  \qquad
  \subfigure[A run of \SMsubtraction, computing $11 - 3 = 8$.\label{fig:SMsubtraction_run-example}]{\qquad%
    \begin{tikzpicture}[x=4.5mm,y=4.5mm,font=\footnotesize]
\definecolor{couleurinit}{RGB}{155,20,157}
\tikzstyle{drawinit}=[draw=couleurinit]
\tikzstyle{nodeinit}=[pos=0.5, below, sloped]
\definecolor{couleurwall_0}{RGB}{0,0,0}
\tikzstyle{drawwall_0}=[draw=couleurwall_0]
\tikzstyle{nodewall_0}=[pos=0.5, left]
\definecolor{couleurwall_a}{RGB}{244,8,8}
\tikzstyle{drawwall_a}=[draw=couleurwall_a]
\tikzstyle{nodewall_a}=[pos=0.5, left]
\definecolor{couleurwall_b}{RGB}{0,0,255}
\tikzstyle{drawwall_b}=[draw=couleurwall_b]
\tikzstyle{nodewall_b}=[pos=0.5, right]
\definecolor{couleurwall_result}{RGB}{155,20,157}
\tikzstyle{drawwall_result}=[draw=couleurwall_result]
\tikzstyle{nodewall_result}=[pos=0.5, left]
\definecolor{couleurZIG}{RGB}{0,255,0}
\tikzstyle{drawZIG}=[draw=couleurZIG]
\tikzstyle{nodeZIG}=[pos=0.5, above, sloped]
\definecolor{couleurZAG}{RGB}{0,255,0}
\tikzstyle{drawZAG}=[draw=couleurZAG]
\tikzstyle{nodeZAG}=[pos=0.5, above, sloped]
\definecolor{couleurzig}{RGB}{10,83,5}
\tikzstyle{drawzig}=[draw=couleurzig]
\tikzstyle{nodezig}=[pos=0.5, above, sloped]
\definecolor{couleurzag}{RGB}{10,83,5}
\tikzstyle{drawzag}=[draw=couleurzag]
\tikzstyle{nodezag}=[pos=0.5, above, sloped]
\definecolor{black}{RGB}{0,0,0}
\tikzstyle{nodearrow_b}=[pos=0.5, below]
\tikzstyle{nodearrow_a}=[pos=0.5, below]
\tikzstyle{nodearrow_r}=[pos=0.5, above]

\draw[drawwall_b] (3.000000,0.000000) -- node[nodewall_b]{\SigSubWallb} (3.000000,4.000000);
\draw[drawinit] (-1.000000,0.000000) -- node[nodeinit]{\SigSubInit} (3.000000,4.000000);
\draw[drawzag] (3.000000,4.000000) -- node[nodezag]{\SigSubZagSmall} (0.000000,7.000000);
\draw[drawwall_0] (0.000000,0.000000) -- node[nodewall_0]{\SigSubWallZero} (0.000000,7.000000);
\draw[drawwall_b] (3.000000,4.000000) -- node[nodewall_b]{\SigSubWallb} (3.000000,10.000000);
\draw[drawzig] (0.000000,7.000000) -- node[nodezig]{\SigSubZigSmall} (3.000000,10.000000);
\draw[drawwall_a] (11.000000,0.000000) -- node[nodewall_a]{\SigSubWalla} (11.000000,12.000000);
\draw[drawZIG] (3.000000,4.000000) -- node[nodeZIG]{\SigSubZigBig} (11.000000,12.000000);
\draw[drawZAG] (11.000000,12.000000) -- node[nodeZAG]{\SigSubZagBig} (8.000000,15.000000);
\draw[drawZIG] (3.000000,10.000000) -- node[nodeZIG]{\SigSubZigBig} (8.000000,15.000000);
\draw[drawwall_0] (0.000000,7.000000) -- node[nodewall_0]{\SigSubWallZero} (0.000000,16.000000);
\draw[drawwall_result] (8.000000,15.000000) -- node[nodewall_result]{\SigSubWallr} (8.000000,16.000000);

\draw[arrows=<->] (0.000000,-0.500000) -- node[nodearrow_a]{$a$} (11.000000,-0.500000);
\draw[arrows=<->] (0.000000,-1.000000) -- node[nodearrow_b]{$b$} (3.000000,-1.00000);
\draw[arrows=<->] (0.000000,16.500000) -- node[nodearrow_r]{$a-b$} (8.000000,16.500000);
\end{tikzpicture}}%
  \caption{Computation of the subtraction operator, using three distinct speeds.}\label{fig:SMsubtraction_machine}
\end{figure}
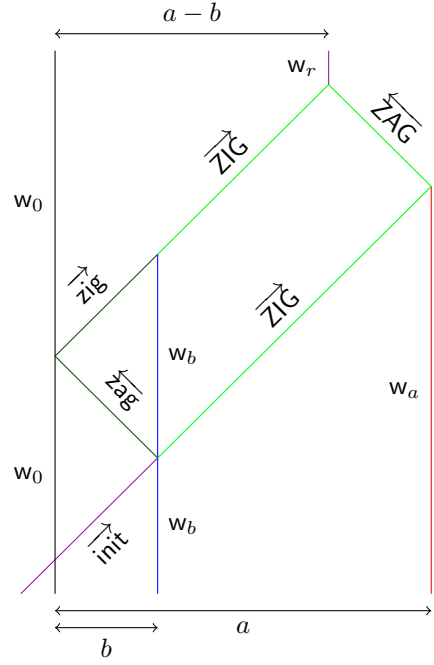

\paragraph{A 3-speed machine for the modulo}

We now construct a 3-speed machine \SMmodulo computing the mathematical operation $a \bmod b$ between two positive values $a$ and $b$.
This mathematical operation corresponds basically to successive (possibly zero) subtractions of $b$ to $a$ until the result is strictly smaller than $b$.
We therefore consider and adapt the 3-speed machine \SMsubtraction defined above which computes a single subtraction $a - b$.

We define meta-signals and collision rules in order to implement the following idea.
As for the \SMsubtraction 3-speed machine, the two positive values $a$ and $b$ are encoded by the distance between a stationary meta-signal \SigModWallZero and respective stationary signals \SigModWalla and \SigModWallb.
We reuse most of the meta-signals and collision rules defined for the \SMsubtraction 3-speed machine, and adapt rules~\ref{rules:SMsubtraction_3}, \ref{rules:SMsubtraction_5}, \ref{rules:SMsubtraction_6} and \ref{rules:SMsubtraction_7} of \SMsubtraction so that the machine repeats the subtraction operation as long as the result $r = a - i \times b$ is still greater or equal to $b$ (where $i \geq 0$ is the number of subtractions processed so far), that is as long as the shift of the meta-signal \SigModWallr has not \emph{crossed} the meta-signal \SigModWalla.
Note that the collisions rules handling the end of the computation must consider the two possible cases, where either $r = a \bmod b > 0$ (rule~\ref{rules:SMmodulo_8}) or $r = a \bmod b = 0$ (rule~\ref{rules:SMmodulo_6}, then rule~\ref{rules:SMmodulo_9}).
The initial configuration of this machine is set to $\Config[0]^{a,b} = \{ \SigModInit@-1, \SigModWallZero@0, \SigModWallb@b, \SigModWalla@a \}$.

As for the previous machine, using basic geometry notions and observing that the signals shifting operation defines a parallelogram on the machine's diagram, one can easily prove that the position of the final signal \SigSubWallr is such that the distance between \SigSubWallZero and \SigSubWallr corresponds exactly to $r = a \bmod b$.
The definition of this machine is given in \Fig{fig:SMmodulo_rules}, and we give a run example of this machine for $a = 11$ and $b = 3$ in \Fig{fig:SMmodulo_run-example}.

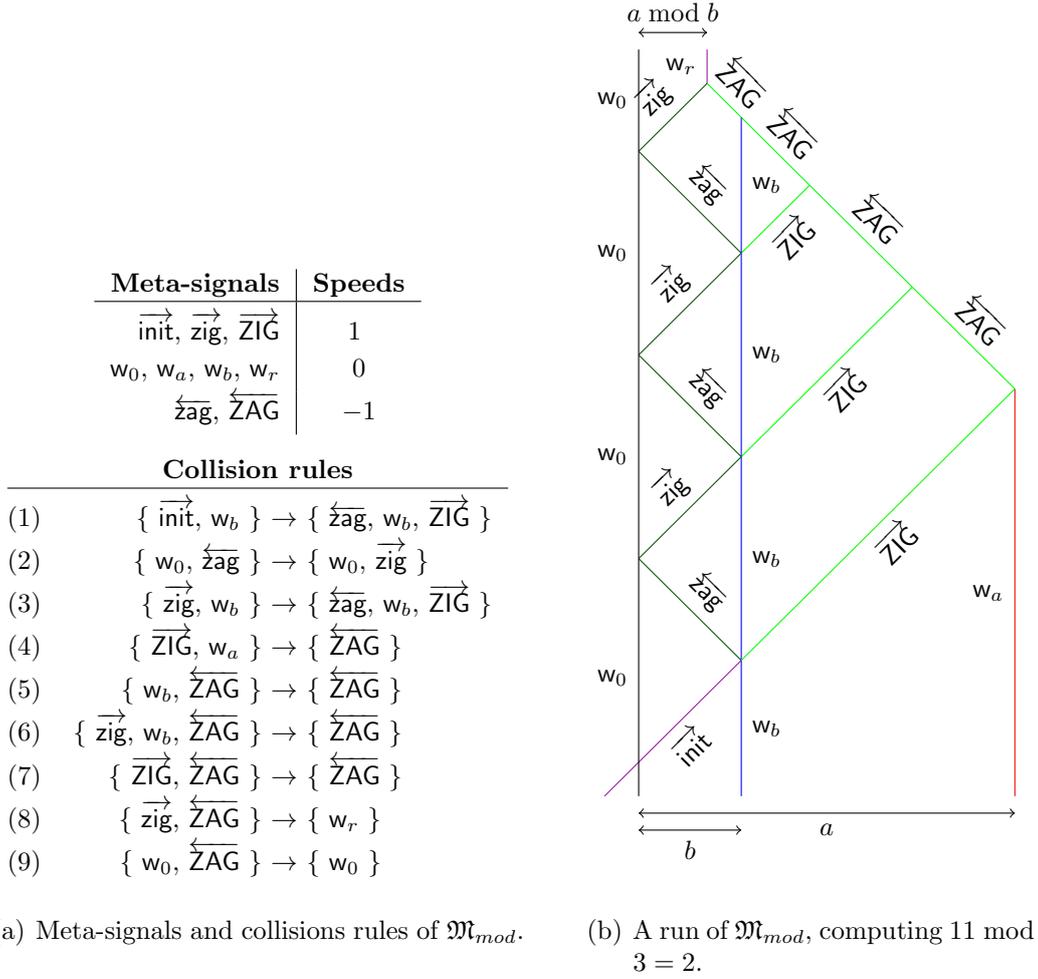
\begin{figure}[hbt]%
  \centering
  \subfigcapskip 1em
  \subfigure[Meta-signals and collisions rules of \SMmodulo.\label{fig:SMmodulo_rules}]{%
    \footnotesize%
    \begin{tabular}[b]{c}
      \begin{tabular}[b]{r|c}
		\bf Meta-signals & \bf Speeds \\
		\hline
		\SigModInit, \SigModZigSmall, \SigModZigBig  &  $1$ \rule{0pt}{5mm} \\[.2em]
		\SigModWallZero, \SigModWalla, \SigModWallb, \SigModWallr  &  $0$ \\[.2em]
		\SigModZagSmall, \SigModZagBig  &  $-1$ \\[.2em]
      \end{tabular}
      \\[1em]%
      \begin{tabular}[b]{@{\scriptsize}>{(\refstepcounter{rulesnumbersmod}\therulesnumbersmod)}l r@{ $\rightarrow$ }l}
		\multicolumn{3}{c}{\bf Collision rules} \\
		\hline
		\label{rules:SMmodulo_1} & \{ \SigModInit, \SigModWallb \} & \{ \SigModZagSmall, \SigModWallb, \SigModZigBig \}\rule{0pt}{5mm} \\[.2em] 
		\label{rules:SMmodulo_2} & \{ \SigModWallZero, \SigModZagSmall \} & \{ \SigModWallZero, \SigModZigSmall \} \\[.2em] 
		\label{rules:SMmodulo_3} & \{ \SigModZigSmall, \SigModWallb \} & \{ \SigModZagSmall, \SigModWallb, \SigModZigBig \} \\[.2em] 
		\label{rules:SMmodulo_4} & \{ \SigModZigBig, \SigModWalla \} & \{ \SigModZagBig \} \\[.2em] 
		\label{rules:SMmodulo_5} & \{ \SigModWallb, \SigModZagBig \} & \{ \SigModZagBig \} \\[.2em] 
		\label{rules:SMmodulo_6} & \{ \SigModZigSmall, \SigModWallb, \SigModZagBig \} & \{ \SigModZagBig \} \\[.2em] 
		\label{rules:SMmodulo_7} & \{ \SigModZigBig, \SigModZagBig \} & \{ \SigModZagBig \} \\[.2em] 
		\label{rules:SMmodulo_8} & \{ \SigModZigSmall, \SigModZagBig \} & \{ \SigModWallr\} \\[.2em] 
		\label{rules:SMmodulo_9} & \{ \SigModWallZero, \SigModZagBig \} & \{ \SigModWallZero\} \\[.2em] 
      \end{tabular}%
    \end{tabular}}%
  \qquad
  \subfigure[A run of \SMmodulo, computing $11 \bmod 3 = 2$.\label{fig:SMmodulo_run-example}]{%
    \begin{tikzpicture}[x=4.5mm,y=4.5mm,font=\footnotesize]
\definecolor{couleurinit}{RGB}{155,20,157}
\tikzstyle{drawinit}=[draw=couleurinit]
\tikzstyle{nodeinit}=[pos=0.5, below, sloped]
\definecolor{couleurwall_0}{RGB}{0,0,0}
\tikzstyle{drawwall_0}=[draw=couleurwall_0]
\tikzstyle{nodewall_0}=[pos=0.5, left]
\definecolor{couleurwall_a}{RGB}{244,8,8}
\tikzstyle{drawwall_a}=[draw=couleurwall_a]
\tikzstyle{nodewall_a}=[pos=0.5, left]
\definecolor{couleurwall_b}{RGB}{0,0,255}
\tikzstyle{drawwall_b}=[draw=couleurwall_b]
\tikzstyle{nodewall_b}=[pos=0.5, right]
\definecolor{couleurwall_result}{RGB}{155,20,157}
\tikzstyle{drawwall_result}=[draw=couleurwall_result]
\tikzstyle{nodewall_result}=[pos=0.5, left]
\definecolor{couleurZIG}{RGB}{0,255,0}
\tikzstyle{drawZIG}=[draw=couleurZIG]
\tikzstyle{nodeZIG}=[pos=0.5, below, sloped]
\definecolor{couleurZAG}{RGB}{0,255,0}
\tikzstyle{drawZAG}=[draw=couleurZAG]
\tikzstyle{nodeZAG}=[pos=0.5, above, sloped]
\definecolor{couleurzig}{RGB}{10,83,5}
\tikzstyle{drawzig}=[draw=couleurzig]
\tikzstyle{nodezig}=[pos=0.5, above, sloped]
\definecolor{couleurzag}{RGB}{10,83,5}
\tikzstyle{drawzag}=[draw=couleurzag]
\tikzstyle{nodezag}=[pos=0.5, above, sloped]
\tikzstyle{nodearrow_b}=[pos=0.5, below]
\tikzstyle{nodearrow_a}=[pos=0.5, below]
\tikzstyle{nodearrow_r}=[pos=0.5, above]

\draw[drawwall_b] (3.000000,0.000000) -- node[nodewall_b]{\SigModWallb} (3.000000,4.000000);
\draw[drawinit] (-1.000000,0.000000) -- node[nodeinit]{\SigModInit} (3.000000,4.000000);
\draw[drawzag] (3.000000,4.000000) -- node[nodezag]{\SigModZagSmall} (0.000000,7.000000);
\draw[drawwall_0] (0.000000,0.000000) -- node[nodewall_0]{\SigModWallZero} (0.000000,7.000000);
\draw[drawwall_b] (3.000000,4.000000) -- node[nodewall_b]{\SigModWallb} (3.000000,10.000000);
\draw[drawzig] (0.000000,7.000000) -- node[nodezig]{\SigModZigSmall} (3.000000,10.000000);
\draw[drawwall_a] (11.000000,0.000000) -- node[nodewall_a]{\SigModWalla} (11.000000,12.000000);
\draw[drawZIG] (3.000000,4.000000) -- node[nodeZIG]{\SigModZigBig} (11.000000,12.000000);
\draw[drawzag] (3.000000,10.000000) -- node[nodezag]{\SigModZagSmall} (0.000000,13.000000);
\draw[drawwall_0] (0.000000,7.000000) -- node[nodewall_0]{\SigModWallZero} (0.000000,13.000000);
\draw[drawZAG] (11.000000,12.000000) -- node[nodeZAG]{\SigModZagBig} (8.000000,15.000000);
\draw[drawZIG] (3.000000,10.000000) -- node[nodeZIG]{\SigModZigBig} (8.000000,15.000000);
\draw[drawwall_b] (3.000000,10.000000) -- node[nodewall_b]{\SigModWallb} (3.000000,16.000000);
\draw[drawzig] (0.000000,13.000000) -- node[nodezig]{\SigModZigSmall} (3.000000,16.000000);
\draw[drawZAG] (8.000000,15.000000) -- node[nodeZAG]{\SigModZagBig} (5.000000,18.000000);
\draw[drawZIG] (3.000000,16.000000) -- node[nodeZIG]{\SigModZigBig} (5.000000,18.000000);
\draw[drawzag] (3.000000,16.000000) -- node[nodezag]{\SigModZagSmall} (0.000000,19.000000);
\draw[drawwall_0] (0.000000,13.000000) -- node[nodewall_0]{\SigModWallZero} (0.000000,19.000000);
\draw[drawZAG] (5.000000,18.000000) -- node[nodeZAG]{\SigModZagBig} (3.000000,20.000000);
\draw[drawwall_b] (3.000000,16.000000) -- node[nodewall_b]{\SigModWallb} (3.000000,20.000000);
\draw[drawZAG] (3.000000,20.000000) -- node[nodeZAG]{\SigModZagBig} (2.000000,21.000000);
\draw[drawzig] (0.000000,19.000000) -- node[nodezig]{\SigModZigSmall} (2.000000,21.000000);
\draw[drawwall_0] (0.000000,19.000000) -- node[nodewall_0]{\SigModWallZero} (0.000000,22.000000);
\draw[drawwall_result] (2.000000,21.000000) -- node[nodewall_result]{\SigModWallr} (2.000000,22.000000);

\draw[arrows=<->] (0.000000,-0.500000) -- node[nodearrow_a]{$a$} (11.000000,-0.500000);
\draw[arrows=<->] (0.000000,-1.000000) -- node[nodearrow_b]{$b$} (3.000000,-1.00000);
\draw[arrows=<->] (0.000000,22.500000) -- node[nodearrow_r]{$a \bmod b$} (2.000000,22.500000);
\end{tikzpicture}}%
  \caption{Computation of the modulo (\SMmodulo), using three distinct speeds.}\label{fig:SMmodulo_machine}
\end{figure}

%
%
%
\subsection{Some geometrical properties}
\label{subsec:properties}
We give some properties and relations between signal machines, which will be useful for characterizing diagrams having accumulations, relatively to diagrams of some other machines.
Notions and properties presented below define informally relations of embedding and equivalence between signal machines, based on \Def{def:diagram-equivalency} of equivalent diagrams.
Indeed, two signal machines generating equivalent diagrams are intuitively {\em equivalent}. We also introduce a notion of inclusion between diagrams.

\paragraph{Transformations under affine functions}

We show in this paragraph that every signal machine can be transformed into an equivalent signal machine 
whose speed values include $0$ and $1$ (or any two other distinct real numbers).
In fact, we show that applying an affine function to speed values does not change the space-time diagram structure:

\begin{lemma}\label{lemma:affine-trans}
  Let \SM be a signal machine and $f:\R\to\R$ an affine function of strictly positive ratio.
  Let \SM[f] be the signal machine obtained by applying $f$ to all speeds of \SM, \ie, the speed function of \SM[f] is $f\circ\SpeedFun$ (where \SpeedFun is the one of \SM).
  Then \SM[f] generates space-time diagrams topologicaly equivalent to the ones generated by the machine \SM.
\end{lemma}

\begin{proof}
  Let $a\in\R^+$ and $b\in\R$ so that $f(x) = a\cdot x + b$ for all $x\in\R$.
  Let \STD be a space-time diagram of \SM and $\STD'$ the diagram generated by 
  \SM[f] with the same initial configuration.

  Adding the constant $b$ to all speeds drifts progressively all positions and leaves all dates unchanged.
  It can easily be checked that for all $(x,t)\in\SpaceTime$, $\STD'(x,t) = \STD(x-b\cdot t,t)\enspace$.
  In the case of a signal \Sig located in $x_0$ at a given time $t$, its new position (in \STD)
  after a time duration $\Delta t$ is given by $x_1= \Delta t \cdot \SpeedFun(\Sig) + x_0$ (unless it collides before).
  In $\STD'$, its new position is given by 
  $x_1' = \Delta t\cdot (\SpeedFun(\Sig)+b) + x_0 
	= \Delta t \cdot \SpeedFun(\Sig) + x_0 + \Delta t \cdot b 
	= x_1 + \Delta t \cdot b $ and we have indeed $\STD'(x,t) = \STD(x-b\cdot t,t)\enspace$.
  In the case of a collision happening at coordinates $(x, t)$ between two signals \Sig[1] and \Sig[2] 
  (if more than two signals collide, the argument is the same but with a system of equations instead of one), 
  we know that there is a time $t_0$ so that $t$ is solution of the equation 
  $(t-t_0)\cdot\SpeedFun(\Sig[1]) + x_1 = (t-t_0)\cdot\SpeedFun(\Sig[2]) + x_2\enspace$,
  where $x_1$ (resp. $x_2$) is the spatial position of \Sig[1] (resp. \Sig[2]) at time $t_0$.
  Adding $b$ to speeds does not change the equation ($(t-t_0)\cdot b = \Delta t\cdot b$ appears on both side of the equation), 
  so that the time $t$ of the collision remains the same in the drifted diagram $\STD'$.
  After drifting, the new location $x'$ of the collision is given by 
  $x' = (t-t_0)\cdot(\SpeedFun(\Sig[1]) + b) + x_1 = x + \Delta t \cdot b \enspace$, so we also obtain in the case of a collision
  $\STD'(x,t) = \STD(x-b\cdot t,t)\enspace$.
  As position of all signals and collisions are drifted, the position of an accumulation will also be drifted.

  We show in the same way that multiplying all speeds by $a$ modifies all dates but 
  keep the spatial position values (because $a > 0$). We have for all $(x,t)\in\SpaceTime$, 
  $\STD'(x,t) = \STD(x,a\cdot t)\enspace$.

  Finally, applying $f$ to speeds is equivalent to the condition that 
  for all $(x,t)\in\SpaceTime$, $\STD'(x,t) = \STD(x-b\cdot t,a\cdot t)\enspace$.
  
  Since adding $b$ to speeds only drifts all positions and multiplying speeds by $a$ contracts (or distends) uniformly all times, parallel signals, colliding signals and 
  simultaneous collisions in \STD still are in $\STD'$.

  The function $h:\SpaceTime\to\SpaceTime$ defined by $h(x,t) = (x-b\cdot t,a\cdot t)$ is a homeomorphism
  (both components of $h$ are continuous and the bijectivity is easily checked).
  So there exists a homeomorphism $h$ so that for all $(x,t)\in\SpaceTime$, $\STD'(x,t)=\STD(h(x,t))$, \ie, $\STD$ and $\STD'$ are equivalent by \Def{def:diagram-equivalency}.
\end{proof}

It follows that, given a signal machine \SM whose speed values include $a$ and $b$ (with $a < b$), we can always transform \SM into an equivalent machine \SM[f] so that speed values of \SM[f] include the values $c$ and $d$ with $c < d$.
Indeed, the function $f:\R\to\R$ given by $f(x)=\frac{d-c}{b-a}\cdot x + \frac{cb-ad}{b-a}$ is an affine function of strictly positive ratio (since $c<d$ and $a<b$), and $f$ verifies $f(a)=c$ and $f(b)=d$.
We obtain by the previous lemma that \SM and \SM[f] generate equivalent topological space-time diagrams.

We show the same way that we can apply affine functions to initial configurations without changing the structure of the generated diagram:

\begin{lemma}\label{lemma:config_affine-transformation}
  Let $f:\R\to\R$ be an affine function of strictly positive ratio and \STD a diagram generated by a machine \SM from a configuration \Config[0].
  Then the diagram $\STD'$ generated by \SM from the  configuration $\Config[0]'$ defined by $\Config[0]'(x)=\Config[0](f(x))$ for all $x\in\R$,
  is equivalent to the diagram \STD.
\end{lemma}

Indeed, for all $(x,t)\in\SpaceTime$, we have $\STD'(x,t)=\STD(x-b,\frac{1}{a}\cdot t)$.

\paragraph{Notion of support}

We will define now a notion of \emph{support}, both for machines and diagrams.
Intuitively, given a signal machine \SM, the \emph{support signal machine \Supp{\SM}} of \SM will be defined by considering the set of distinct speed values of \SM.
As we are looking for accumulations, all the collision rules will be set to produce the maximal number of outcoming signals.
This follows the intuitive idea that accumulations occur more easily when collisions create a lot of signals. 
Then, any space-time diagram of \SM will be embedded into a \emph{support space-time diagram} of \Supp{\SM}, with the property of keeping the existence of accumulations.

Let $\SM=(\SigSet, \SpeedFun, \ColSet)$ be a signal machine. 
We define on \SigSet the binary relation $\sim$ so that for all $\mu, \sigma \in \SigSet$, $\mu \sim \sigma \Leftrightarrow \SpeedFun(\mu) = \SpeedFun(\sigma)$. 
$\sim$ is clearly an equivalence relation.
An equivalence class for $\sim$ contains exactly all meta-signals of \SigSet having the same speed.
We write $\EqClass{\Sig}$ to designate the equivalence class of the meta-signal \Sig.
Since each class is finite, we can choose a system of representants $\{\Sig[i]\}_{1\leq i \leq n}$, 
where $n$ is the number of equivalence classes, \ie, the number of distinct speed values of \SM.

\begin{definition}\label{def:support-machine}
  Let $\SM=( \SigSet, \SpeedFun, \ColSet)$ be a signal machine. 
  The {\em support machine \Supp{\SM} of \SM} is the machine $\Supp{\SM}= (\SigSet', \SpeedFun', \ColSet')$ so that:
    \begin{enumerate}[(i)]
      \item $\SigSet' = \SetQuotient{\SigSet}{\sim}$ ;
      \item $\SpeedFun':\SigSet'\to\R$ is defined by $\SpeedFun'(\EqClass{\Sig})=\SpeedFun(\Sig)$ for all $\mu$; 
      \item $\ColSet' = \{\ C^- \rightarrow C^+ \ |\ C^- \subseteq \SigSet', \Card{C^-} \geq 2 \text{ and } C^+= \SigSet' \ \}$.
    \end{enumerate}
\end{definition}

For each distinct speed value, we choose only one meta-signal of the original machine \SM having this speed.
The set of collision rules is defined as follow: for each possible collision, the set of outcoming signals
is $\SigSet'$ (the whole set of meta-signals). Such rules can be defined since all meta-signals in $\SigSet'$ have distinct speeds.
Note that if \SM is a $n$-speed machine, then \Supp{\SM} is also a $n$-speed machine.

We can extend the canonical surjection $\Sig\mapsto\EqClass{\Sig}$ into a surjection $\Pi:\ColSet\to\ColSet'$.
For each $C = \{\Sig[i]^-\}_{i \in I} \rightarrow \{\Sig[j]^+\}_{j \in J} \in \ColSet$, we define 
$\Pi(C) = C'\in\ColSet'$ where $C' = \{\EqClass{\Sig[i]^-}\}\rightarrow\SigSet'$ 
(since the set of outcoming meta-signals of $C'$ is the whole set of meta-signals $\SigSet'$, $C'$ is indeed in \ColSet' defined previously).
We can also extend this surjection to the set of initial configurations.
Given an initial configuration \Config[0] of the machine \SM, 
we write \Supp{\Config[0]} for the configuration \Config[0] in which every meta-signal \Sig is replaced by the meta-signal \EqClass{\Sig}.
Clearly, \Supp{\Config[0]} is an initial configuration of the machine \Supp{\SM}.

We can now define the notion of {\em support space-time diagram}:

\begin{definition}[Support diagram]
  Let \STD be a space-time diagram of the signal machine \SM started from the  initial configuration \Config[0].
  We define \Supp{\STD}, the {\em support diagram of \STD}, as the space-time diagram of \Supp{\SM} executed on the initial configuration \Supp{\Config[0]}.
\end{definition}

This notion of support diagram of another diagram has to be carefully distinguished from the notion of support of a diagram: 
the {\em support diagram} of a diagram \STD is a diagram (generated by a support machine) whereas the {\em support} of the diagram \STD is a set 
(given by $\Support{\STD} = \{ (x,t)\in\SpaceTime ~|~ \STD(x,t)\neq\Void \}$.

\begin{figure}[hbt]
  \centering
  \Width 1cm
  \newcommand{\XT}{-1}
  \subfigure[A space-time diagram.]{%
    \begin{tikzpicture}[x=\Width,y=\Width]
      \draw[<->] (\XT.75,0) -- node[right,pos=1] {\Space}  (4.75,0);
      \draw[->] (\XT,-.25) --   (\XT,3.5) node[above] {\Time};
      \begin{scope}
       \clip (-.25,0) rectangle (4.75,3.31) ;
      \begin{scope}[nodes={inner sep=.07em,minimum size=.1em,draw,circle}]
        \draw (1,1) node (A) {};
        \draw (3.4,2.2) node (B) {};
        \draw (3,3) node (C) {};        
      \end{scope}
      \draw[green] (0,0) -- (A) -- (C) -- (3.5,3.5);
      \draw[thick,dashed,purple] (3.4,0) -- (B);
      \begin{scope}[very thick]
        \draw (1.5,0) -- (A) -- (-.25,3.5);
        \draw (B) -- (C) -- (2.75,3.5);
      \end{scope}
      \begin{scope}[very thick,dotted,blue]
        \draw (A) -- (B);
        \draw (C) -- (4,3.5);
      \end{scope}
    \end{scope}
    \end{tikzpicture}\quad}
  \subfigure[The corresponding support diagram.]{%
    \begin{tikzpicture}[x=\Width,y=\Width]
    \begin{scope}
       \clip (-.25,0) rectangle (4.75,3.31) ;
       \begin{scope}[nodes={inner sep=.07em,minimum size=.1em,draw,circle}]
        \draw (1,1) node (A) {};
        \draw (3.4,2.2) node (B) {};
        \draw (3,3) node (C) {};
        \draw (3.4,3.4) node (D) {}; 
        \draw (3.4,3.2) node (E) {};    
      \end{scope}
	\draw[thick] (0,0) -- (A) -- (C) -- (D) ;
	\draw[thick] (3.4,0) -- (B) ;
        \draw[thick] (1.5,0) -- (A) -- (-.25,3.5) ;
        \draw[thick] (B) -- (C) -- (2.75,3.5) ;]
        \draw [thick](A) -- (B);
        \draw[thick] (E) -- (4,3.5);
        \draw[thick] (E) -- (3.51,3.31);
        \draw[thick] (A) -- (1,3.5);
        \draw[thick] (E) -- (3.2,3.5);
        \draw[thick] (B) -- (E);
        \draw[thick] (E) -- (D);
        \draw[thick] (C) -- (E);
        \draw[thick] (B) -- (4.75,2.9);
        \draw[thick] (B) -- (4.75,2.9);
        \draw[thick] (C) -- (3,3.55);
        \draw[thick] (B) -- (4.75,3.55) ;
        \draw[thick] (D) -- (3.4,3.5) ;
        \draw[thick] (D) -- (3.35,3.5) ;
        \draw[thick] (D) -- (3.6,3.5) ;
      \end{scope}
      \draw[<->] (\XT.75,0) -- node[right,pos=1] {\Space}  (4.75,0);
      \draw[->] (\XT,-.25) -- (\XT,3.5) node[above] {\Time};
    \end{tikzpicture}}
  \caption{A space-time diagram and its support diagram.\label{fig:diagram-support_example}}
\end{figure}
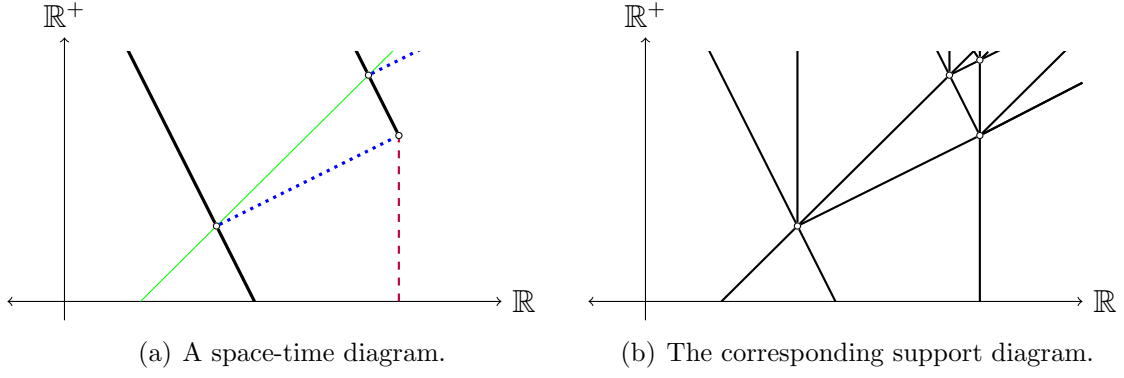

\paragraph{Inclusion of diagrams}

A relation of inclusion can be defined between diagrams (resp. configurations): a diagram (resp. configuration) is said to be included in another diagram (resp. configuration) 
if all its non-\Void positions are also non-\Void positions for the second diagram (resp. configuration).

\begin{definition}[Diagrams inclusion]\label{def:diagram-inclusion}
  Let $\STD$ and $\STD'$ be two space-time diagrams, respectively defined on $\Space\times [0;T]$ and $\Space\times [0;T']$.
  We say that {\em $\STD'$ is included in $\STD$}  (or that  {\em $\STD'$ is supported by $\STD$}) 
  if $\Support{\STD'}\cap\Space\times [0;inf(T,T')]\subseteq\Support{\STD}\cap\Space\times [0;inf(T,T')]$.
\end{definition}

We note $\STD'\STDinclu\STD$ if $\STD'$ is included in \STD.
The restriction of supports to the set $\Space\times [0;inf(T,T')]$ is necessary to compare supports of diagrams only on a space-time area on which they are both defined.

We can show that support diagrams ``bound'' the structures of diagrams in the sense that any diagram is included in its support diagram:

\begin{lemma}\label{lemma:diagram_support-inclusion}
  For any diagram $\STD$, we have $\STD\STDinclu\Supp{\STD}$.
\end{lemma}

\begin{proof}
  Let $\STD$ be a diagram of a signal machine $\SM=( \SigSet, \SpeedFun, \ColSet)$
  and $\Supp{\STD}$ be the support diagram of $\STD$, generated by $\Supp{\SM}= (\SigSet', \SpeedFun', \ColSet')$, 
  the support machine of \SM.
  Suppose that $\STD$ and $\Supp{\STD}$ are respectively defined on $\Space\times[0;T]$ and $\Space\times[0;T']$.
  Taking the support diagram do not remove any signal or collision of the previous diagram and can only add new object.
  For each signal (resp. collision) in $\STD\cap\Space\times [0;inf(T,T')]$ occuring at $(x,t)$, there exists a signal (resp. a collision)
  in $\Supp{\STD}\cap\Space\times [0;inf(T,T')]$ occuring at $(x,t)$.
  Indeed, for a collision $C = \{\Sig[i]^-\}_{i \in I} \rightarrow \{\Sig[j]^+\}_{j \in J}$ so that 
  $\STD(x, t)=C$, we have $\Supp{\STD}(x,t) = \Pi(C) = \{\EqClass{\Sig[i]^-}\}_{i \in I} \rightarrow \SigSet'\enspace$.
  $\Pi(C)$ has the same position $(x, t)$ that $C$, but it outputs more signals, corresponding exactly to all meta-signals 
  of the support machine.
  For a signal \Sig so that $\STD(x,t) = \Sig$, either $t=0$ and we 
  have $\Supp{\STD}(x,t) = \Supp{\STD[0]}(x) = \EqClass{\Sig}\enspace$, or \Sig was created by a collision $C$
  at $(x_0, t_0)$. In this case, we have $\Supp{\STD}(x_0, t_0) = \Pi(C)$ and the set of outcoming signals of $\Pi(C)$
  is $\SigSet'$, so $\Pi(C)$ generates the signal $\EqClass{\Sig}$. 
  As \Sig and $\EqClass{\Sig}$ are both generated in $(x_0, t_0)$ in their respective diagrams and as their speeds
  are equal by definition of $\EqClass{\Sig}$, 
  they have the same motion equations and will take the same positions.
  In particular we have 
  $\Supp{\STD}(x, t) = \EqClass{\Sig}\enspace$.
  So in both cases, if $\STD(x,t) = \Sig$ then $\Supp{\STD}(x, t) = \EqClass{\Sig}\enspace$.
  
  Assume that \STD contains an accumulation at $(x, t)$: $\STD(x,t) = \Accu$.
  There is an infinite number of signals and collisions in $\CausalPast(x, t)$ (the causal past of $(x, t)$ in \STD).
  As for every signal (resp. collision) occuring in $\CausalPast(x, t)$ in the diagram \STD, 
  there is a signal (resp. a collision) exactly at the same position in the support diagram \Supp{\STD}, 
  there must also be an infinite number of signals and collisions in $\Supp{\CausalPast}(x, t)$, 
  the causal past of $(x, t)$ in \Supp{\STD}.
  So we have $\Supp{\STD}(x, t) = \Accu\enspace$.
  
  Finally, we have $\forall(x,t)\ \STD(x,t)\neq\Void \Rightarrow \Supp{\STD}\neq\Void$, \ie, $\Support{\STD}\subseteq\Support{\Supp{\STD}}$.
\end{proof}

More generally, we can show that the relation of inclusion between diagrams keeps the existence of accumulations:

\begin{lemma}\label{lemma:accu-diagram_inclusion}
  Let \STD and $\STD'$ be two diagrams so that $\STD'\STDinclu\STD$.
  Then:\\
  $\forall (x,t)\in\SpaceTime\ \STD(x,t)\neq\Accu\ \Rightarrow \forall (x,t)\in\SpaceTime\ \STD'(x,t)\neq\Accu$.
\end{lemma}

That is, if $\STD'$ is included in $\STD$ and $\STD$ have no accumulation, then neither does $\STD'$.

\begin{proof}
  We have $\Support{\STD'}\STDinclu\Support{\STD}$.
  Since \STD doesn't contain any accumulation, we have $\forall(x,t)\in\SpaceTime$, $\STD(x,t)\neq\Accu$.
  For all $(x,t)$ and for all neighbourhood $\mathcal{V}_{(x,t)}\subseteq\SpaceTime$, 
  there is only a finite number of non-\Void positions in $\mathcal{V}_{(x,t)}$, \ie $\mathcal{V}_{(x,t)}\cap\Support{\STD}$ is finite.
  From $\Support{\STD'}\STDinclu\Support{\STD}$, we deduce that $\mathcal{V}_{(x,t)}\cap\Support{\STD'}$ is also finite.
  So for all $(x,t)\in\SpaceTime$, $\STD'(x,t)\neq\Accu$ \ie $\STD'$ have no accumulation.
\end{proof}

Combination of \Lem{lemma:diagram_support-inclusion} and \ref{lemma:accu-diagram_inclusion} implies the following corollary: 

\begin{corollary}\label{lemma:accu-transfert}
  Let \STD be a space-time diagram of a machine \SM.
  If \Supp{\STD} doesn't contain any accumulation, then neither does \STD.
\end{corollary}

We can also define the notion of {\em inclusion of configurations}: 
a configuration $\Config'$ is included in a configuration $\Config$ if $\Support{\Config'}\subseteq\Support{\Config}$.

Please note that by definition of support machines, if two initial configurations $\Config[0]$ and $\Config[0]'$ of the same support machine, producing respectively the diagrams $\STD$ and $\STD'$ (which are support diagrams since the machine is a support machine), are such that any signal in $\Config[0]'$ is also in $\Config[0]$ at the same position, it follows that $\STD'$ is included in $\STD$.

\paragraph{Topological accumulations and dynamics}

The following lemma links accumulation with dynamics:

\begin{lemma}\label{lemma:accu-dynamics}
  There is a (dynamic) accumulation at $(x,t)$ if and only if 
  there exists a sequence $(C_n)_{n\in\N}$ of collisions ordered by dates so that $\lim\limits_{n \to \infty} (x_n,t_n) = (x,t)$, 
  where $(x_n, t_n)$ are the coordinates of the collision $C_n$.
\end{lemma}

\begin{proof}
  In the case of a dynamic accumulation at $(x,t)$, \Def{def:isolated-accumulation} provides the sequence $(C_n)_{n\in\N}$:
  it is the sequence of collisions accumulating in \CausalPast(x,t), the causal past of $(x,t)$. 
  Since for all $n$, $C_n\in\CausalPast(x,t)$, we have $\lim\limits_{n \to \infty} (x_n,t_n)=(x,t)$.
  For the other implication, \Def{def:dynamics} implies directly that $(x,t)$ is a dynamic accumulation.
\end{proof}

%
%
%
\section{Cases of 2 and 4 speeds}
\label{sec:4-2speeds}
This paper deals with the link between accumulation and the numer of speeds of a signal machine.
We address in this section two cases that provide some bounds on the number of speeds that allow or forbid accumulations:
they can be generated by $4$-speed machines, whereas $2$-speed machines are unable to produce accumulations.
The case of $1$-speed signal machines is not detailed since it is trivial (no collision can occur between signals having the same speed), 
and the border case of $3$-speed machines will be handled in \Sec{sec:3speeds}.
\subsection{Case of 4 speeds}
\label{subsec:4speeds}
The case of $4$ speeds is directly settled by an accumulation example with a $4$-speed machine, 
and we prove here formally that the diagram of \Fig{fig:most-basic-accumulation} in \Sec{sec:introduction} contains an accumulation.

\paragraph{Zeno's paradoxes}
Accumulations can be seen as a variant of the famous Zeno's paradox called the {\em dichotomy paradox}. 
This paradox is a characterization of continuous spaces, in which distances can be divided infinitely in smaller non-zero parts.
Yet, such a distance can be runned in a finite time even though an infinite number of (smaller and smaller) distances have to be runned.
Accumulations can be understood in this paradox meaning: 
an accumulation is the realization of an infinite number of steps ---mainly collisions--- during a finite time.
For instance, the accumulation given below corresponds to an infinite number of back-and-forth (with two collisions at each step) 
between two signals so that the distance between them gets smaller and smaller.

\paragraph{A simple example of accumulation with 4 speeds}
To provide the simple accumulation of \Fig{fig:accu4} with only $4$ distinct speeds, we consider the signal machine \SM[4] defined by \Fig{fig:accu4-rules},
in which collision rules define a bounce of \SigZig (resp. \SigZag) on \SigRight (resp. \SigLeft).

\begin{figure}[hbt]%
  \centering\small
  \subfigcapskip 1em
  \subfigure[Meta-signals and rules of {\SM[4]}.\label{fig:accu4-rules}]{%
    \begin{tabular}[b]{c}
      \begin{tabular}[b]{c|c}
	\bf Meta-signal & \bf Speed\\\hline
	\SigZig   & $4$\rule{0pt}{5mm}\\[.2em]
	\SigLeft  & $1/2$\\[.2em]
	\SigRight & $-1/2$\\[.2em]
	\SigZag   & $-4$\\[.2em]
      \end{tabular}\\[1em]%
      \begin{tabular}[b]{r@{ $\rightarrow$ }l}
	\multicolumn{2}{c}{\bf Collision rules}\\\hline
	\{ \SigLeft, \SigZag \}  & \{ \SigLeft, \SigZig \}\rule{0pt}{5mm}\\[.2em]
	\{ \SigZig, \SigRight \} & \{ \SigZag, \SigRight \}
      \end{tabular}%
    \end{tabular}}\quad%
  \subfigure[Accumulating with $4$ speeds.\label{fig:accu4}]{\qquad%
  \includegraphics[width=0.35\textwidth]{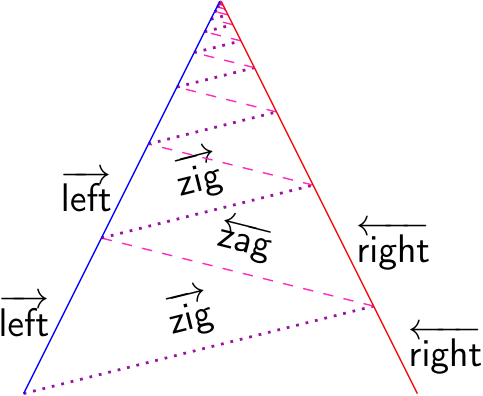}\qquad}%
  \caption{Meta-signals, rules and an accumulation of the machine {\SM[4]}.}\label{fig:accu4-machine}
\end{figure}

This machine allows to generate an accumulation, when started from a well-chosen but very simple initial configuration:

\begin{lemma}\label{lemma:ex-accu4}
  The signal machine \SM[4] generates an accumulation at coordinates $(0, 2)$
  when started from $\Config[0]=\{\ \SigAt{\SigLeft}{-1},\ \SigAt{\SigZig}{-1},\ \SigAt{\SigRight}{1}\ \}$,
\end{lemma}

\begin{proof}
  We consider the sequence $(C_n)_{n\in\N}$ of consecutive collisions on the diagram
  generated by \SM[4] from the configuration \Config[0]
  and we note $(x_n, t_n)$ the corresponding (spatial and temporal) coordinates.
  We also define $\delta_n$ as the duration between two consecutive collisions $C_n$ and $C_{n+1}$, \ie,
  $\delta_n = t_{n+1} - t_n\enspace$.
  
  We compute coordinates of each $C_n$ and show that the sequence of collisions happening during 
  the computation is an alternation of the two collision rules defined previously, \ie, for all $n \geq 0$ we have:
  \[C_n = 
    \begin{cases}
      \ \ColRule{\SigZig,\SigRight}{\SigZag,\SigRight} & \text{ if } n \text{ is odd} \textit{ (right-bounce)}\\[0.3em]
     \ \ColRule{\SigLeft,\SigZag}{\SigLeft,\SigZig} & \text{ if } n \text{ is even} \textit{ (left-bounce)}
    \end{cases}\enspace .
  \]
  We suppose by convention that $C_0$ of coordinate $(x_0, t_0)$ is so that $x_0 = -1$ (initial positions of \SigLeft and \SigZig) and $t_0 = 0$,
  and that $C_0$ is a right-bounce.
  This can be interpreted as a collision happening in the initial configuration and produces the same configuration as \Config[0].
  Note that positions of \SigLeft and \SigRight are always opposite since they have opposite speeds and they are initially
  disposed symetrically to $0$ (collision rules ensure that they keep their motions after each collision).

  Suppose that a configuration at a time $t_n$ is $\Config[t_n]=\{\ \SigAt{\SigLeft}{x_n},\ \SigAt{\SigZig}{x_n},\ \SigAt{\SigRight}{-x_n}\ \}$.
  Any configuration coming from a left-bounce collision $C_n$ of coordinates $(x_n, t_n)$ verifies this displaying of signals.
  In particular, this is the case for the initial configuration \Config[0] with $x_0= -1$.
  It is clear from the disposition of signals that the next collision $C_{n+1}$ to happen is a right-bounce \ie \ColRule{\SigZig,\ \SigRight}{\SigZag,\ \SigRight}
  because \SigZig (of speed $4$) is moving right and is located on the left of \SigRight (of speed $-\frac{1}{2}$) which is moving left.
  Signals \SigLeft and \SigZig cannot collide since \SigLeft is on the left of \SigZig, which moves to the right faster than \SigLeft.
  The collision between \SigZig and \SigRight is deduced from the respective dynamics, and its coordinates $(x,t$) satisfy:
  $\begin{cases}
     x = 4\cdot (t-t_n) + x_n & \text{(motion of \SigZig)}\\
     x = -\frac{1}{2}\cdot (t-t_n) - x_n & \text{(motion of \SigRight)}
  \end{cases}\enspace.$

  At the time $t_{n+1}$, the next collision happens and both signals \SigZig and \SigRight occupy the same 
  position after the duration $\delta_{n}=t_{n+1}-t_n$ which satisfies $4\cdot \delta_{n}+x_n=-\frac{1}{2}\cdot \delta_{n} - x_n$. 
  The delay $\delta_{n}$ and the position $x_{n+1}$ are given by:
  \begin{equation}
    \delta_{n} = -\frac{4}{9}\cdot x_n \quad\text{ and }\quad  x_{n+1}=-\frac{7}{9}\cdot x_n \enspace.\label{eq:space-right-bounce}
  \end{equation}
  Since incoming signals of the collision $C_{n+1}$ have been replaced by outcoming signals,
  the configuration at the time $t_{n+1}$ is given by:
  $\Config[t_{n+1}]=\{\ \SigAt{\SigLeft}{-x_{n+1}},\ \SigAt{\SigZag}{x_{n+1}},\ \SigAt{\SigRight}{x_{n+1}}\ \}$.

  For the symmetric case, when $\Config[t_n]=\{\ \SigAt{\SigLeft}{-x_{n}},\ \SigAt{\SigZig}{x_{n}},\ \SigAt{\SigRight}{x_{n}}\ \}$
  and $C_n$ is a right-bounce, the next collision $C_{n+1}$ is a left-bounce and we obtain by 
  the same way its position $x_{n+1}$ and the duration $\delta_n$ between $C_n$ and $C_{n+1}$:
  \begin{equation}
    \delta_{n} = \frac{4}{9}\cdot x_n \quad \text{and} \quad x_{n+1} = -\frac{7}{9}\cdot x_n \label{eq:left-bounce}
  \end{equation}

  So, starting from the initial configuration \Config[0] and applying successively the previous computations 
  for each configuration after each collision, we obtain an alternation of right-bounces and left-bounces (starting with a right-bounce).
  By \Eq{eq:space-right-bounce} and \Eq{eq:left-bounce}, 
  sequences $(x_n)_{n\in\N}$ and $(\delta_n)_{n\in\N}$ satisfy:\\
  \[
    \left\{
    \begin{array}{l@{}}
      x_0 = -1\ ; \ x_{n+1} = -\frac{7}{9}\cdot x_n \\[0.5em]
      \delta_n \ =\ \left(-1\right)^{n+1}\cdot\frac{4}{9}\cdot x_n
    \end{array}
    \right. \text{ which gives }\forall n\in\N,
    \left\{
    \begin{array}{l@{}}
      x_n = (-1)^{n+1}\cdot\left(\frac{7}{9}\right)^n \\[0.5em]
      \delta_n = \frac{4}{9}\cdot\left(\frac{7}{9}\right)^{n}
    \end{array}
    \right.\enspace.
  \]
  Finally, from the relation $t_{n+1} = t_n + \delta_n$ and $t_0 = 0 $, we obtain for all $n\in\N\,$:
  \[
    t_{n+1}\ = \ \sum \limits_{i=0}^{n} \delta_{i}
      \ =\  \sum \limits_{i=0}^{n} \frac{4}{9} \times \left(\frac{7}{9}\right)^{i}
      \ =\  \frac{4}{9} \times \sum \limits_{i=1}^{n} \left(\frac{7}{9}\right)^{i} \enspace.
  \]
  The sequence $(t_n)_{n\in\N}$ is infinite, strictly increasing and positive.
  It admits a limit which is given by the sum of the geometrical sequence 
  of ratio $\frac{7}{9}~<~1\,$:
  \[
    \lim\limits_{n \to \infty} t_n =\ \frac{4}{9} \times \lim\limits_{n \to \infty}\sum \limits_{i=0}^{n} \left(\frac{7}{9}\right)^{i}%
      \ = \ \frac{4}{9} \times \frac{1}{1-\frac{7}{9}}
      \ = \ 2 \enspace.
  \]
  The sequence $(x_n)_{n\in\N}$ of spatial positions also admits a limit:
  \[
    \lim\limits_{n \to \infty} x_n\ =\ \lim\limits_{n \to \infty} \left|\left(-1\right)^{n}\cdot\left(\frac{7}{9}\right)^{n}\right|
      \ =\ \lim\limits_{n \to \infty}\left(\frac{7}{9}\right)^{n} 
      \ =\ 0\enspace.
  \]
  So the machine \SM[4] runned from initial configuration $c_0$ produces a sequence of successive collisions at coordinates $(x_n, t_n)$ 
  satisfying $\lim\limits_{n\to\infty} (x_n,t_n) = (0,2)$.
  According to \Def{def:dynamics} and \Lem{lemma:accu-dynamics}, we have $\Config[\TLimit](0)=\Config[2](0)=\Accu$, \ie, the point $(0,2)$ is an accumulation.
\end{proof}

Example of the machine \SM[4] gives us the following:

\begin{corollary}\label{coro:accu4}
  Accumulations can be generated by $4$-speed signal machines.
\end{corollary}

\begin{remark}
  As soon as a signal machine contains four meta-signals, all having distinct speeds so that two of them can bounce alternatively between the two other ones, 
  then this signal machine can generate accumulations in kind of the one described above (with proper initial positions for three of these four meta-signals).
  This suggests that the ability to accumulate is not so rare for signal machines having a sufficient number of speed values.
\end{remark}

%
%
%
\subsection{Case of 2 speeds}
\label{subsec:2speeds}

We now consider the case of $2$-speed machines, and prove that no matter what are these two speeds, no accumulation can occur in such machine.

\paragraph{Reduction to a normalized machine}
Let $\SM^{a,b} = (\SigSet, \SpeedFun, \ColSet)$ be a signal machine with meta-signals having only two  distinct speeds $a$, $b$ $\in \R$ with $a<b$.
We define a function $g : \R \to \R $ by $g(x) = \frac{x}{b-a} - \frac{a}{b-a}$.
We have $g(a) = 0$ and $g(b) = 1$ (in fact, $g$ is the function $f$ defined in \Subsec{subsec:properties}, with $c = 0$ and $d = 1$).

We normalize the machine $\SM^{a,b}$ into the machine $\SM^{0,1} = (\SigSet, \SpeedFun', \ColSet)$ where $\SpeedFun'= g\circ\SpeedFun$.
Speed values of $\SM^{0,1}$ are $0$ and $1$. 
As $g$ is an affine function with a strictly positive ratio (because $a<b$) and $\SM^{0,1} = \SM[g]^{a,b}$, we obtain by \Lem{lemma:affine-trans} that $\SM^{0,1}$ produces space-time diagrams equivalent to the ones generated by $\SM^{a,b}$.
In particular, if $\SM^{a,b}$ produces accumulations, then $\SM^{0,1}$ will also produce accumulations.

\paragraph{Accumulating is impossible with only 2 speeds}
We can now give the proof that no signal machine having only two distinct speeds can produce an accumulation.
By \Lem{lemma:affine-trans} and \Cor{lemma:accu-transfert}, it is enough to prove that no accumulation can occur in any support diagrams of the signal machine having only two meta-signals, one of speed $0$ and one of speed $1$.

Hence we can consider, without any loss of generality, a signal machine \SM[2] having 
only two meta-signals: \SigZero of speed $0$ and \SigOne of speed $1$. 
As \SM[2] contains only two meta-signals, there is only one collision rule to be 
defined in the machine and to have \SM[2] to be a support machine, we necessarily have:
$\ColRule{\SigOne,\SigZero}{\SigZero,\SigOne}\enspace.$

\begin{lemma}\label{lemma:max-accu2}
  Let \Config[0] an initial configuration for \SM[2]. 
  Let $i$ be the number of signals \SigOne in \Config[0] and $j$ the number of signals \SigZero.
  Then the diagram generated by \SM[2] starting from \Config[0] contains at most $i\times j$ collisions. 
\end{lemma}

\begin{proof}
  Any finite initial configuration (see \eg \Fig{fig:2speeds_random}) can be rearranged
  in an initial configuration that maximize the number of possible collisions.
  Indeed, the maximum number of collisions is obtained if each signal \SigOne collides
  with each signal \SigZero. This is possible only if each \SigOne is initially disposed at the left of each signal \SigZero. 
  So for computing the exact upper-bound of the number of collisions, we consider an initial configuration similar to the one displayed \Fig{fig:2speeds_max-collision}, \ie, the position of each signal \SigOne is strictly lower than the positions of any signal \SigZero.


  \begin{figure}[hbt]
    \centering
    \subfigure[Diagram from a random initial configuration.\label{fig:2speeds_random}]{\quad%
      \includegraphics[height=4cm]{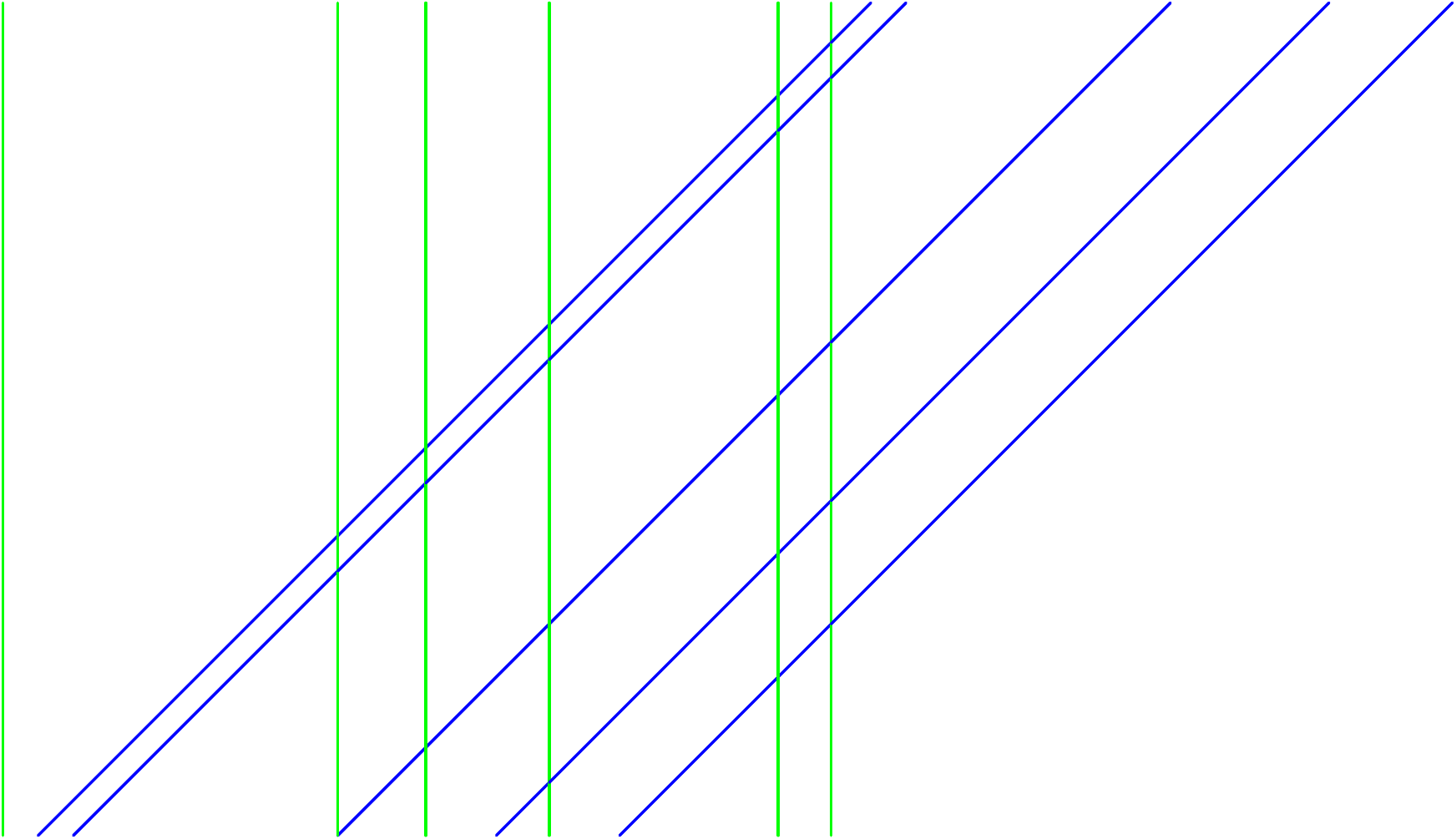}\quad}\qquad%
    \subfigure[Bounding the number of collisions.\label{fig:2speeds_max-collision}]{\quad%
      \includegraphics[height=4cm]{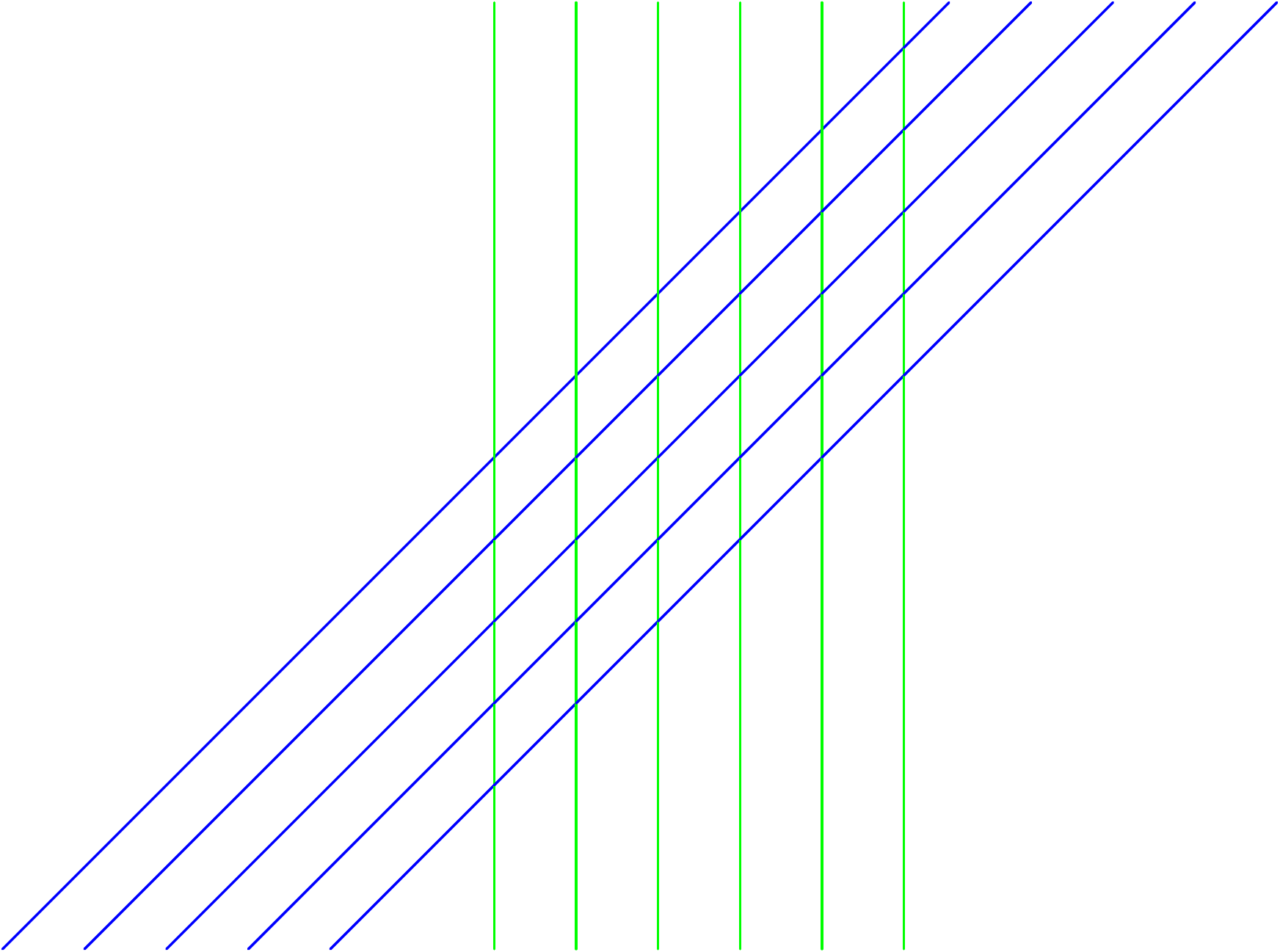}\quad}%
    \caption{Diagrams of the machine \SM[2].}\label{fig:2speeds_no-accu}
  \end{figure}

  Since the unique collision rule outputs all incoming signals, each signal \SigOne is not annihilated after colliding the first (\ie, the left-most) signal \SigZero 
  encountered but \SigOne continues and collides with all the signals \SigZero, that is the $j$ signals \SigZero present in the initial configuration. 
  After colliding the last (\ie, the right-most) signal \SigZero, \SigOne propagates to the right \emph{ad infinitum} and does not collide again.

  Thus each signal \SigOne generates $j$ collisions.
  As the number of signals \SigOne is $i$ and as each one collides with all signals \SigZero, the total number of collisions is $i\times j\enspace$.
\end{proof}

\begin{remark}
  Given $n$ signals in the initial configuration, the maximum number of collisions $i \times j$
  is obtained by taking $i=\lfloor\frac{n}{2}\rfloor$ signals \SigOne and $j=n-i$ signals \SigZero
  in the initial configuration so that each signal \SigOne is initially located at the left of any signal \SigZero.
\end{remark}

\begin{proposition}\label{prop:no-accu2}
  Accumulations cannot be generated by a signal machine (rational-like or not) having only two distinct speeds and starting from a finite configuration.
\end{proposition}

\begin{proof}
  As explained before, any $2$-speed machine \SM can be normalized into a $2$-speed machine $\SM^{0,1}$ (with speeds $0$ and $1$). 
  This machine has for support machine the machine \SM[2].
  By \Lem{lemma:max-accu2}, \SM[2] runned from a finite initial configuration can produce at most a finite number of collisions, 
  and as a necessary (but no sufficient) condition for having accumulation is to contain an infinite number of signals/collisions in the diagram, it follows that \SM[2] cannot generate accumulation.
  \Lemma{lemma:accu-transfert} and \Cor{lemma:affine-trans} imply that neither $\SM^{0,1}$ nor \SM can produce accumulation.
  We conclude that no accumulation can be generated by any $2$-speed machine.
\end{proof}

It has to be underlined that in the case of two signals, the rationality of speeds or initial positions doesn't play any role 
(the normalization of speeds $a$ and $b$ to speeds $0$ and $1$ can be done for any numbers $a$ and $b$, rationals or not, 
and initial positions do not need to be rational in the proof of \Lem{lemma:max-accu2}).
But the hypothesis of finitude of the initial configuration is necessary.

Obviously, it is always possible to create an accumulation with only two speeds if we consider {\em infinite initial configurations}: 
it is enough to place stationary signals such that a (static) accumulation is present in the initial configuration 
\eg an accumulation in $0$ by placing a stationary signal at each position $-1/n$ (this infinite initial configuration remains rational).
Then a right signal \SigOne (having speed $1$) at position $-2$ will collide all stationary signals before time $2$, 
producing a (dynamic) accumulation at coordinates $(0, 2)$, as illustrated by \Fig{fig:accu2-infinite_configuration}
(wihere the rule used is \ColRule{\SigOne,\SigZero}{\SigOne}).

\begin{figure}[hbt]
  \centering
  \includegraphics[width=.3\textwidth]{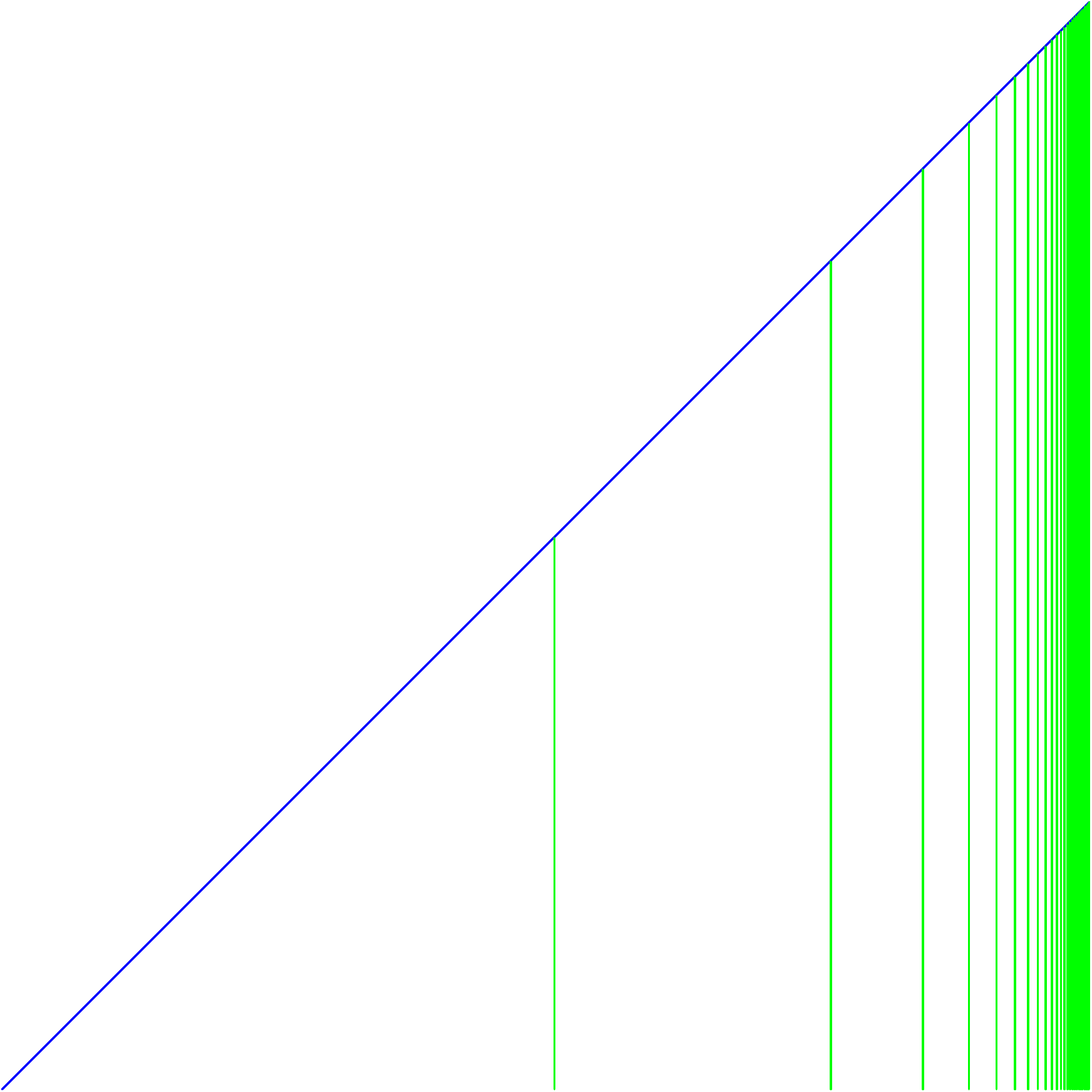}%
  \caption{Accumulation with $2$ speeds and an infinite initial configuration.}\label{fig:accu2-infinite_configuration}
\end{figure}

%
%
%
\section{Case of 3 speeds}
\label{sec:3speeds}
We showed in the previous section that building an accumulation is easy with $4$ speeds, whereas it is impossible with only $2$ distinct speeds. 
In this section, we study the border case ---accumulations with $3$ speeds--- and we show that two sub-cases have to be distinguished:
one treating rational-like signal machines and the other one dealing with fully irrational machines.
The main result is that a $3$-speed machine \SM {\em could} produce accumulations if and only if \SM involves an irrational ratio
(there exists an irrational ratio either between two speed values or between two initial positions).
First, we exhibate a simple signal machine and an irrational configuration from which the machine produces an accumulation:
this is done by implementing a geometrical version of Euclid's algorithm, and by executing an infinite run of this algorithm.
We also provide, with a slight modification of this machine, an example of accumulation from a rational initial configuration
and with $3$ speeds, one of them having an irrational value.
In a second part, we prove that any diagram of a $3$-speed rational machine is included in some {\em regular} diagrams called {\em \StripMultiName[es]};
none of these \StripMultiName[es] contain accumulation, and so neither does any diagram of a $3$-speed rational machine.
\subsection{Case of irrational machines}
\label{sec:3speeds_accu}
%
%
Euclid's algorithm is based on the computation of the remainder of two numbers, and to implement it on signal machines,
we use the geometrical computation of the modulo of two reals values, as mentionned in \Subsec{subsec:examples}:
starting with a configuration that encodes some values $a$ and $b$, 
we can obtain after some number of collisions the value of $a\bmod b$, also encoded between two vertical signals.
We describe below how we can compute the {\em greatest common divisor} of two values by iterating the process of a modulo computation,
and we use the properties of Euclid's algorithm to build an accumulation.

\paragraph{A $3$-speed machine implementing Euclid's algorithm}

To provide an accumulation with three speeds, we define a simple $3$-speed machine, \SMeuclid.
This machine computes the greates common divisor of two values by implementing and following the steps of Euclid's algorithm.

Euclid's algorithm, starting from two real numbers $a$ and $b$ ($b\leq a$),
defines the sequences $(a_n)_{n\in\N}$ and $(b_n)_{n\in\N}$ by the following recursion:
\begin{center}
  \begin{tabular}{c c}
    $\begin{cases}
      a_0 = a \\
      b_0 = b 
    \end{cases}$ 
    \text{ \ et \ } & 
    $\begin{cases}
      a_{n+1} = b_n \\
      b_{n+1} = a_n \bmod b_n
    \end{cases}\enspace.$ 
  \end{tabular}
\end{center}

This recursion also provides the sequence $(r_n)_{n\in\N}$ of remainders of the successive euclidean divisions, given by $r_n = a_n \bmod b_n$,
and the sequence $(q_n)_{n\in\N}$ of the quotients, defined by $q_n=a_b/b_n$.
If the sequence $(a_n)_{n\in\N}$ becomes equal to zero from one rank, then the greatest $a_n$ so that $a_n \neq 0$ is the {\em greatest common divisor} ($\gcd$) of $a$ and $b$.
Otherwise, the $\gcd$ of $a$ and $b$ is not defined.

This algorithm can be geometrically implemented by a $3$-speed machine composed by seven meta-signals (three of speed $0$, two of speed $1$ and two of speed $-1$) and eight rules.
We give in \Fig{fig:SMeuclid-rules} the definition of such a machine \SMeuclid, and a run example is displayed in \Fig{fig:euclid-run_example}, 
which corresponds to the computation of $\gcd(8,3)$ 
(coded by the distance between the two remaining signals \SigEuclidWallZero at the very top of the diagram).

\newcounter{rulesnumbersgcd}
\setcounter{rulesnumbersgcd}{0}
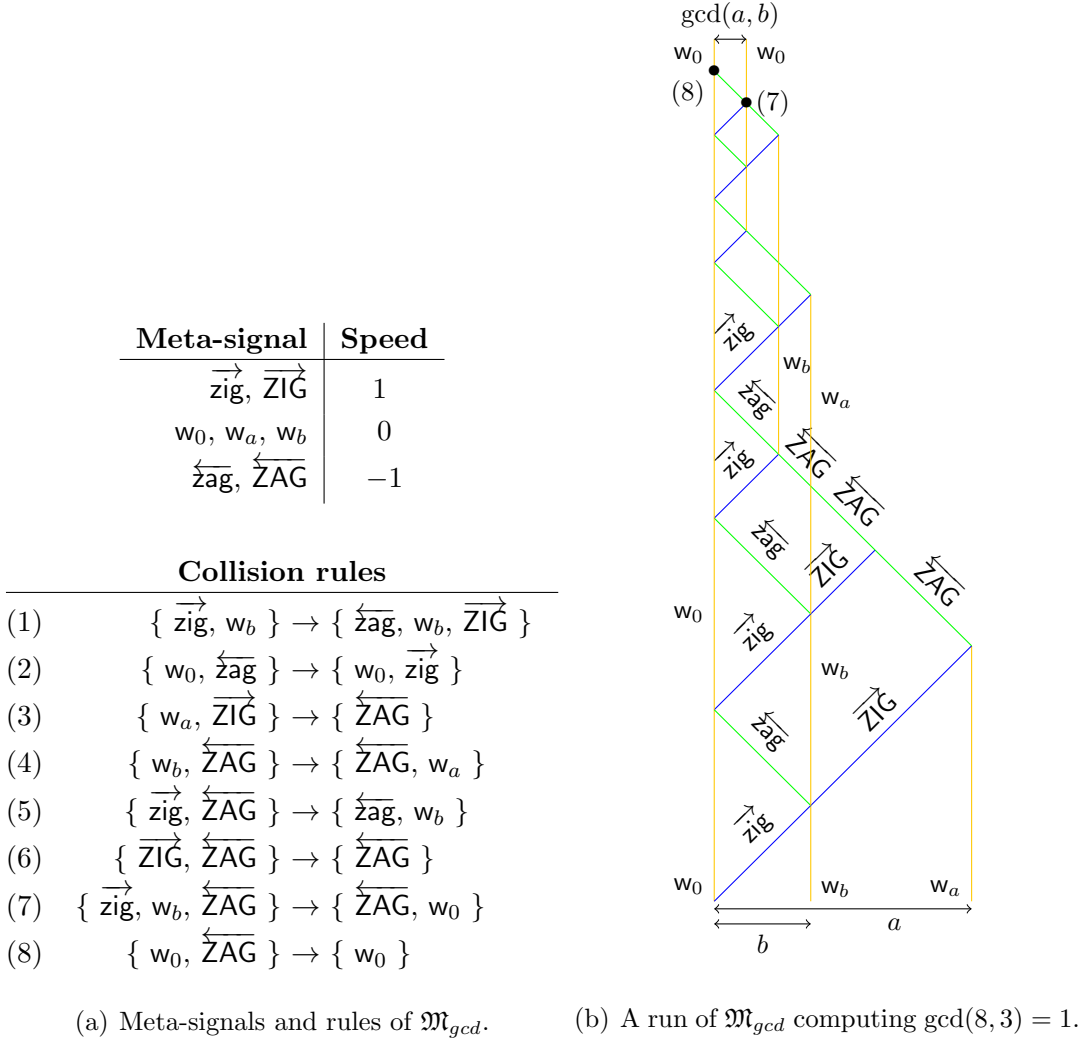
\begin{figure}[hbt]%
  \centering\small
  \subfigcapskip 1em
  \subfigure[Meta-signals and rules of \SMeuclid.\label{fig:SMeuclid-rules}]{%
    \begin{tabular}[b]{c}
      \begin{tabular}[b]{r|c}
	\bf Meta-signal & \bf Speed \\
	\hline
	\SigEuclidZigSmall, \SigEuclidZigBig & $1$ \rule{0pt}{5mm}\\[.2em]
	\SigEuclidWallZero, \SigEuclidWalla, \SigEuclidWallb & $0$\\[.2em]
	\SigEuclidZagSmall, \SigEuclidZagBig & $-1$\\[.2em]
      \end{tabular}\\[2em]%
      \begin{tabular}[b]{@{\scriptsize}>{(\refstepcounter{rulesnumbersgcd}\therulesnumbersgcd)}l r@{ $\rightarrow$ }l}
	\multicolumn{3}{c}{\bf Collision rules}\\\hline
	\label{rules:SMeuclid_1} & \{ \SigEuclidZigSmall, \SigEuclidWallb \} & \{ \SigEuclidZagSmall, \SigEuclidWallb, \SigEuclidZigBig \} \rule{0pt}{5mm}\\[.2em]
	\label{rules:SMeuclid_2} & \{ \SigEuclidWallZero , \SigEuclidZagSmall \} & \{ \SigEuclidWallZero , \SigEuclidZigSmall \} \\[.2em]
	\label{rules:SMeuclid_3} & \{ \SigEuclidWalla , \SigEuclidZigBig \} & \{ \SigEuclidZagBig \} \\[.2em]
	\label{rules:SMeuclid_4} & \{ \SigEuclidWallb , \SigEuclidZagBig \} & \{ \SigEuclidZagBig, \SigEuclidWalla \}\\[.2em]
	\label{rules:SMeuclid_5} & \{ \SigEuclidZigSmall , \SigEuclidZagBig \} & \{ \SigEuclidZagSmall, \SigEuclidWallb \}\\[.2em]
	\label{rules:SMeuclid_6} & \{ \SigEuclidZigBig, \SigEuclidZagBig \} & \{ \SigEuclidZagBig \}  \\[.2em]
	\label{rules:SMeuclid_7} & \{ \SigEuclidZigSmall, \SigEuclidWallb, \SigEuclidZagBig \} & \{ \SigEuclidZagBig, \SigEuclidWallZero\} \\[.2em]
	\label{rules:SMeuclid_8} & \{ \SigEuclidWallZero, \SigEuclidZagBig \} & \{ \SigEuclidWallZero\}\\[.2em]%
      \end{tabular}%
    \end{tabular}}%
  \subfigure[A run of \SMeuclid computing $\gcd(8,3) = 1$.\label{fig:euclid-run_example}]{\quad\qquad%
    \begin{tikzpicture}[x=1.1em,y=1.1em,font=\footnotesize]
\definecolor{couleurwall0}{RGB}{255,200,0}
\tikzstyle{drawwall0}=[draw=couleurwall0]
\tikzstyle{nodewall0}=[pos=0.5, left]
\definecolor{couleurwallx}{RGB}{255,200,0}
\tikzstyle{drawwallx}=[draw=couleurwallx]
\tikzstyle{nodewallx}=[pos=0.7,right=-0.2em]
\definecolor{couleurwally}{RGB}{255,200,0}
\tikzstyle{drawwally}=[draw=couleurwally]
\tikzstyle{nodewally}=[pos=0.45, right]
\definecolor{couleurzig}{RGB}{0,0,255}
\tikzstyle{drawzig}=[draw=couleurzig]
\tikzstyle{nodezig}=[pos=0.6, above, sloped]
\definecolor{couleurZIG}{RGB}{0,0,255}
\tikzstyle{drawZIG}=[draw=couleurZIG]
\tikzstyle{nodeZIG}=[pos=0.5, above, sloped]
\definecolor{couleurzag}{RGB}{0,255,0}
\tikzstyle{drawzag}=[draw=couleurzag]
\tikzstyle{nodezag}=[pos=0.6, above, sloped]
\definecolor{couleurZAG}{RGB}{0,255,0}
\tikzstyle{drawZAG}=[draw=couleurZAG]
\tikzstyle{nodeZAG}=[pos=0.5, above, sloped]
\draw[drawwall0] (0.000000,0.000000) -- node[left,pos=0.08]{\SigEuclidWallZero} (0.000000,6.000000);s
\draw[drawwallx] (3.000000,0.000000) -- node[right,pos=0.13]{\SigEuclidWallb} (3.000000,3.000000);
\draw[drawwally] (8.000000,0.000000) -- node[left,pos=0.05]{\SigEuclidWalla} (8.000000,8.000000);
\draw[drawzig] (0.000000,0.000000) -- node[nodezig]{ \SigEuclidZigSmall} (3.000000,3.000000);
\draw[drawzag] (3.000000,3.000000) -- node[nodezag]{ \SigEuclidZagSmall} (0.000000,6.000000);
\draw[drawZIG] (3.000000,3.000000) -- node[nodeZIG]{ \SigEuclidZigBig} (8.000000,8.000000);
\draw[drawwallx] (3.000000,3.000000) -- node[right,pos=0.7]{ \SigEuclidWallb} (3.000000,9.000000);
\draw[drawzig] (0.000000,6.000000) -- node[nodezig]{ \SigEuclidZigSmall} (3.000000,9.000000);
\draw[drawZAG] (8.000000,8.000000) -- node[nodeZAG]{ \SigEuclidZagBig} (5.000000,11.000000);
\draw[drawZIG] (3.000000,9.000000) -- node[nodeZIG]{ \SigEuclidZigBig} (5.000000,11.000000);
\draw[drawzag] (3.000000,9.000000) -- node[nodezag]{ \SigEuclidZagSmall} (0.000000,12.000000);
\draw[drawwall0] (0.000000,6.000000) -- node[nodewall0]{ \SigEuclidWallZero} (0.000000,12.000000);
\draw[drawZAG] (5.000000,11.000000) -- node[nodeZAG]{ \SigEuclidZagBig} (3.000000,13.000000);
\draw[drawwallx] (3.000000,9.000000) -- (3.000000,13.000000);
\draw[drawZAG] (3.000000,13.000000) -- node[nodeZAG]{ \SigEuclidZagBig} (2.000000,14.000000);
\draw[drawzig] (0.000000,12.000000) -- node[nodezig]{ \SigEuclidZigSmall} (2.000000,14.000000);
\draw[drawzag] (2.000000,14.000000) -- node[nodezag]{ \SigEuclidZagSmall} (0.000000,16.000000);
\draw[drawwall0] (0.000000,12.000000) --  (0.000000,16.000000);
\draw[drawwallx] (2.000000,14.000000) -- node[nodewallx]{ \SigEuclidWallb} (2.000000,18.000000);
\draw[drawzig] (0.000000,16.000000) -- node[nodezig]{ \SigEuclidZigSmall} (2.000000,18.000000);
\draw[drawwally] (3.000000,13.000000) -- node[nodewally]{ \SigEuclidWalla} (3.000000,19.000000);
\draw[drawZIG] (2.000000,18.000000) --  (3.000000,19.000000);
\draw[drawZAG] (3.000000,19.000000) -- (2.000000,20.000000);
\draw[drawwallx] (2.000000,18.000000) -- (2.000000,20.000000);
\draw[drawzag] (2.000000,18.000000) -- (0.000000,20.000000);
\draw[drawwall0] (0.000000,16.000000) -- (0.000000,20.000000);
\draw[drawZAG] (2.000000,20.000000) -- (1.000000,21.000000);
\draw[drawzig] (0.000000,20.000000) -- (1.000000,21.000000);
\draw[drawzag] (1.000000,21.000000) -- (0.000000,22.000000);
\draw[drawwall0] (0.000000,20.000000) -- (0.000000,22.000000);
\draw[drawwallx] (1.000000,21.000000) -- (1.000000,23.000000);
\draw[drawzig] (0.000000,22.000000) -- (1.000000,23.000000);
\draw[drawwally] (2.000000,20.000000) -- (2.000000,24.000000);
\draw[drawZIG] (1.000000,23.000000) -- (2.000000,24.000000);
\draw[drawzag] (1.000000,23.000000) -- (0.000000,24.000000);
\draw[drawwall0] (0.000000,22.000000) -- (0.000000,24.000000);
\draw[drawZAG] (2.000000,24.000000) -- (1.000000,25.000000);
\draw[drawwallx] (1.000000,23.000000) -- (1.000000,25.000000);
\draw[drawzig] (0.000000,24.000000) -- (1.000000,25.000000);
\draw[drawZAG] (1.000000,25.000000) -- (0.000000,26.000000);
\draw[drawwall0] (0.000000,24.000000) -- (0.000000,26.000000);
\draw[drawwall0] (0.000000,26.000000) -- node[left,pos=.5]{ \SigEuclidWallZero} (0.000000,27);
\draw[drawwall0] (1.000000,25.000000) -- node[right,pos=.75]{ \SigEuclidWallZero} (1.000000,27);
\draw[<->] (0,-\Sep) --node[below,pos=.7]{$a$} (8,-\Sep) ;
\draw[<->,yshift=-.2cm] (0,-\Sep) --node[below]{$b$} (3,-\Sep) ;
\draw[<->] (0,27) --node[above]{$\gcd(a,b)$} (1,27) ;
\draw (1.,25) node {$\bullet$}; \draw (1.,25) node[right,inner sep=.3em]{\footnotesize $(\ref{rules:SMeuclid_7})$} ;
\draw (0,26) node {$\bullet$}; \draw (0.,26) node[below left,inner sep=.3em]{\footnotesize $(\ref{rules:SMeuclid_8})$} ;
\end{tikzpicture}\quad\qquad}%
  \caption{Meta-signals, rules and a run of the machine \SMeuclid.}\label{fig:accu3-machine}
\end{figure}

As done in \Subsec{subsec:examples}, the stationary meta-signals (\SigEuclidWallZero, \SigEuclidWalla and \SigEuclidWallb) are used to encode two real numbers: 
the real value $a$ (resp. $b$) is the distance between signals \SigEuclidWallZero and \SigEuclidWalla (resp. \SigEuclidWallb).
In our geometrical version, the step $a_{n+1} = b_n$ is implemented by the rule (\ref{rules:SMeuclid_4}) of \Fig{fig:SMeuclid-rules}
and the step $b_{n+1} = a_n \bmod b_n$ by the rule (\ref{rules:SMeuclid_5}).
Rules (\ref{rules:SMeuclid_7}) and (\ref{rules:SMeuclid_8}) correspond to the two last collisions when the process halts, as shown in the top of \Fig{fig:euclid-run_example}.

We denote by $\Config^{a,b}$ ($a > b$) the configuration using only four signals (including one non-stationary signal) 
and respecting the previous encoding of the values between the stationary signals:
$\Config^{a,b} = \{\ \SigAt{\SigEuclidWallZero}{0},\ \SigAt{\SigEuclidZigSmall}{0},\ \SigAt{\SigEuclidWallb}{b},\ \SigAt{\SigEuclidWalla}{a}\ \}$.
The distance between \SigEuclidWallZero and \SigEuclidWalla (resp. \SigEuclidWallb) is indeed $a$ (resp. $b$).

\begin{remark}
  Starting from a configuration having two stationary signals and one moving (say to the right) from the first stationary signal to the second
  the time of a back-and-forth depends on the distance $d$ between the stationary signals and the speeds of the signals making the bounce. 
  This is the cases here with a configuration $\Config^{a,b}$, where \SigEuclidZigSmall and \SigEuclidZagSmall will make a back-and-forth between \SigEuclidWallZero and \SigEuclidWallb 
  (resp. \SigEuclidZigBig and \SigEuclidZagBig making a back-and-forth between \SigEuclidWallZero and \SigEuclidWalla).
  This time is equal to the sum of the time for the right and left signal to reach the opposite wall and is given by $(1+1)\cdot d= 2d$ (remember the speeds are $1$ and $-1$).
  Note that in the general case of non-null speeds $\nu_1$ and $-\nu_2$, the time of a back-and-forth is given by 
  $\frac{1}{\nu_1}d+\frac{1}{\nu_2}d=\left(\frac{1}{\nu_1}+ \frac{1}{\nu_2}\right)\cdot d=\uptau\cdot d$,
  where $\uptau=\frac{1}{\nu_1}+ \frac{1}{\nu_2}$ is the time a of unitary back-and-forth.
\end{remark}

Let us describe briefly the evolution of \SMeuclid on such a configuration. 
Consider first the case $a\bmod b \neq 0$.
After signals \SigEuclidZigBig (firstly \SigEuclidZigSmall) and \SigEuclidZagBig (firstly \SigEuclidZagSmall) 
have completed the back-and-forth between \SigEuclidWallZero and \SigEuclidWalla, 
that is after the duration $\uptau a$ according to the previous remark, 
the configuration has the same form than the initial configuration $\Config^{a,b}$, 
but now with the two walls \SigEuclidWalla and \SigEuclidWallb that have been moved to new positions. 
Indeed, with respect to the rules, the initial \SigEuclidWalla has now disappeared, the initial \SigEuclidWallb has be turned into
\SigEuclidWalla by \SigEuclidZagBig and the collision between \SigEuclidZigSmall 
(which was bouncing between signals \SigEuclidWallZero and \SigEuclidWallb, being turned alternatively into \SigEuclidZigSmall and \SigEuclidZagSmall)
and \SigEuclidZagBig has created a new \SigEuclidWallb, the first wall \SigEuclidWallZero remaining unchanged.
This process is displayed in \Fig{fig:accu3-euclid}.

If we call $a'$ (resp. $b'$) the distance between \SigEuclidWallZero and the new \SigEuclidWalla 
(resp. \SigEuclidWallb) as illustrated by \Fig{fig:accu3-euclid}, the new configuration at time 
$2 a$ is $\Config^{a',b'}=\{\SigEuclidWallZero@0,\ \SigEuclidZigSmall@0,\ \SigEuclidWallb@b',\ \SigEuclidWalla@a'\}$.
We have: $2 a = k \cdot 2 b + 2 r$,
where $k\in\N$ and $0\leq 2 r < 2 b$ (by definition of $2 r$ which is defined from the wall appearing between walls encoding $b$).
So $2 r$ is the remainder in the Euclidean division of $2 a$ by $2 b$:
$2 r = 2 a\bmod 2 b$, that is $r=a \bmod b$.
We have $a'=b$ and $b'=r=a\bmod b\neq 0$.
Finally, from the configuration $\Config^{a,b}$, we obtain after the duration $2 a$ the configuration
$\Config^{a',b'}=\Config^{b,a\bmod b}$.

For the case $r=0$ (\ie when $b$ divides $a$), 
the collision involving \SigEuclidZagBig and \SigEuclidWallb also involves \SigEuclidZigSmall and the rule (\ref{rules:SMeuclid_7}) is applied. 
The next collision necessarily happens between \SigEuclidWallZero and \SigEuclidZagBig, and after the application of the  rule (\ref{rules:SMeuclid_8}) 
only two signals \SigEuclidWallZero, spaced by the distance $b$ remains on the space.

By iterating the process and starting with $a=a_0$ and $b=b_0$, 
we can define the sequences $(a_n)_{n\in\N}$ and $(b_n)_{n\in\N}$ by identifying at each step of the process 
the distances $a,a',b$ and $b'$ respectively with the values $a_n, a_{n+1}, b_n$ and $b_{n+1}$.
The sequences thus obtained correspond to the ones defined by the double recursion of Euclid's algorithm with
initial conditions $a_0 = a$ and $b_0=b$: we have $a_{n+1} = b_n$ and $b_{n+1} = a_n \bmod b_n$.
If for some $n_0$, $a_{n_0}$ is a multiple of $b_{n_0}$, it means that \SigEuclidZagBig and \SigEuclidWallb will also collide with
\SigEuclidZigSmall into a triple collision. 
As mentionned above, rules (\ref{rules:SMeuclid_7}) and (\ref{rules:SMeuclid_8}) bring the process to a halt, leaving two signals \SigEuclidWallZero:
the distance between them is the last non-null remainder and so, it is the $\gcd$ of $a_0$ and $b_0$.

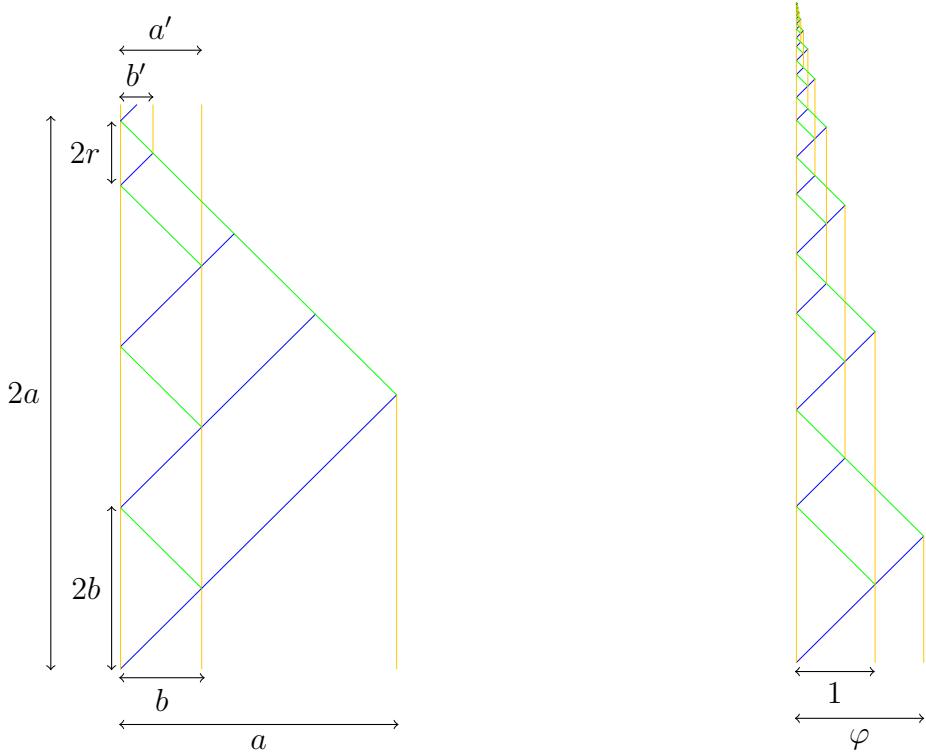
\begin{figure}[hbt]
  \Height 7.5cm
  \Width 0.5\Height
  \Sep 1mm 
  \ArrowSpacing 9mm 
  \centering
  \subfigure[A step of the geometrical Euclid's algorithm.\label{fig:accu3-euclid}]{\qquad\quad%
    \begin{tikzpicture}[x=\Width,y=\Height]
      \draw (0,0) node [above right,inner sep=0]{\includegraphics[height=\Height]{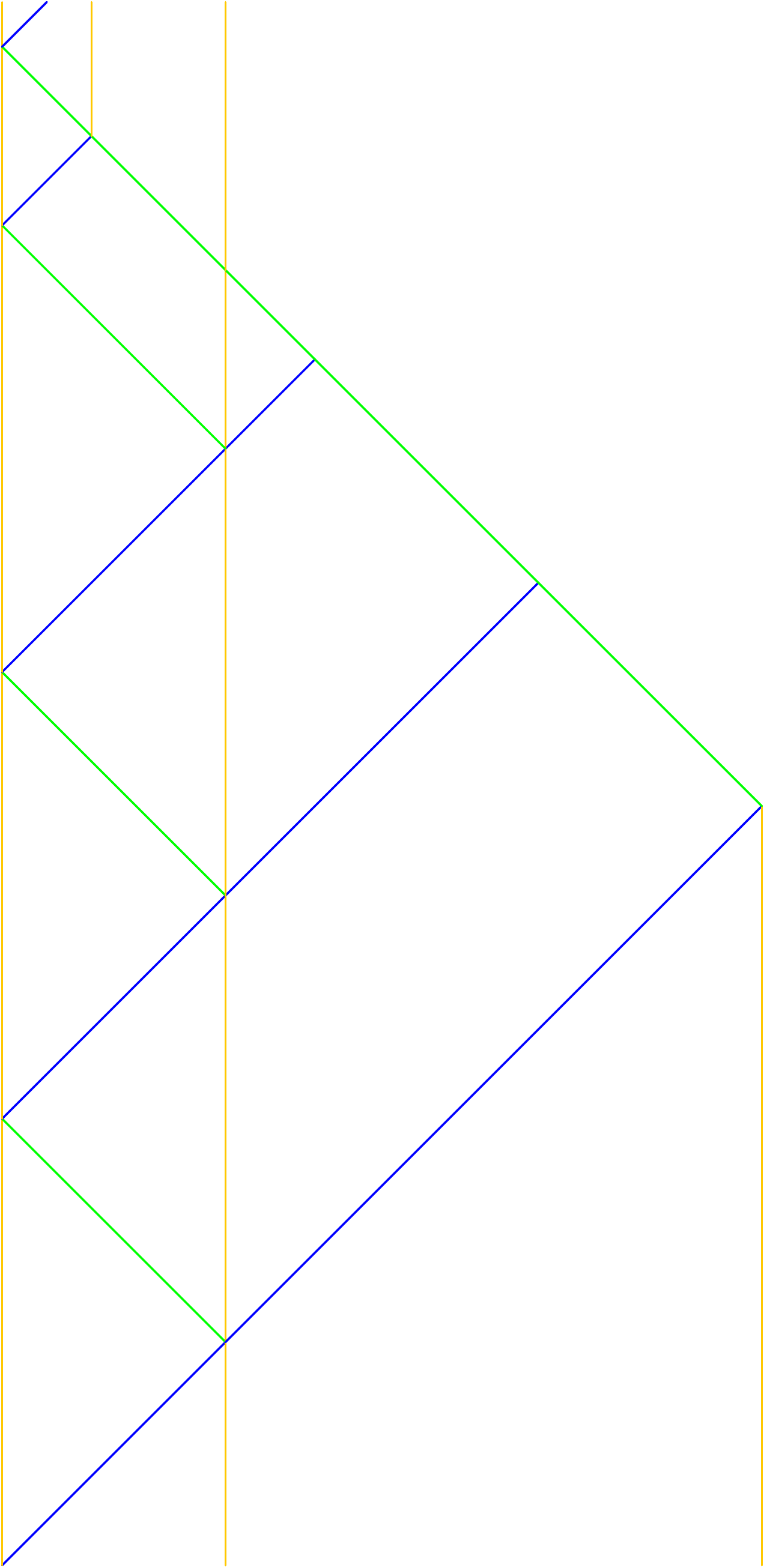}};
      \draw[<->][arrow] (-\Sep,0) -- node[left] {$2 b$} (-\Sep,0.29);
      \draw[<->][arrow] (-\ArrowSpacing,0) -- node[left] {$2 a$} (-\ArrowSpacing,.98);
      \draw[<->][arrow] (-\Sep,0.86) -- node[left] {$2 r$} (-\Sep,.97);
      \draw[<->][arrow] (0,-0.8\ArrowSpacing) -- node[below] {$a$} (.98\Width,-0.8\ArrowSpacing);
      \draw[<->][arrow] (0,-\Sep) -- node[below] {$b$} (0.3,-\Sep);
      \draw[<->][arrow,yshift=\Height] (0,\Sep) -- node[above] {$b'$} (0.12,\Sep);
      \draw[<->][arrow,yshift=\Height] (0,0.8\ArrowSpacing) -- node[above] {$a'$} (0.29,0.8\ArrowSpacing);
    \end{tikzpicture}\qquad\quad}%
  \hfill\Width 1.7cm
  \subfigure[An accumulation with $3$ speeds.\label{fig:accu3-phi}]{\qquad\qquad%
    \begin{tikzpicture}[x=\Width]
      \draw (0,0) node [above right,inner sep=0]{\includegraphics[width=\Width]{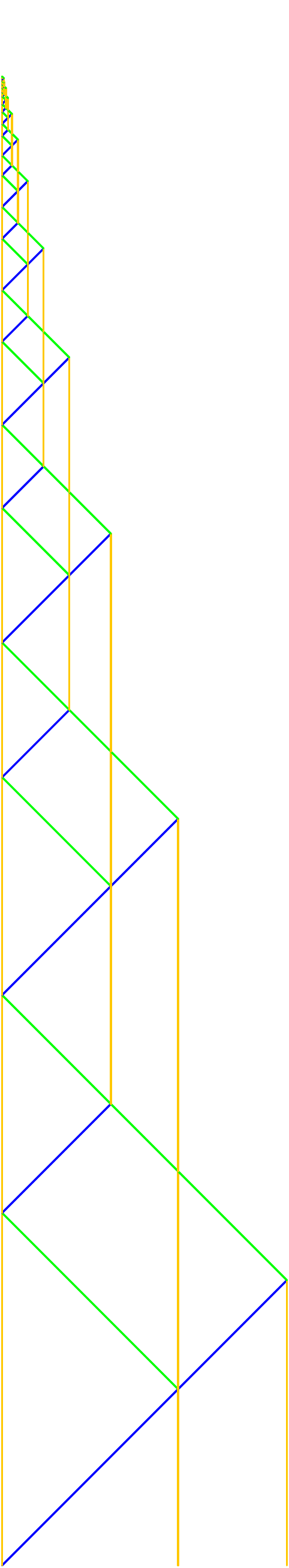}};
      \draw[<->][arrow] (0,-\Sep) -- node[below] {$1$} (0.618,-\Sep);
      \draw[<->][arrow] (0,-.8\ArrowSpacing) -- node[below] {$\varphi$} (1,-.8\ArrowSpacing);
    \end{tikzpicture}\qquad\qquad\qquad}%
  \caption{Achieving an accumulation with the machine \SMeuclid.}\label{fig:accu3}
\end{figure}

\paragraph{Non-termination of Euclid's algorithm}

To provide an accumulation, an infinite number of collisions should take place in a finite time, and so this process should be infinite. 
This means that we need to execute Euclid's algorithm on some values $a$, and $b$ such that it doesn't end.
To achieve this, it is enough to use two {\em incommensurate reals} \ie, two reals such that their ratio is irrational. 
Indeed, it is know from \cite[Th.~161, p.~136]{hardy+wright60book} that:
\[\text{{\em Euclid's algorithm halts on inputs $a$ and $b$}}\ \ \Longleftrightarrow\ \ \frac{a}{b}\in\Q\enspace.\]

To obtain an infinite run, we set the values $a=\varphi$ and $b=1$, where $\varphi=\frac{1+\sqrt 5 }{2}$ is the golden number.
The generated diagram is given by \Fig{fig:accu3-phi}, and
the corresponding initial configuration is
$\Config[0]^{\varphi,1} = \{\SigEuclidWallZero@0,\ \SigEuclidZigSmall@0,\ \SigEuclidWallb@1,\ \SigEuclidWalla@\varphi\}$.

Executed on these inputs, the algorithm doesn't terminate because $\frac{a}{b} = \varphi$ is irrational, 
and it produces a sequence $(r_n)_{n\in\N}$ of remainders wich is striclty positive and decreasing (by definition of the Euclidean division).
Also, as $\varphi$ satisfies $\varphi = 1 + \frac{1}{\varphi}$, the developement of $\varphi$ into a continued fraction is:
\[ \varphi = 1+ \frac{1}{1+\frac{1}{1+\frac{1}{\ldots}}}\]
and it provides that $q_n = 1$ for all $n$ (see \cite[p.~134]{hardy+wright60book} for more details on the link between continued fraction and Euclid's algorithm).

\paragraph{Convergence of sum of $a_n$ and limit of the sequence of collision times}

We prove here that the diagram of \Fig{fig:accu3-phi} indeed contains an accumulation.
Starting from $a_0=a$, $b_0=b$, by using the definition of $a_{n+1} = b_n$ and $b_{n+1} = r_n$ and the relation $a_n = b_n \cdot q_n + r_n $ for all $n\in\N$, 
we can prove in the general case that for all $n$ we have $a_{n+2} \leq \frac{1}{2}a_n$.\\
Indeed,  we have:
\begin{align*}
a_n &= b_n \cdot q_n + r_n \\
    &= a_{n+1} \cdot q_n + b_{n+1} &&\text{ since } a_{n+1} = b_n \text{ and } b_{n+1} = r_n\,, \\
    &= (b_{n+1}\cdot q_{n+1} + r_{n+1})\cdot q_n + b_{n+1} \\
    &= b_{n+1} (1 + q_{n+1}\cdot q_n) + r_{n+1}\cdot q_n \\
    &= a_{n+2} (1 + q_{n+1}\cdot q_n) + r_{n+1}\cdot q_n &&\text{ since } a_{n+2} = b_{n+1}\,, \\
    &\geq 2 a_{n+2} + r_{n+1} &&\text{ because }\forall n\ q_n\geq 1\,.   
\end{align*}
We know that starting from $a_0 = \varphi$ and $b_0 = 1$, since the ratio $a_0/b_0$ is irrational, 
the algorithm doesn't terminate and so we have $r_{n} > 0$ for all $n$. We also have for all $n$ $q_n =1$. 
So in the case of inputs  $\varphi$ and $1$, the previous inequality simply becomes:
\begin{equation*}
  a_{n+2} < \frac{1}{2} a_n\enspace.
\end{equation*}
By a simple induction on $n$, we obtain $\forall n\geq 1$, $a_{2n} < \frac{1}{2^n}a_0 $ and $a_{2n+1}<\frac{1}{2^n}a_1$.
Finally:
\begin{align*}
\sum \limits_{i=0}^{n} a_i &=\ \sum \limits_{i=0}^{\lfloor n/2\rfloor} a_{2i} + \sum \limits_{i=0}^{\lfloor n/2\rfloor} a_{2i+1} \\
			     &<\ \sum \limits_{i=0}^{\lfloor n/2\rfloor} \frac{1}{2^i}\cdot a_0 + \sum \limits_{i=0}^{\lfloor n/2\rfloor} \frac{1}{2^i}\cdot a_1\\
			     &<\ a_0 \cdot \sum \limits_{i=0}^{\lfloor n/2\rfloor} \frac{1}{2^i} + a_1 \cdot \sum \limits_{i=0}^{\lfloor n/2\rfloor} \frac{1}{2^i}
\end{align*}
and when $n \to \infty$, we get for the values $a_0=\varphi$ and $a_1=b_0=1$:
\begin{align*}
  \lim\limits_{n \to \infty} \sum \limits_{i=0}^{n} a_i &<\ 2\cdot(\varphi+1)
\end{align*}

Since the series $\Sigma a_i$ is upper-bounded by $2\cdot(\varphi+1)$, and has positive terms, the infinite sum $\sum \limits_{i=0}^{\infty} a_i$ converges to a finite limit.

Each time a recursive step of the algorithm starts, at least one collision has occured (in fact at least $3$), 
and so each duration $2 a_i$ contains at least one collision 
(remember that the value $2$ is the coefficient of a back-and-forth duration, and it only depends on the speeds, which values are here $1$ and $-1$).
More precisely, for all $n\in\N$, there is a collision occuring at coordinates $(0, 2\cdot\sum \limits_{i=0}^{n} a_i)$,
and between time $2\cdot\sum \limits_{i=0}^{n} a_i$ and $2\cdot\sum \limits_{i=0}^{n+1} a_i$, there is only a finite number of collisions.

The total sum of these durations is given by $\TLimit=\sum \limits_{i=0}^{\infty} 2 a_i$ and corresponds to the total height of the construction.
As we showed previsouly, the infinite sum  $\sum \limits_{i=0}^{\infty} a_i$ is bounded by $2\cdot(\varphi+1)$ and so \TLimit is upper-bounded by $4\cdot(\varphi+1)$.
We can conclude that before the finite time \TLimit, there is an infinite number of collisions,
because there is at least one collision occuring at each time $t=2\cdot\sum \limits_{i=0}^{n} a_i < \TLimit$.
This implies by definition of the sequence of collisions $(C_n)_{n\in\N}$ (\Def{def:dynamics}), that the sequence of collision times is converging to \TLimit.
In the general case, \TLimit is bounded by $2\uptau\cdot(a_0+b_0)$, where $\uptau=\frac{1}{\nu_1}+\frac{1}{\nu_2}$ (for non-null speeds $\nu_1$ and $\nu_2$)
and $a_0$ and $b_0$ are positive real numbers.

By \Lem{lemma:accu-dynamics}, the diagram of \SMeuclid started on the initial configuration $\Config[0]^{1,\varphi}$ contains an accumulation happening
at coordinates $(0,\TLimit)$ \ie it satisfies $\Config[\TLimit](0)=\Accu$.
More generally, by using the property that Euclid's algorithm halts on inputs $a$ and $b$ if and only if $\frac{a}{b}$ is irrational, 
we obtain the following:

\begin{lemma}\label{lemma:accu3-euclid_accumulation}
  The machine \SMeuclid produces an accumulation when executed from any initial configuration $\Config[0]^{a,b}$ satisfying $\frac{a}{b}\notin\Q$.
\end{lemma}

\begin{remark}
  According to this lemma, $\varphi$ can be replaced by any other irrational value $x$ (greater than $1$): 
  the machine \SMeuclid run from $\Config[0]^{1,x}$ will also produce an accumulation. Here, $\varphi$ has been chosen because of the regularity of 
  the resulted diagram, coming from the fact that Euclid's algorithm started on $\varphi$ and $1$ satisfies $q_n = 1$ for all $n$.
  
  We can also provide an accumulation by using \Supp{\SM[3]^{-1,0,1}}, the support machine of \SMeuclid, using only three meta-signals,
  and by running it on \Supp{\Config[0]^{1,\varphi}}, the support configuration of $\Config[0]^{1,\varphi}$. 
  We obtain the diagram of \Fig{fig:accu3-full_phi}, in which an accumulation happens at position $x=1$ and strictly 
  before the time \TLimit computed previsouly. Another example is displayed by \Fig{fig:accu3-full_other}.
  The fact that these support diagrams contain an accumulation directly follow by \Cor{lemma:accu-transfert} and \Lem{lemma:accu3-euclid_accumulation}.
  \begin{figure}[hbt]
    \centering
    \subfigure[Support diagram for {$\Config[0]^{1,\varphi}$}.\label{fig:accu3-full_phi}]{%
      \includegraphics[height=4.3cm]{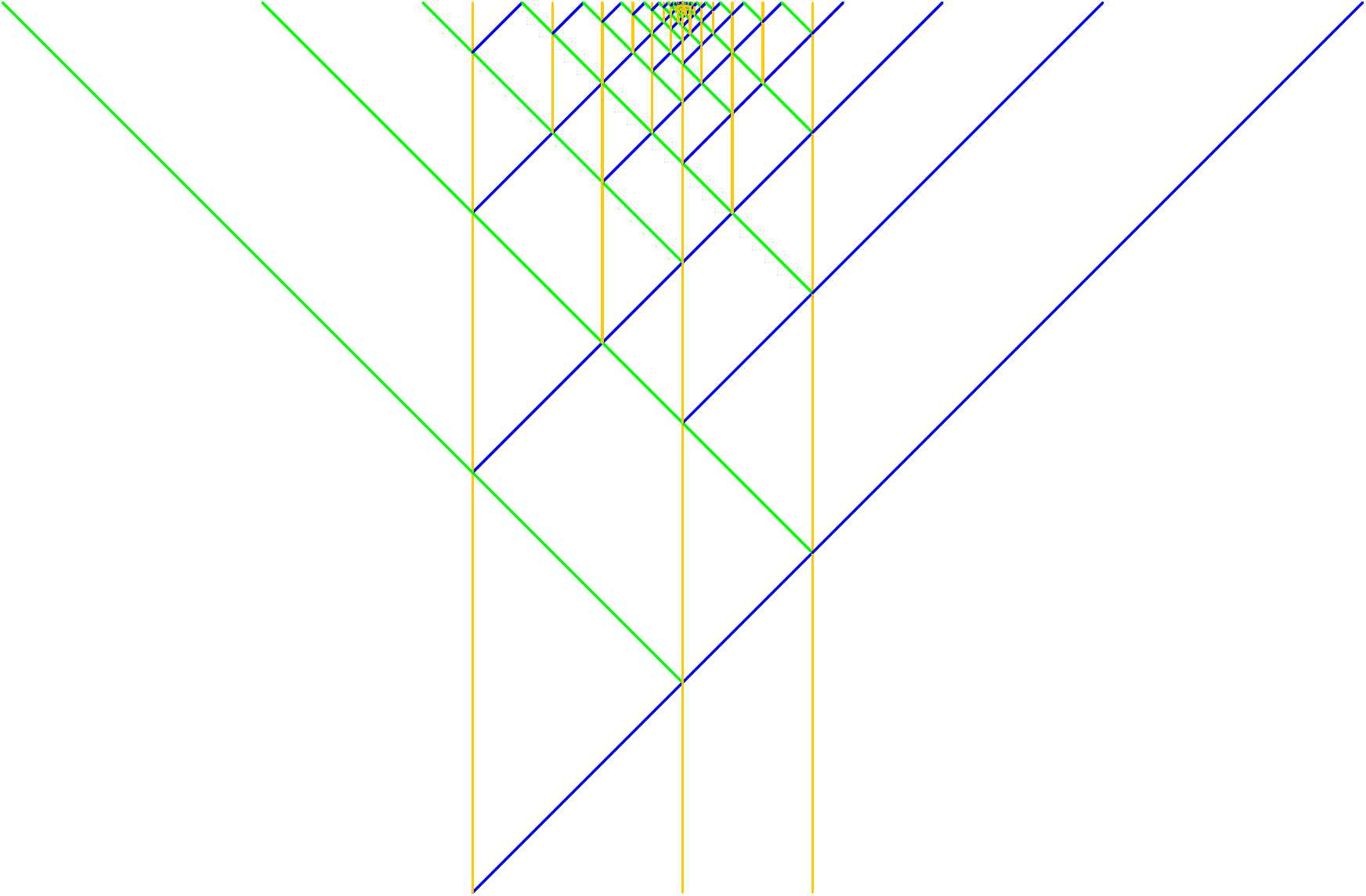}\quad}
    \subfigure[Another example.\label{fig:accu3-full_other}]{%
      \includegraphics[height=4.3cm]{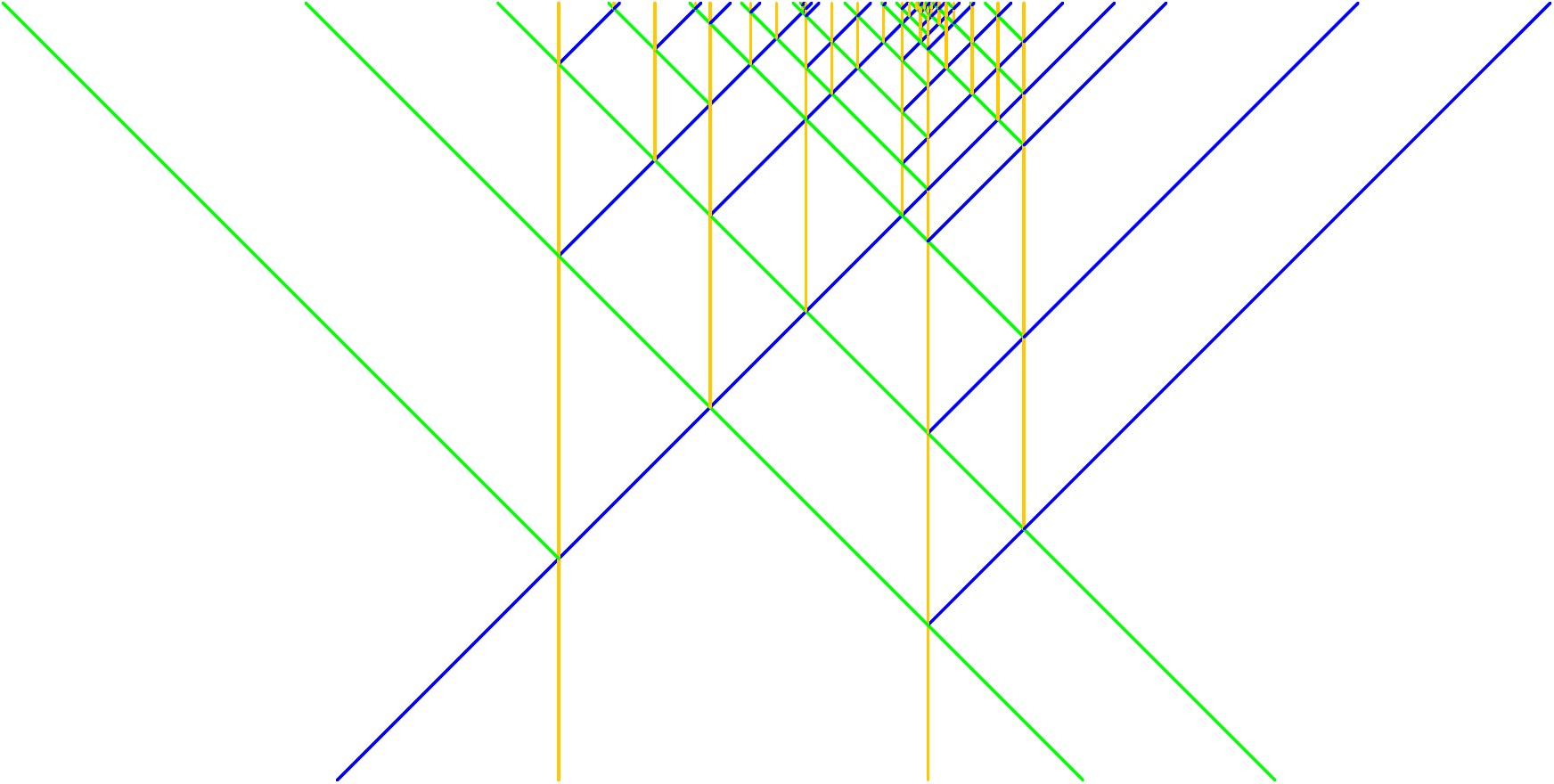}}
    \caption{Examples of supports diagrams for \SMeuclid with an accumulation.}\label{fig:accu3-full_diagrams}
  \end{figure}
\end{remark}

%
%
\paragraph{Accumulation with an irrational speed and a rational initial configuration}

We can also provide an accumulation starting with a rational initial configuration, but with a signal machine including meta-signals
of irrational speed, which can be used to create an irrational distance.

To achieve this, we start from the machine \SMeuclid in which all speeds $1$ are replaced by speeds $\varphi\in\R\smallsetminus\Q$.
This new machine \SMeuclidPhi is given in \Fig{fig:accu3-irrational_speed-diagram-machine}.
We can easily set an initial rational configuration so that, by using the fact that $\frac{-1}{\varphi}$
(the ratio between the two non-zero speeds) is irrational, 
we obtain after two collisions an irrational configuration wich can be used to run Euclid's algorithm.
We use a new meta-signal, \SigEuclidStart, of speed $\varphi$, and the rule 
\ColRule{\SigEuclidStart,\ \SigEuclidZagSmall}{\SigEuclidZagSmall,\ \SigEuclidWallb} to set a stationary signal \SigEuclidWallb at an irrational position.
Indeed, as illustrated by \Fig{fig:accu3-irrational_speed-start}, 
if we start from the rational configuration $\Config[0] = \{\SigEuclidWallZero @ 0,\ \SigEuclidStart @ 0,\ \SigEuclidZagSmall @ 1,\ \SigEuclidWalla @ 1\}$,
a simple computation shows that the first collision happens at the position $x= \frac{1}{1+\varphi} = \frac{1}{\varphi} = \varphi - 1$ 
and that the configuration at time $1$ has the form 
$\Config[1]^{1,\varphi - 1} = \{\SigEuclidWallZero @ 0,\ \SigEuclidZigSmall @ 0,\ \SigEuclidWallb @ \varphi - 1,\ \SigEuclidWalla @ 1 \}$,
which is the required form to start an infinite execution of the Euclid's algorithm explained previously.
In this initialization, the irrational speeds is used to create an irrational distance from a rational initial configuration.
After the initialization, the process never halts because the ratio $\frac{1}{\varphi-1}=\varphi$ is irrational, 
and because the rules of \SMeuclidPhi (except the one involving \SigEuclidStart) are the same than \SMeuclid.
The sequence $(a_n)_{n\in\N}$ is the same than the previous example, 
but shifted by one step because of the initialization which sets the first values at $a_0 = 1$ and $a_1 = \varphi - 1$
(respectively to the values $a_1$ and $a_2$ of the previous example). 
An accumulation is thus produced at position $x=0$ and time $\TLimit=1+\uptau\sum\limits_{i=1}^{\infty} a_n$,
where $\uptau$ is the time of a unitary back-and-forth: it depends only on speeds and is given by $\uptau=1+\frac{1}{\varphi}=\varphi$.
The whole diagram is displayed in \Fig{fig:accu3-irrational_speed-diagram}.

\begin{figure}[hbt]%
  \centering\small
  \subfigcapskip 1em
  \subfigure[Meta-signals and rules of \SMeuclidPhi.\label{fig:accu3-irrational_speed-diagram-machine}]{%
    \begin{tabular}[b]{c}
      \begin{tabular}[b]{r|c}
	\bf Meta-signal & \bf Speed \\
	\hline
	\SigEuclidZigSmall, \SigEuclidZigBig, \SigEuclidStart & $\varphi$ \rule{0pt}{5mm}\\[.2em]
	\SigEuclidWallZero, \SigEuclidWalla, \SigEuclidWallb & $0$ \\[.2em]
	\SigEuclidZagSmall, \SigEuclidZagBig & $-1$ \\[.2em]
      \end{tabular}\\[1em]%
      \begin{tabular}[b]{r@{ $\rightarrow$ }l}
	\multicolumn{2}{c}{\bf Collision rules}\\\hline
        \{ \SigEuclidStart,\ \SigEuclidZagSmall \} & \{ \SigEuclidZagSmall,\ \SigEuclidWallb \} \rule{0pt}{5mm}\\[.2em]
	\{ \SigEuclidZigSmall, \SigEuclidWallb \} & \{ \SigEuclidZagSmall, \SigEuclidWallb, \SigEuclidZigBig \}\\[.2em]
	\{ \SigEuclidWallZero , \SigEuclidZagSmall \} & \{ \SigEuclidWallZero , \SigEuclidZigSmall \}\\[.2em]
	\{ \SigEuclidWalla , \SigEuclidZigBig \} & \{ \SigEuclidZagBig \}\\[.2em]
	\{ \SigEuclidWallb , \SigEuclidZagBig \} & \{ \SigEuclidZagBig, \SigEuclidWalla \}\\[.2em]
	\{ \SigEuclidZigSmall , \SigEuclidZagBig \} & \{ \SigEuclidZagSmall, \SigEuclidWallb \}\\[.2em]
	\{ \SigEuclidZigBig, \SigEuclidZagBig \} & \{ \SigEuclidZagBig \}\\[.2em]
	\{ \SigEuclidZigSmall, \SigEuclidWallb, \SigEuclidZagBig \} & \{ \SigEuclidZagBig, \SigEuclidWallZero\}\\[.2em]
	\{ \SigEuclidWallZero, \SigEuclidZagBig \} & \{ \SigEuclidWallZero\}%
      \end{tabular}%
    \end{tabular}}
  \subfigure[Creating an irrational distance.\label{fig:accu3-irrational_speed-start}]{\quad%
    \begin{tikzpicture}[x=4.5cm,y=4.5cm]
\definecolor{couleurwall0}{RGB}{255,200,0}
\tikzstyle{drawwall0}=[draw=couleurwall0]
\tikzstyle{nodewall0}=[pos=0.5, right]
\definecolor{couleurwallx}{RGB}{255,200,0}
\tikzstyle{drawwallx}=[draw=couleurwallx]
\tikzstyle{nodewallx}=[pos=0.87, left]
\definecolor{couleurwally}{RGB}{255,200,0}
\tikzstyle{drawwally}=[draw=couleurwally]
\tikzstyle{nodewally}=[pos=0.9, left]
\definecolor{couleurstart}{RGB}{255,0,0}
\tikzstyle{drawstart}=[draw=couleurstart]
\tikzstyle{nodestart}=[pos=0.5, above, sloped]
\definecolor{couleurzig}{RGB}{0,0,255}
\tikzstyle{drawzig}=[draw=couleurzig]
\tikzstyle{nodezig}=[pos=0.5, below, sloped]
\definecolor{couleurZIG}{RGB}{0,0,255}
\tikzstyle{drawZIG}=[draw=couleurZIG]
\tikzstyle{nodeZIG}=[pos=0.5, above, sloped]
\definecolor{couleurzag}{RGB}{0,255,0}
\tikzstyle{drawzag}=[draw=couleurzag]
\tikzstyle{nodezag}=[pos=0.5, above, sloped]
\definecolor{couleurZAG}{RGB}{0,255,0}
\tikzstyle{drawZAG}=[draw=couleurZAG]
\tikzstyle{nodeZAG}=[pos=0.5, above, sloped]
\draw[drawzag] (1.000000,0.000000) -- node[nodezag]{\SigEuclidZagSmall} (0.618034,0.381966);
\draw[drawstart] (0.000000,0.000000) -- node[nodestart]{\SigEuclidStart} (0.618034,0.381966);
\draw[drawzag] (0.618034,0.381966) -- node[nodezag]{\SigEuclidZagSmall} (0.000000,1.000000);
\draw[drawwall0] (0.000000,0.000000) -- node[nodewall0]{\SigEuclidWallZero} (0.000000,1.000000);
\draw[drawwall0] (0.000000,1.000000) -- node[nodewall0]{\SigEuclidWallZero} (0.000000,1.200000);
\draw[drawzig] (0.000000,1.000000) -- node[nodezig]{\SigEuclidZigSmall} (0.323607,1.200000);
\draw[drawwallx] (0.618034,0.381966) -- node[nodewallx]{\SigEuclidWallb} (0.618034,1.200000);
\draw[drawwally] (1.000000,0.000000) -- node[nodewally]{\SigEuclidWalla} (1.000000,1.200000);
\draw[<->][arrow] (0,-\Sep) -- node[below] {$1$} (1,-\Sep);
\draw[<->][arrow] (0,1.25) -- node[above] {$\varphi-1$} (0.618,1.25);

\end{tikzpicture}}%
  \qquad
  \subfigure[Infinite run.\label{fig:accu3-irrational_speed-diagram}]{\quad%
    \includegraphics[width=.10\textwidth]{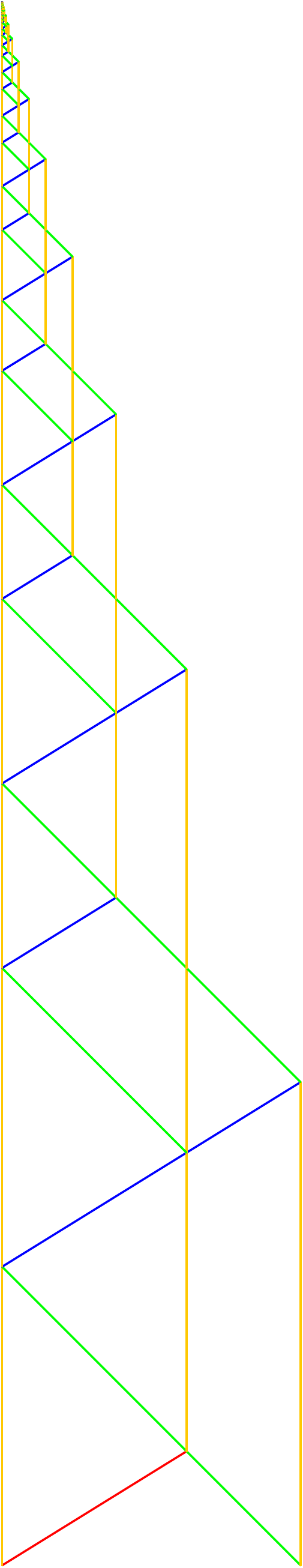}%
  \qquad}%
  \caption{Meta-signals, rules and an infinite run of the machine \SMeuclidPhi.}\label{fig:accu3-machine_irrat-speed}
\end{figure}
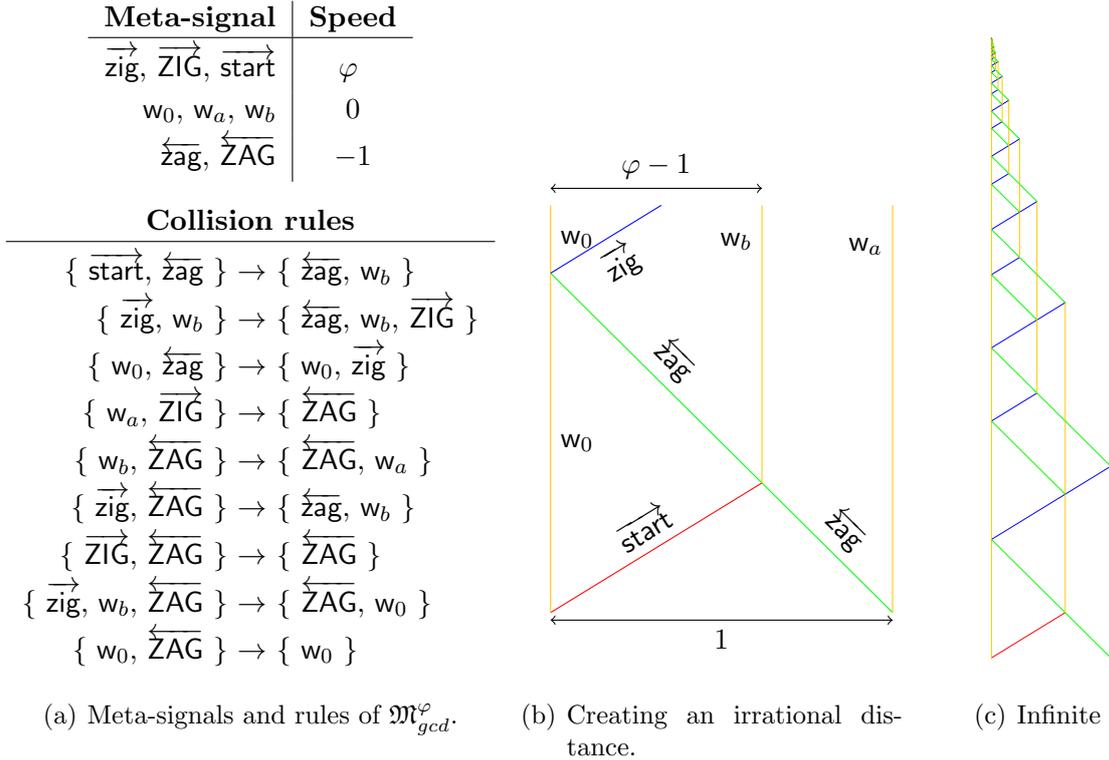

\medskip

With \Lem{lemma:accu3-euclid_accumulation}, we finally obtain:

\begin{proposition}\label{prop:accu3}
  There exists $3$-speed signal machines defined with an irrational ratio between either two speeds or two of its initial positions, that produce accumulations.
\end{proposition}

\subsection{Case of rational-like machines}
\label{sec:3speeds_no-accu}
We show now that the existence of incommensurate values is a necessary condition for accumulating with some $3$-speed machine.
This proof is done in two steps.
First, we bring back the study of any rational-like $3$-speed machine to the study of a rational $3$-speed support machine
(having speeds $-1$, $0$ and a third rational speed $\nu$).
Second, we show that the diagrams of a rational $3$-speed support machine are included on some regular structures 
---called {\em \StripMultiName[es]}--- that cannot contain any accumulation.
\subsubsection{Normalization of speeds}
\label{subsec:3speeds_no-accu_normalization}
%
%
In the same way that we did for $2$-speed machines,
we can reduce the problem of accumulating with $3$-speed machines to 
machines having speeds $-1$, $0$ and $\nu$ where $\nu$ is a real positive number.

Indeed, let $\SM^{a,b,c}$ be a signal machine having only three distinct speeds $a, b$ and $c$ with $a < b < c$. 
We define a machine $\SM^{-1,0,\nu}$, equivalent to $\SM^{a,b,c}$.
Consider the function $h : \R\to\R$ defined in \Subsec{subsec:properties}, with $c = -1$ and $d = 0$:
$h(x) = \frac{x}{b-a} - \frac{b}{b-a}\enspace$.
We have then $f(a) = -1$, $f(b) = 0$ and $f(c) = \frac{c-b}{b-a} > 0$.
We call $\nu = f(c)$.
By \Lem{lemma:affine-trans}, since $h$ is an affine function of strictly positive ratio, 
diagrams of  $\SM^{a,b,c}$ and $\SM^{-1,0,\nu}$ will be equivalent.
In particular, if $\SM^{a,b,c}$ produces a space-time diagram including an accumulation,
$\SM^{-1,0,\nu}$ will also produce a corresponding diagram including an accumulation.

In this section, we want to prove that no rational-like signal machine produces accumulation when started on rational-like initial configuration.
As done for $2$ speeds, for any rational-like machine $\SM^{a,b,c}$, we will only consider by now the support signal machine \Supp{\SM^{-1,0,\nu}}
of  the machine $\SM^{-1,0,\nu}$, equivalent to $\SM^{a,b,c}$, where $\nu = \frac{c-b}{b-a}$.
The support machine \Supp{\SM^{-1,0,\nu}} contains exactly $3$ meta-signals: one for each speed amoung $-1, 0$ and $\nu$.
Each collision rules of \Supp{\SM^{-1,0,\nu}} outputs the maximal possible number of meta-signals, that is the $3$ meta-signals of the machine.

Note that if $\SM^{a,b,c}$ is rational-like, then $\SM^{-1,0,\nu}$ is a rational machine: 
if all ratios between $a$, $b$ and $c$ are rational, then $\nu$ will also be rational.
Indeed, for any real non zero numbers $x, y$, the following holds:
$\frac{x}{y}\in\Q \Longleftrightarrow \frac{x-y}{y}\in\Q\enspace.$
Since $\nu = \frac{c-b}{b-a} = \frac{c}{b-a} - \frac{b}{b-a}$, 
we deduce that $\frac{a}{b}, \frac{a}{c} \in\Q\Rightarrow\nu\in\Q\enspace.$

So we will express $\nu$ in the form of a fraction $\nu=\frac{p}{q}$ where $p$ and $q$ are positive coprimes integer, 
so that the fraction is irreducible (and positive since $\nu$ is).
To simplify the notation, we will note \SMnu or \SMpq to designate \Supp{\SM^{-1,0,\nu}}, 
the support machine having the $3$ rational speeds $-1, 0$ and $\nu=\frac{p}{q}$.

\subsubsection{Meshes and diagrams}
\label{sec:3speeds_no-accu_proofs}
%
%
We construct a family of diagrams such that every diagram generated by the support machine \SMnu is included in one diagram of the family.
Since no diagram in the family contains any ccumulation, the machine \SMnu cannot produce an accumulation.
This is done in two steps. 
In a first time, we show that diagrams obtained from a special form of initial configurations eventually become periodic: 
we call such a diagram a {\em \StripMultiName}.
In a second time, we prove that any rational initial configuration \Config[0] of \SMnu can be extended into one of the configurations 
that produce \StripMultiName[es] and so there exists a \StripMultiName which includes the support diagram of the diagram started from \Config[0].

\paragraph{Notions of {\StripName[s]} and {\StripMultiName[es]}}

We first introduce {\em \StripName[s]}, which are used as the basic components of meshes.
Recall that in this section, the speed $\nu$ is given by the rational number $p/q$, where $p, q \in\N$ are coprimes.

\begin{definition}[Strip]\label{def:strip}
  Let be $\Posi, w\in\R^{+}$.
  We call {\em \Strip{p}{q}{\Posi}{w}} the diagram generated by \SMnu from the initial configuration:
  \[
  \left\{\ \SigAt{[\SigL, \SigR]}{\Posi},\ \SigAt{\SigS}{\Posi+\frac{i}{p+q}\cdot w},\ \SigAt{[\SigL,\SigR]}{\Posi+w}\ \middle|\ 0\leq i \leq p+q\right\}\enspace.
  \]
\end{definition}

\Figure{fig:noaccu3-strip_example} displays a \Strip{\frac{2}{3}}{5}{0}{1} (the diagram has been croped on both sides, 
signals leaving on both sides are supposed to propagate indefinitely).

The parameter $\Posi$ is the position of the begining of the \StripName \ie the position of the left-most stationary signal, 
and $w$ is the total width of the \StripName (so the position of the right-most stationary is $\Posi + w$).
Parameters $p$ and $q$ (remember that we have the rational positive speed $\nu=p/q$) provides $p+q$, which the number of subdivisions of same length of the interval $[\Posi; \Posi+w]$.
This is the number of such subdivisions that would have been created during the evolution of the machine  
without placing all the initial stationary signals between the two extremal stationary signals, as illustrated by \Fig{fig:noaccu3-strip_uncomplete}.
Intuitively, it corresponds to the minimal number of divisions of the strip in equal parts so that all 
collisions between two signals \SigR and \SigL occurs exactly at the position of a subdivision
So by placing a stationary signal \SigS at each such position in the initial configuration
every collision between \SigR and \SigL also involves \SigS and is a triple collision.

So there are $p+q-1$ stationary signals between the left-most and the right-most stationary signals,
and their position are $\Posi + \frac{i}{p+q} \cdot w$ for $1 \leq i \leq p+q-1$.
Including the walls, a stationary signal \SigS is set at each position $\Posi + \frac{i}{p+q}$ for $0 \leq i \leq p+q$.
\Figure{fig:noaccu3-strip_parameters} gives the geometrical meanings of a \StripName parameters.

\begin{figure}[hbt]
  \centering
  \subfigure[An uncomplete \StripName.\label{fig:noaccu3-strip_uncomplete}]{%
    \includegraphics[width=.3\textwidth]{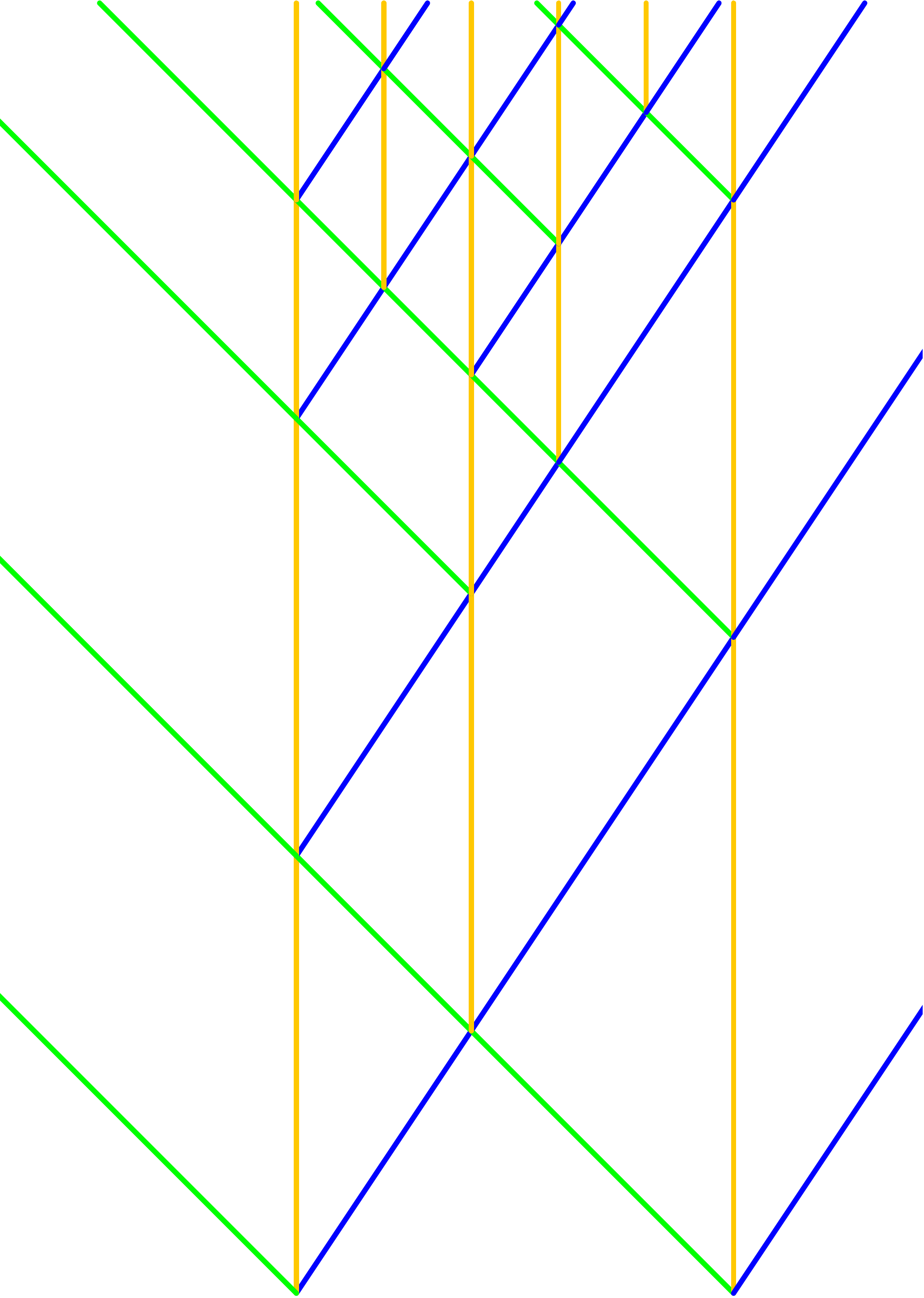}}%
  \hfill
  \subfigure[A \Strip{2}{3}{0}{1}.\label{fig:noaccu3-strip_example}]{%
    \includegraphics[width=.25\textwidth]{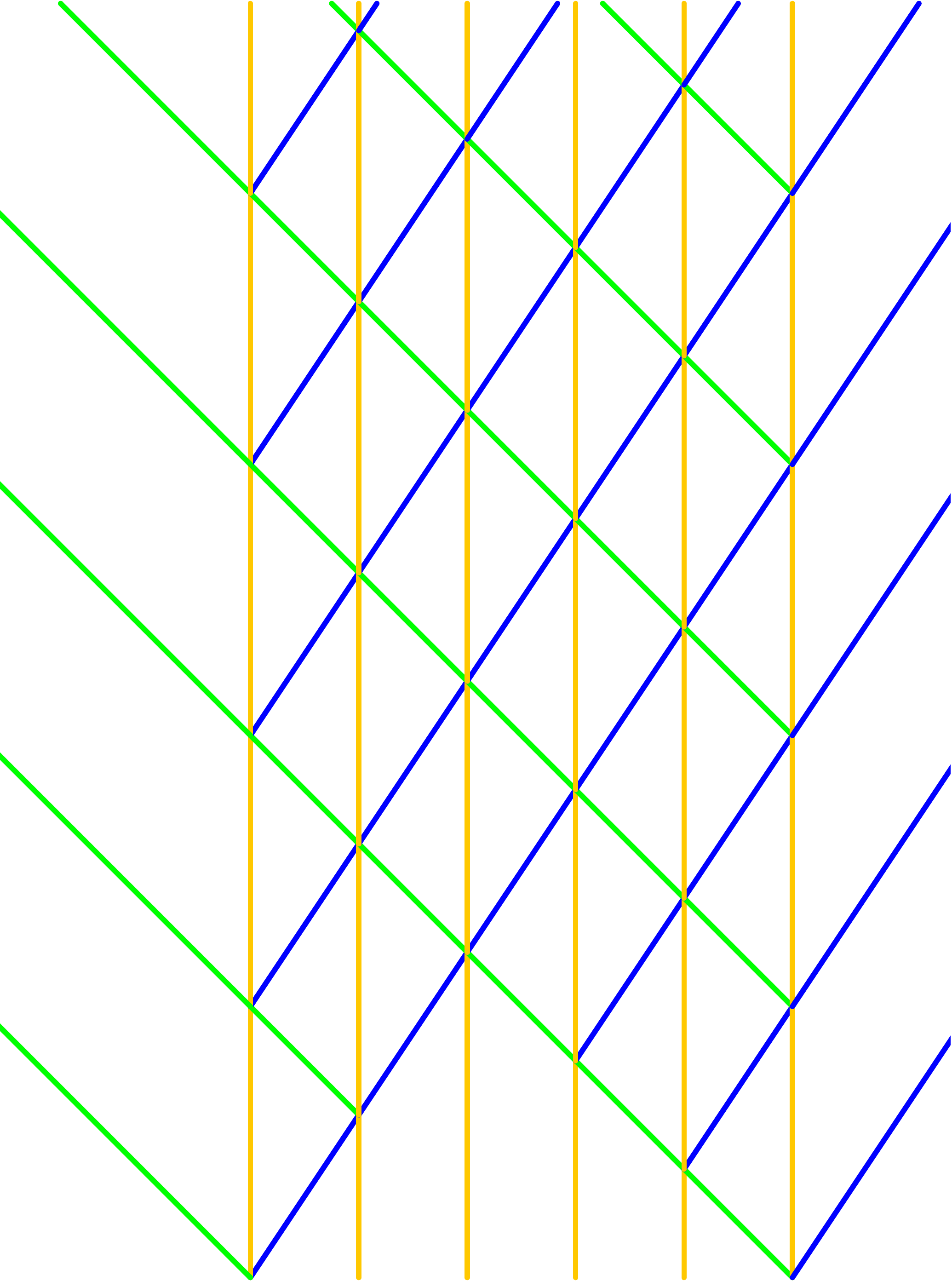}}%
  \hfill
  \subfigure[Zig-zag on one and two subdivisions.\label{fig:noaccu3-zig-zag_strip}]{\qquad%
    \includegraphics[width=.035\textwidth]{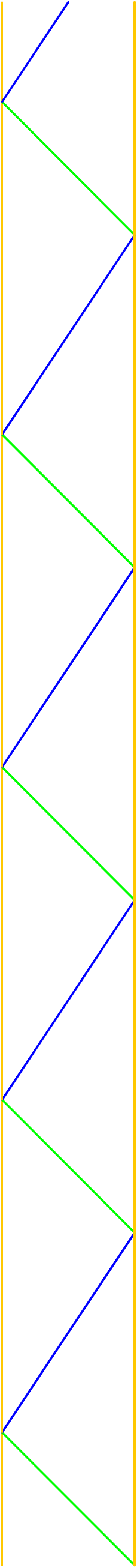}\qquad%
    \includegraphics[width=.07\textwidth]{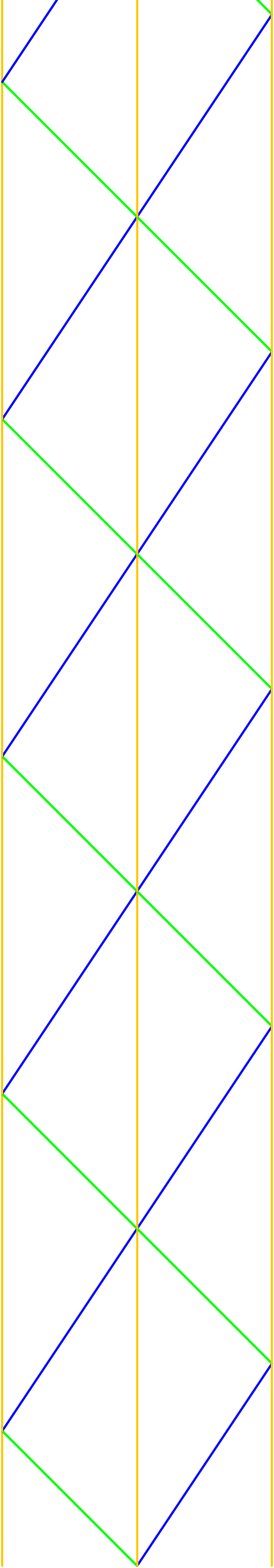}\qquad}%
  \hfill
  \caption{Basics structures for diagrams with $3$ speeds.}\label{fig:noaccu3-structures}
\end{figure}

\begin{figure}[hbt]
  \centering
  \Height 3.8cm
  \Width 0.5\Height
  \ArrowSpacing 9mm 
  \subfigure[Parameters of a \Strip{p}{q}{\Posi}{w}.\label{fig:noaccu3-strip_parameters}]{\qquad%
    \begin{tikzpicture}[x=\Width,y=\Height]
      \draw (0,0) node [above right,inner sep=0]{\includegraphics[height=\Height]{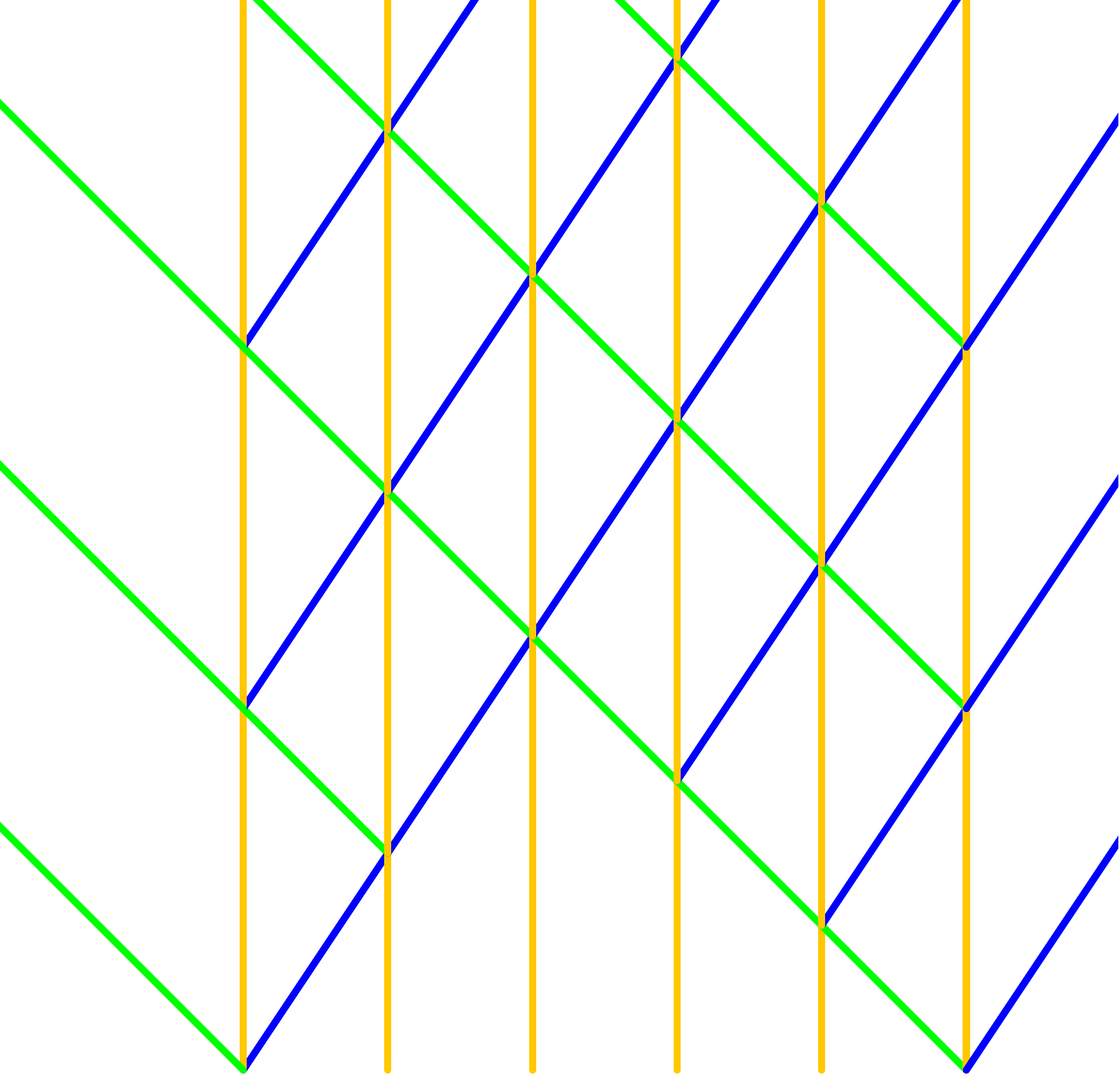}}; 
      \draw (0.45,0) node[inner sep=0em]{$\bullet$} ; \draw (0.45,0) node[below,inner sep=.4em]{$x_0$} ; 
      \draw[<->][arrow] (0.45,-.8\ArrowSpacing) -- node[below] {$w$} (1.8,-.8\ArrowSpacing); 
      \draw[-][dashed] (1.8,0) -- node[below] {$p$} (2.05,0); \draw[-][dashed] (2.05,0) -- node[right] {$q$} (2.05,.2); 
      \draw[decorate,decoration={brace,amplitude=10pt}](0.45,1.02\Height) -- node[above,yshift=.5cm]{$p+q$ subdivisions} (1.8,1.02\Height) ; 
    \end{tikzpicture}\qquad}
  \hfill
  \Height 3.8cm
  \Width 0.5\Height
  \subfigure[Parameters of a \StripMulti{p}{q}{\Posi}{w}{k}.\label{fig:noaccu3-multi-strip_parameters}]{%
    \begin{tikzpicture}[x=\Width,y=\Height]
      \draw (0,0) node [above right,inner sep=0]{\includegraphics[height=\Height]{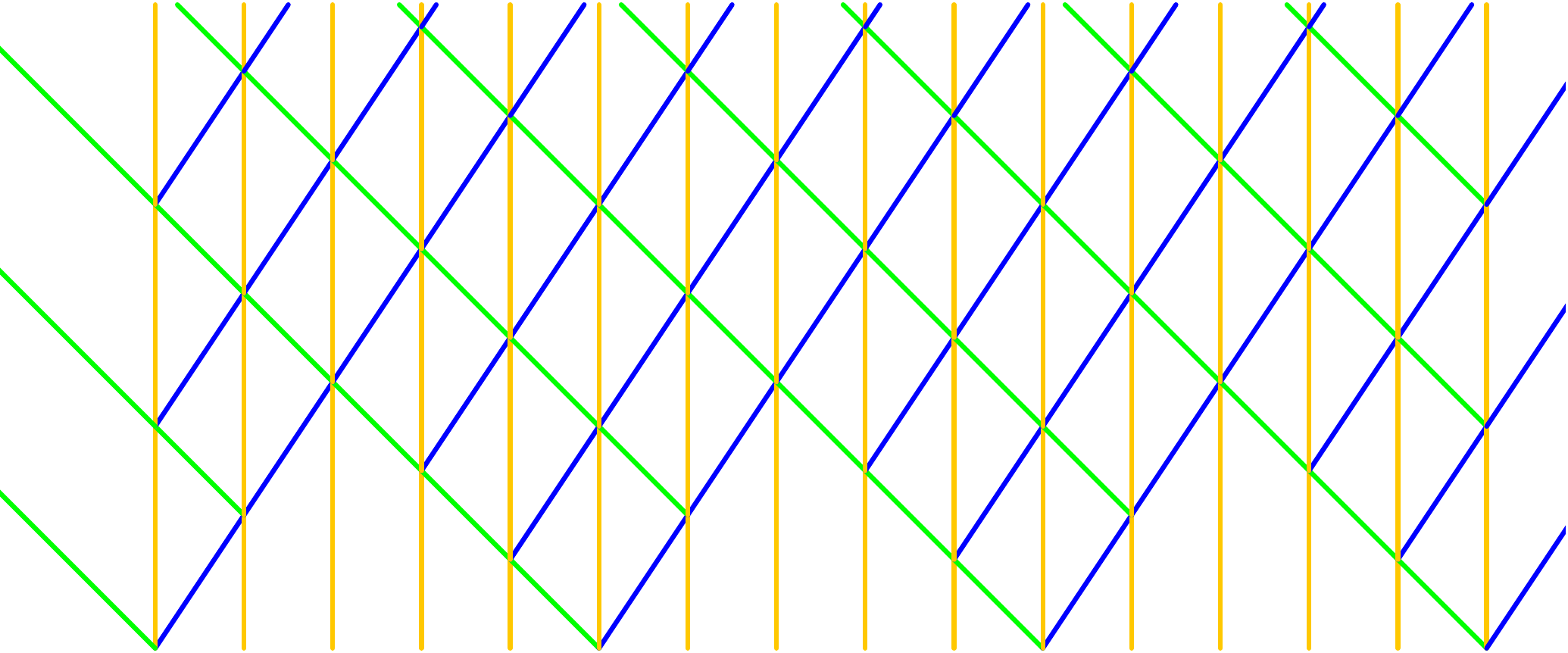}};
      \draw (0.47,0) node[inner sep=0em]{$\bullet$} ; \draw (0.47,0) node[below,inner sep=.4em]{$x_0$} ; 
      \draw[<->][arrow] (0.47,-.8\ArrowSpacing) -- node[below] {$w$} (1.84,-.8\ArrowSpacing); 
      \draw[<->][arrow] (1.85,-.8\ArrowSpacing) -- node[below] {$w$} (3.2,-.8\ArrowSpacing); 
      \draw[<->][arrow] (3.21,-.8\ArrowSpacing) -- node[below] {$w$} (4.55,-.8\ArrowSpacing); 
      \draw[-][dashed] (4.55,0) -- node[below] {$p$} (4.8,0); \draw[-][dashed] (4.8,0) -- node[right] {$q$} (4.8,.2); 
      \draw[decorate,decoration={brace,amplitude=10pt}](0.47,1.02\Height) -- node[above,yshift=.5cm]{$p+q$ subdivisions} (1.84,1.02\Height) ; 
      \draw[decorate,decoration={brace,amplitude=10pt}](1.85,1.02\Height) -- node[above,yshift=.5cm]{\ldots} (3.2,1.02\Height) ; 
      \draw[decorate,decoration={brace,amplitude=10pt}](3.21,1.02\Height) -- node[above,yshift=.5cm]{$p+q$ subdivisions} (4.55,1.02\Height) ; 
      \draw[decorate,decoration={brace,amplitude=10pt},yshift=1.3\Height](0.45,0) -- node[above,yshift=.5cm]{($k$ times)} (4.55,0) ; 
    \end{tikzpicture}}
  \caption{Parameters of \StripName[s] and \StripMultiName[es].\label{fig:noaccu3-strips_parameters}}  
\end{figure}

We first prove two lemmata, stating that the structure of a \StripName is {\em regular}:
in the central part of the \StripName, the diagram behaves with respect to a periodic pattern
and space outside the central part contains only parallel signals that will never collide.

\begin{lemma}\label{lemma:strip-no_collision_outside}
  No collision can happen on the space outside the interval $[\Posi; \Posi+w]$ in a \Strip{p}{q}{x}{w}.
\end{lemma}

\begin{proof}
  Any signal going crossing the left-most signal \SigS, initially placed in $\Posi$ is necessarily a signal \SigL.
  By the form of the configuration \Config[0], there is no signal placed before the position $\Posi$ and since signals going on the left
  side of the first stationary signal are all parallel, there is no possible collision on space before the position $\Posi$.
  The same happens symetrically on the other side.
\end{proof}

We show next that for proper parameters, a \StripName is a regular structure, composed by a vertical grid with parallel signals leaving on both sides.

\begin{lemma}\label{lemma:strip_periodicity}
  Any \Strip{p}{q}{\Posi}{w} becomes periodic on the interval $[\Posi;\Posi+w]$ after the time $t_T = \frac{q}{p+q}\cdot w$.
  After this time, the period is given by $T=\frac{w}{p}$.
\end{lemma}

This lemma means that for all $t\geq t_T$, for all $x\in [\Posi; \Posi + w]$ we have: $\Config[t+T](x)=\Config[t](x)$.
We cannot have a complete periodicity on the whole space because of signals leaving the \StripName on both sides.
That is why we restrict the study of a strip evolution to the interval defined by the extremal stationary signals, that is the interval $[\Posi; \Posi + w]$.

\begin{proof}
  Let us begin with a simple remark: for three signals \SigS with same distance between the two first and the two last signals, 
  times of back-and-forth of signals \SigL and \SigR starting both from the central \SigS are the same, as illustrated by \Fig{fig:noaccu3-zig-zag_strip}.
  Indeed, for two signals \SigS spaced by a distance $l$ and a signal \SigR of speed $\nu$ starting from the first \SigS, 
  it will take a time $\nu l$ to \SigR for going the second \SigS. After the collision, \SigR makes a bounce on \SigS and is turned
  into a signal \SigL (of speed $-1$), which will need a time $1\cdot l$ to go back to the first \SigS.
  The total time of the back-and-forth is $\nu l+l=(1+\nu)l$. 
  The symetric back-and-forth (when starting with \SigL from the right \SigS) requires the same amount of time $(1+\nu)l$.
  So for three signals \SigS so that the middle \SigS is at distance $l$ of the two other \SigS, if one signal \SigL and one signal \SigR
  leave the central \SigS at the same time, since the time of a back-and-forth is the same for both of them (because the distance to run is the same),
  they will collide with the central \SigS exactly at the same time and the collision involved will be a triple collision.
  By the rules of a support machine, all possible signals will be output from this collision: \SigR and \SigL will leave at the same time 
  and the previous computation will apply again. So from a triple collision, there will be a triple collisions after each duration $(1+\nu)l$.
  For the same reasons, this also holds for any number of signals \SigS if two successive \SigS are spaced by the same distance.
  In the case of a \Strip{p}{q}{\Posi}{w}, the speed is $\nu=\frac{p}{q}$,
  the distance between \SigS signals is given by $\frac{w}{p+q}$ and so the time of a back-and-forth in a subdivision is 
  $(1+\frac{p}{q})\cdot \frac{w}{p+q} = \frac{w}{q}$.

  To claim that all collisions between a \SigR signal and a \SigL signal also involves a stationary signal \SigS,
  it remains to show that the collision between the two non-stationary signals of the initial configuration 
  happens exactly at the position of an initial \SigS.
  Let us compute the coordinates of this ``central'' collision between the initial \SigR and \SigL signals, initially disposed at positions \Posi and $\Posi+w$.
  This coordinates $(x_C,t_C)$ satisfy $\nu\cdot t_C+\Posi=-t_C+\Posi+w$ et $x_C=\nu\cdot t_C+\Posi$, 
  that is $x_C=\frac{\nu}{\nu+1} w+\Posi$ and $t_C=\frac{w}{\nu+1}$.
  Since $\nu=\frac{p}{q}$, we get $x_C=\frac{p}{p+q} w+\Posi$ and $t_C=\frac{q}{p+q} w$.
  So the position $x_C$ is $\frac{p}{p+q}w+\Posi$ and can indeed be put in the form $\Posi+\frac{i}{p+q}w$ where $0\leq i\leq n$,
  which corresponds to the position of an initial signal \SigS.

  With the remark made at the beginning of the proof, we can conclude that, from the time $t_C$, 
  all collisions (except the ones happening on the two extremal walls) are triple collisions.
  Thus the diagram become periodic between the two walls (\ie between positions \Posi and $\Posi+ w$)
  after the time $t_T=t_C$, and the duration of a period is given by the time of a back-and-forth \ie $T=\frac{w}{q}$.
\end{proof}
 
Regarding the ``external parts'' of the \StripName, some signals are generated with a regular spacing and propagate indefinitely.
On the left part, signals \SigL (moving on the left) are all spaced by a distance $d = (1+\frac{p}{q})\times\frac{w}{p+q}$,
which is the distance covered by a signal \SigL during the time of a back-and-forth on a subdivision.
On the right part, signals \SigR are all spaced by a horizontal distance $d=\frac{(1+\frac{p}{q})\cdot w\cdot p}{q(p+q)}$.

In fact, the number $p+q$ of subdivisions corresponds to $1/gcd\left(\frac{\nu}{\nu+1},1\right)$ 
(note that when $\nu$ is rationnal, this number is indeed an integer).
This value is deduced from the study of the coordinates of the collision $C$.
Any other multiple of $1/gcd\left(\frac{\nu}{\nu+1},1\right)$ for the number of subdivisions also allows to get a \StripName 
which becomes perdiodic after the time $t_T$ (the difference being that the subdivisions are more or less narrow in function of the multiple chosen).

\bigskip

From \StripName[s] that we use as elementary structures, we can now define a more general and regular structure:

\begin{definition}[Mesh]\label{def:multistrip}
  Let be $k \in\N$ and $\Posi, w\in\R^{+}$.
  We call {\em \StripMulti{p}{q}{\Posi}{w}{k}} a diagram generated by \SMnu from the initial configuration:
  \[
  \left\{\SigAt{[\SigL,\SigS,\SigR]}{\Posi+l\cdot w},\ \SigAt{\SigS}{\Posi+\left(\frac{i}{p+q}+j\right)\cdot w}\ \middle|\ 0\leq l \leq k,\ 0<i<p+q\text{ and }0\leq j<k\right\}\enspace.
  \]
\end{definition}

A \StripMulti{p}{q}{\Posi}{w}{k} corresponds to $k$ copies of a \Strip{p}{q}{\Posi}{w} juxtaposed side by side so that the walls of two glued copies are superposed.
Since all copies have the same parameters $w$ (the width) and $p+q$ (the number of equal subdivisions of $w$, also the number of zig-zag substrips), 
their substrips have all the same width given by $\frac{w}{p+q}$. 

Note that a \Strip{p}{q}{\Posi}{w} is a \StripMulti{p}{q}{\Posi}{w}{1}.
A \StripMulti{2}{3}{0}{1}{3} is given in \Fig{fig:noaccu3-multi-strip_example}: it corresponds to three copies of the \StripName given in \Fig{fig:noaccu3-strip_example}.

\begin{figure}[hbt]
  \centering
  \includegraphics[width=.7\textwidth]{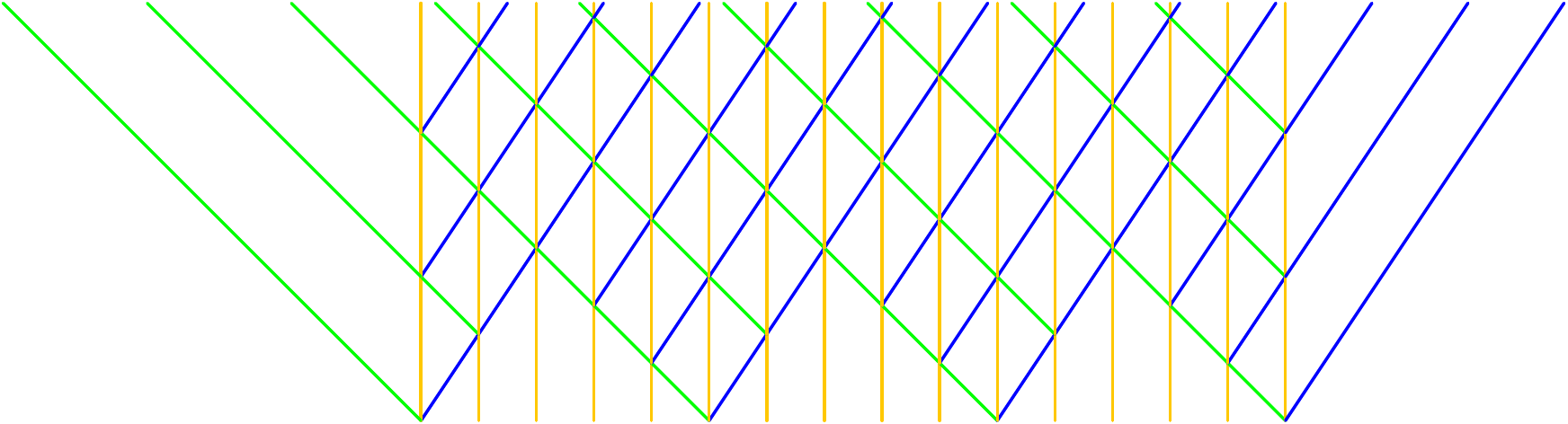}%
  \caption{A  \StripMulti{2}{3}{0}{1}{3}.}\label{fig:noaccu3-multi-strip_example}
\end{figure}

\medskip 

The regularity of \StripMultiName[es] can be deduced from the one of \StripName[s]:

\begin{lemma}\label{lemma:multi-strip_periodicity}
  Any \StripMulti{p}{q}{\Posi}{w}{k} becomes periodic in the interval $[x_0; x_0+k\cdot w]$, 
  with a period $T = \frac{w}{p}$ and after the time $t_T = \frac{q}{p+q}\cdot w$.
\end{lemma}

\begin{proof}
  For all $j\in\llbracket 0;k-1\rrbracket$, the configuration
  \[\left\{\ \SigAt{\SigS}{\Posi+\left(\frac{i}{p+q}+j\right)\cdot w}, \SigAt{\SigL}{\Posi+j\cdot w}, \SigAt{\SigR}{\Posi+j\cdot w}\ \middle|\ 0\leq i\leq p+q\ \right\}\]
  generates a \Strip{p}{q}{\Posi+jw}{w} (by definition of a \StripName).
  The initial configuration corresponding to the \StripMulti{p}{q}{\Posi}{w}{k} can be decomposed into such configurations (with a junction at each position $x=\Posi+j\cdot w$).
  Since all these configurations correspond to \StripName[s] having all the same parameters, back-and-forths in all subdivisions take the same duration. 
  It follows from the proof of \Lem{lemma:strip_periodicity} that all collisions between signals \SigL and \SigR also involve stationary signals \SigS, 
  including the ones at the junction of two \StripName[s] of the \StripMultiName.
  Since all \StripName[s] are periodic with the same period from the same time $t_C$, the \StripMultiName is also periodic.
  Its period is the same that the \StripName[s] and is given by \Lem{lemma:strip_periodicity}: 
  the \StripMulti{p}{q}{\Posi}{w}{k} is perdiodic in the interval $[x_0; x_0+k\cdot w]$ with a period $T=\frac{w}{p}$ 
  and from the time $t_T=\frac{q}{p+q}\cdot w$.
\end{proof}

One essential consequence of this lemma is the following corollary:

\begin{corollary}\label{coro:multi-strip_no-accu}
  No accumulation can occur on a \StripMulti{p}{q}{\Posi}{w}{k}.
\end{corollary}

\begin{proof}
  Because of \Lem{lemma:strip-no_collision_outside}, we just need to show that no accumulation occurs in the interval $[\Posi ; \Posi + k\cdot w]$.
  By \Lem{lemma:multi-strip_periodicity}, any \StripMultiName becomes periodic of period $T=\frac{w}{p}$, after a finite time $t_T$.
  There is only a finite number of collisions before the time $t_T$.
  For any time $t > t_T$, there can be only of finite number of collisions happening during the time interval $[t; t + T]$.
  Since a necessary condition for having an accumulation is the existence of a time interval during 
  which an infinite number of collision occurs, no accumulation can appear in a \StripMultiName.
\end{proof}

\paragraph{Inclusion in a \StripMultiName}

We show here that for any finite initial configuration \Config[0] having only rational ratio between distances,
there is a configuration $\Config[0]^{\StripMultiName}$ of a \StripMultiName $\mathcal{S}$ so that
\Config[0] is included in $\Config[0]^{\StripMultiName}$.
It will follow directly that the whole diagram generated from \Config[0] is entirely included in $\mathcal{S}$.
\Figure{fig:noaccu3_random-diagram-and-mesh} illustrates this idea, 
by providing an example of an arbitrary diagram and a \StripMultiName  that includes the whole arbitrary diagram
(extremal initial positions of the diagram of \Fig{fig:no-accu3_random-ci} match with those of \Fig{fig:no-accu3_random-ci_mesh}).

\begin{figure}[hbt]
  \centering
  \subfigure[Diagram \STD from an arbitrary initial configuration.\label{fig:no-accu3_random-ci}]{\quad%
    \includegraphics[width=.5\textwidth]{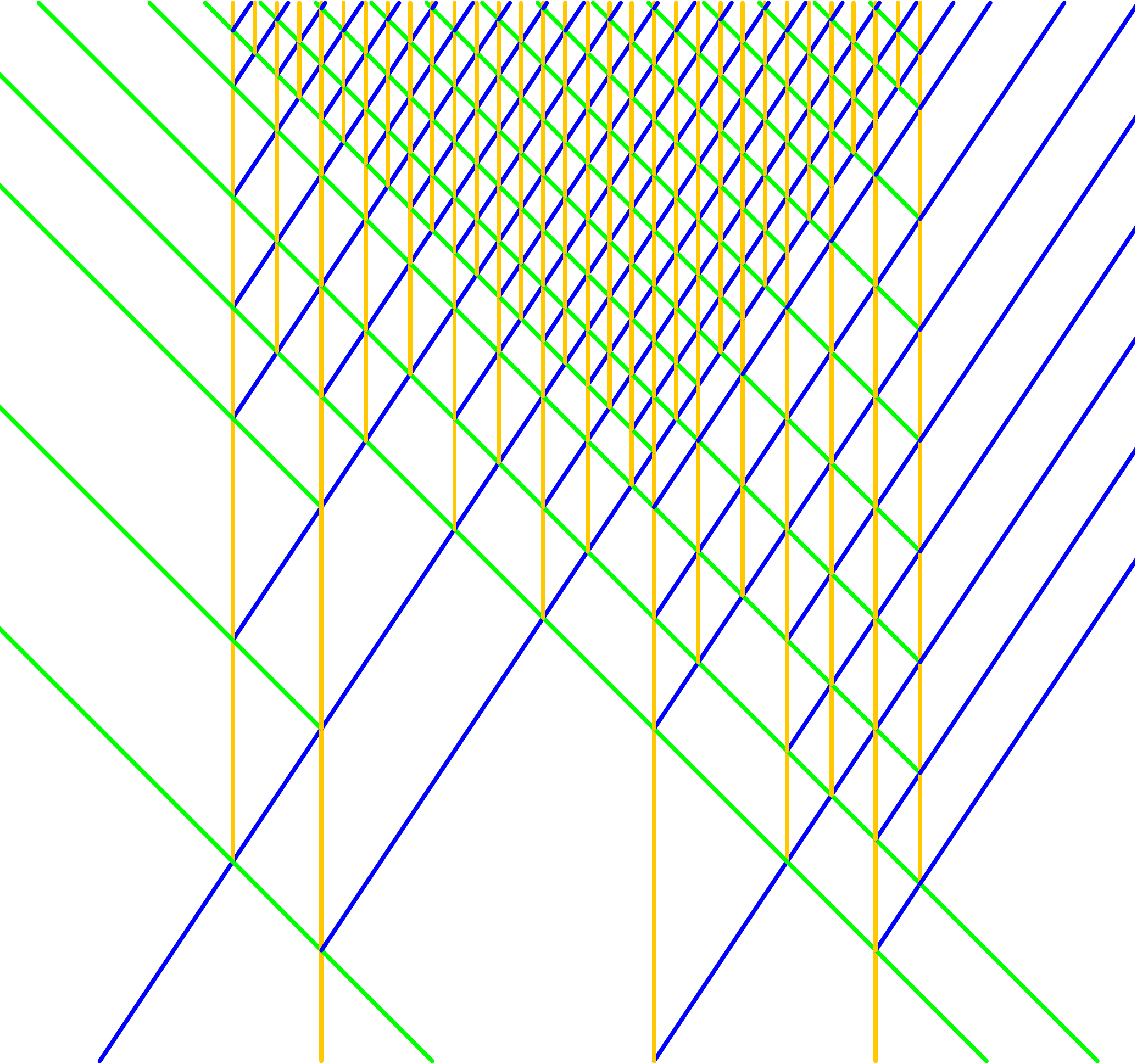}\quad}\\%
  \subfigure[The corresponding \StripMultiName $S$.\label{fig:no-accu3_random-ci_mesh}]{%
    \includegraphics[width=.65\textwidth]{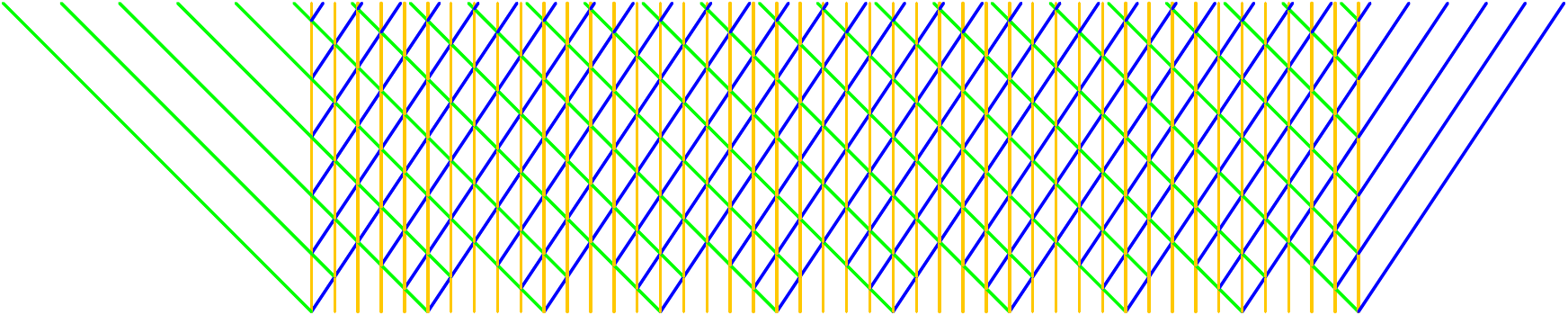}}%
  \caption{A diagram of \SMnu and a \StripMultiName which includes it.}\label{fig:noaccu3_random-diagram-and-mesh}
\end{figure}

\begin{lemma}\label{lem:diagram-to-multi-strip}
  For every space-time diagram \STD generated by \SMnu from a finite initial configuration having only rational ratio between distances,
  there exist $k \in\N$ and $\Posi, w \in\R$ so that \STD is included in a \StripMulti{p}{q}{\Posi}{w}{k}.
\end{lemma}

\begin{proof}
  Let \Config[0] be a finite rational-like initial configuration \ie
  $\Config[0] = \{ \sigma_i @ x_i \ | \ x_i \in\R \}_{1\leq i \leq m}$
  where $x_1 \leq x_2 \leq \ldots \leq x_m$ and so that for all $i, j, l\in\llbracket 0;m\rrbracket$ the value $\frac{x_i-x_j}{x_j-x_l}$ $(x_j\neq x_l)$ is rational.
  Let \STD be a diagram generated by \SMnu from \Config[0].
  Let us show that \STD is included in $\mathcal{S}$, the \StripMulti{p}{q}{\Posi}{w}{k} with parameters:
  \begin{itemize}
  \item $w = gcd(x_1-x_0, x_2-x_1,\ldots, x_m-x_{m-1})$,
  \item $k=\frac{x_m-x_0}{w}$.
  \end{itemize}

  These parameters are indeed valid parameters for a \StripMultiName: 
  $w\in\R^{+}$ is well-defined because all ratios between distances in \Config[0] are rational (and is positive);
  and $k\in\N$ by definition of $w$.
  Note that $\mathcal{S}$ starts from the configuration $\Config[0]^{\mathcal{S}}$ defined by:

  \noindent$\left\{\SigAt{[\SigL,\SigS,\SigR]}{\Posi+l\cdot w},\ \SigAt{\SigS}{\Posi+\left(\frac{i}{p+q}+j\right)\cdot w}\ \middle|\  0\leq l \leq k,\ 0<i<p+q,\ 0\leq j<k\right\}.$

  It is enough to show that \Config[0] is included in $\Config[0]^{\mathcal{S}}$ 
  \ie any signal $\sigma_i$ of \Config[0] at position $x_i$ is also present in $\Config[0]^{\mathcal{S}}$ at the same position $x_i$.
  To prove this fact, we show that the position $x_i$ can be written $x_i=\Posi+l\cdot w$, where $0\leq l \leq k$.
  Since there are always the three signals (\SigL, \SigS, \SigR) at each of these positions, we don't have to distinguish some subcases, 
  according to $\sigma$ is either a signal \SigS, \SigL or \SigR.

  The position $x_i$ can be written $x_i = x_0+ l\cdot w$.
  Indeed, the value $w$ divides $x_i-x_0$ because $w = gcd(x_1-x_0, x_2-x_1,\ldots, x_m-x_{m-1})$, and it follows
  that there exists $l\in\N$ such that $x_i-x_0=l\cdot w$, that is: $x_i = x_0 + l\cdot w$.
  Now let us show that $l$ satisfies the inequality of the definition of $\Config[0]^{\StripMultiName}$ \ie that we have $0\leq l \leq k$.
  For $i\in\llbracket 0;m\rrbracket$, we have:
  $$
  \begin{array}{c r@{\ \leq\ } c r@{\ \leq\ } c r}
  0 & & x_i - x_0 & & x_m - x_0 \\
  0 & & l\cdot w  & & x_m - x_0 \\
  0 & & l         & & \frac{x_m - x_0}{w} & \text{ since $w>0$,}\\
  0 & & l         & & k & \text{ by definition of $k$.} \\
  \end{array}
  $$

  So any signal $\sigma$ at position $x_i$ in $\Config[0]$ is also an initial signal in $\Config[0]^{\mathcal{S}}$.
  Since \SMnu is a support machine, it follows directly that $\STD$ is included in the \StripMultiName $\mathcal{S}$.
\end{proof}

We finally obtain:   

\begin{theorem}\label{theo:noaccu3}
  No $3$-speed rational-like signal machine can produce accumulation when started from a rational-like initial configuration.
\end{theorem}

\begin{proof}
  As mentionned in the paragraph on normalizations, every $3$-speed rational-like signal machine \SM can be reduced 
  to a $3$-speed machine having only rational speeds ($-1$, $0$ and $\nu$) and whom support machine is \SMnu.
  It holds by \Cor{lemma:accu-transfert} that if \SMnu started from a configuration doesn't generate accumulations, then neither does \SM (started from the same configuration).
  By \Lemma{lem:diagram-to-multi-strip}, every diagram generated by \SMnu started from a rational-like configuration is included in a \StripMultiName of \SMnu, 
  and since a \StripMultiName doesn't contain accumulation by \Cor{coro:multi-strip_no-accu},  
  it follows that no accumulation can appear in any diagram of \SMnu, and so, also in any diagram of \SM started from a rational-like configuration.
\end{proof}

%

%
%
%

\section{Conclusion}
\label{sec:conclusion}

We have shown that $2$-speed signal machines can't have accumulations and that $4$-speed can freely create accumulations.
Three-speed signal machines can only accumulate if there is an irrational ratio between speeds or between distances in the initial configuration.

The computing power of $2$-speed signal machines is very limited since the length of a computation is at most quadratic in the number of signals in the initial configuration.
The last constructions in \cite{durand-lose11tcs}, provides a Turing-universal signal machine with four speeds.
The case of $3$ speeds has been studied in \cite{durand-lose13cie}: the same dichotomy arises.
In the rational-like case, the dynamics is cyclic with bounded transient time and period; otherwise, any Turing machine could be simulated.
We conjecture that there is a similar dichotomy for hypercomputation.


%
\small
\bibliographystyle{apalike}
\bibliography{biblio}

\begin{thebibliography}{}

\bibitem[Blum et~al., 1989]{blum+shub+smale89}
Blum, L., Shub, M., and Smale, S. (1989).
\newblock On a theory of computation and complexity over the real numbers:
  {NP}-completeness, recursive functions and universal machines.
\newblock {\em Bulletin of the American Mathematical Society}, 21(1):1--46.

\bibitem[Bournez, 1997]{bournez97icalp}
Bournez, O. (1997).
\newblock Some bounds on the computational power of piecewise constant
  derivative systems.
\newblock In {\em 24th International Colloquium on Automata, Languages and
  Programming (ICALP~'97)}, number 1256 in LNCS, pages 143--153.

\bibitem[Cook, 2004]{cook04}
Cook, M. (2004).
\newblock Universality in elementary cellular automata.
\newblock {\em Complex Systems}, 15(1):1--40.

\bibitem[Duchier et~al., 2010]{duchier+durand-lose+senot10isaac}
Duchier, D., Durand-Lose, J., and Senot, M. (2010).
\newblock Fractal parallelism: Solving {SAT} in bounded space and time.
\newblock In Cheong, O., Chwa, K.-Y., and Park, K., editors, {\em 21st
  International Symposium on Algorithms and Computation (ISAAC~'10)}, number
  6506 in LNCS, pages 279--290. Springer.

\bibitem[Duchier et~al., 2012]{duchier+durand-lose+senot12tamc}
Duchier, D., Durand-Lose, J., and Senot, M. (2012).
\newblock Computing in the fractal cloud: modular generic solvers for {SAT} and
  {Q-SAT} variants.
\newblock In Agrawal, M., Cooper, S.~B., and Li, A., editors, {\em 9th
  International Conference on Theory and Applications of Models of Computation
  (TAMC~'12)}, number 7287 in LNCS, pages 435--447. Springer.

\bibitem[Durand-Lose, 2008a]{durand-lose08cie}
Durand-Lose, J. (2008a).
\newblock Abstract geometrical computation with accumulations: {B}eyond the
  {B}lum, {S}hub and {S}male model.
\newblock In Beckmann, A., Dimitracopoulos, C., and L{\"o}we, B., editors, {\em
  Logic and Theory of Algorithms, 4th International Conference on Computability
  in Europe (CiE~'08) (abstracts and extended abstracts of unpublished
  papers)}, pages 107--116. University of Athens.

\bibitem[Durand-Lose, 2008b]{durand-lose08jac}
Durand-Lose, J. (2008b).
\newblock The signal point of view: From cellular automata to signal machines.
\newblock In Durand, B., editor, {\em Journ{\'e}es Automates cellulaires 2008
  ({JAC~'08})}, pages 238--249.

\bibitem[Durand-Lose, 2009a]{durand-lose09nc}
Durand-Lose, J. (2009a).
\newblock Abstract geometrical computation~3: Black holes for classical and
  analog computing.
\newblock {\em Nat. Comput.}, 8(3):455--572.

\bibitem[Durand-Lose, 2009b]{durand-lose09uc}
Durand-Lose, J. (2009b).
\newblock Abstract geometrical computation and computable analysis.
\newblock In Costa, J. and Dershowitz, N., editors, {\em 8th International
  Conference on Unconventional Computation 2009 (UC~'09)}, number 5715 in LNCS,
  pages 158--167. Springer.

\bibitem[Durand-Lose, 2011]{durand-lose11tcs}
Durand-Lose, J. (2011).
\newblock Abstract geometrical computation~4: small {T}uring universal signal
  machines.
\newblock {\em Theoret. Comp. Sci.}, 412:57--67.

\bibitem[Durand-Lose, 2012]{durand-lose12nc-uc}
Durand-Lose, J. (2012).
\newblock Abstract geometrical computation 7\string: Geometrical accumulations
  and computably enumerable real numbers.
\newblock {\em Nat.\ Comput.}, to appear.
\newblock Special issue on Unconv. Comp.~'11.

\bibitem[Durand-Lose, 2013]{durand-lose13cie}
Durand-Lose, J. (2013).
\newblock Irrationality is needed to compute with signal machines with only
  three speeds.
\newblock In Bonizzoni, P., Brattka, V., and L{\"o}we, B., editors, {\em The
  Nature of Computation, 9th International Conference on Computability in
  Europe (CiE~'13)}, LNCS. Springer.
\newblock To appear.

\bibitem[Hardy and Wright, 1960]{hardy+wright60book}
Hardy, G.~H. and Wright, E.~M. (1960).
\newblock {\em {An Introduction to the Theory of Numbers (4$^{th}$ edition)}}.
\newblock Oxford University Press.

\bibitem[Huckenbeck, 1989]{huckenbeck89tcs}
Huckenbeck, U. (1989).
\newblock Euclidian geometry in terms of automata theory.
\newblock {\em Theoret. Comp. Sci.}, 68(1):71--87.

\bibitem[Jacopini and Sontacchi, 1990]{jacopini+sontacchi90}
Jacopini, G. and Sontacchi, G. (1990).
\newblock Reversible parallel computation\string: an evolving space-model.
\newblock {\em Theoret. Comp. Sci.}, 73(1):1--46.

\bibitem[Margenstern, 2000]{margenstern00tcs}
Margenstern, M. (2000).
\newblock Frontier between decidability and undecidability: a survey.
\newblock {\em Theor. Comput. Sci.}, 231(2):217--251.

\bibitem[Mazoyer, 1996]{mazoyer96}
Mazoyer, J. (1996).
\newblock Computations on one-dimensional cellular automata.
\newblock {\em Annals of Mathematics and Artificial Intelligence}, 16:285--309.

\bibitem[Mazoyer and Terrier, 1999]{mazoyer+terrier99}
Mazoyer, J. and Terrier, V. (1999).
\newblock Signals in one-dimensional cellular automata.
\newblock {\em Theoret. Comp. Sci}, 217(1):53--80.

\bibitem[Mycka et~al., 2006]{dacosta06}
Mycka, J., Coelho, F., and Costa, J.~F. (2006).
\newblock {The Euclid Abstract Machine: Trisection of the Angle and the Halting
  Problem}.
\newblock In Calude, C.~S., Dinneen, M.~J., Paun, G., Rozenberg, G., and
  Stepney, S., editors, {\em 5th International Conference on Unconventional
  Computation (UC~'06)}, number 4135 in LNCS, pages 195--206. Springer.

\bibitem[Ollinger and Richard, 2011]{ollinger+richard11tcs}
Ollinger, N. and Richard, G. (2011).
\newblock Four states are enough!
\newblock {\em Theoret. Comp. Sci.}, 412(1-2):22--32.

\bibitem[Rogozhin, 1996]{rogozhin96}
Rogozhin, Y. (1996).
\newblock Small universal turing machines.
\newblock {\em Theoret. Comp. Sci.}, 168(2):215--240.

\bibitem[Weihrauch, 2000]{weihrauch00}
Weihrauch, K. (2000).
\newblock {\em {C}omputable {A}nalysis: an introduction}.
\newblock Springer Verlag.

\bibitem[Woods and Neary, 2007]{neary+woods07mcu}
Woods, D. and Neary, T. (2007).
\newblock Four small universal turing machines.
\newblock In Durand-Lose, J. and Margenstern, M., editors, {\em 5th
  International Conference on Machines, Computations, and Universality
  (MCU~'07)}, number 4664 in LNCS, pages 242--254. Springer.

\end{thebibliography}
\end{document}